\documentclass[twocolumn,english,groupedaddress, superscriptaddress,aps,pre,10pt,floatfix, nofootinbib]{revtex4}

\usepackage[T1]{fontenc}
\usepackage[utf8]{inputenc}
\usepackage{amsmath,amsthm,mathtools,amssymb}
\usepackage{graphicx}
\usepackage[dvipsnames]{xcolor}
\usepackage[normalem]{ulem}
\usepackage{units}
\usepackage{dsfont}
\usepackage{braket}
\usepackage{enumitem}

\newcommand{\nocontentsline}[3]{}
\newcommand{\tocless}[2]{\bgroup\let\addcontentsline=\nocontentsline#1*{#2}\egroup}

\usepackage{orcidlink}

\usepackage{tikz}
\usetikzlibrary{quantikz2}

\usetikzlibrary{math,arrows}

\usepackage{url}

\usepackage{hyperref}

\newtheorem{cor}{Corollary}
\newtheorem{lem}{Lemma}

\newtheorem{thm}{Theorem}
\newtheorem{defi}{Definition}

\theoremstyle{definition} 
\newtheorem{cons}{Mapping}
\theoremstyle{plain} 

\newtheoremstyle{probenv}%
  {0pt}
  {0em}
  {\hangindent=\parindent}
  {}
  {\itshape}
  {{\normalfont:}}
  {.5em}
  {}
  
\theoremstyle{probenv}
\newtheorem*{problem}{Problem}
\theoremstyle{plain} 

\newcommand{\decisionproblem}[4]{
	\begingroup
		\par\noindent\nopagebreak[4]
		\begin{problem}
			\label{#2}\vspace{-0.5em}\colorbox{gray!17!white}{\textsc{#1}}\nopagebreak[4]
		\end{problem}\nopagebreak[4]
		\par\noindent\hangindent=\parindent\textbf{Input}:  #3\nopagebreak[4]
		\par\noindent\hangindent=\parindent\textbf{Question}:  #4
		\par\medskip
	\endgroup
}

\newcommand{\optimizationproblem}[5]{
  \begingroup
	\par\noindent\nopagebreak[4]
	\begin{problem}
		\label{#2}\vspace{-0.5em}\colorbox{gray!17!white}{\textsc{#1}}\nopagebreak[4]
	\end{problem}\nopagebreak[4]
	\par\noindent\hangindent=\parindent\textbf{Input}:  #3\nopagebreak[4]
	\par\noindent\hangindent=\parindent\textbf{Solution}:  #4\nopagebreak[4]
	\par\noindent\hangindent=\parindent\textbf{Objective}:  #5
	\par\medskip
  \endgroup
}

\newcommand{\matr}[1]{{\boldsymbol{#1}}}
\renewcommand{\vec}[1]{{\boldsymbol{#1}}}
\newcommand{\EE}{\mathcal{E}}
\newcommand{\VV}{\mathcal{V}}
\newcommand{\TT}{\mathcal{T}}
\newcommand{\GG}{\mathcal{G}}
\newcommand{\CC}{\mathcal{C}}
\newcommand{\RR}{\mathfrak{R}}

\newcommand{\ii}{\mathrm{i}}
\newcommand{\ee}{\mathrm{e}}

\newcommand{\SU}{{\mathcal{S}U}}
\newcommand{\OPT}{\operatorname{OPT}}

\begin{document}

\title{Towards Quantum Algorithms for the Optimization of Spanning Trees: The Power Distribution Grids Use Case}

\author{Carsten Hartmann\orcidlink{0009-0007-5067-589X}}
\email{c.hartmann@fz-juelich.de}
\affiliation{Forschungszentrum J\"ulich, Institute of Energy and Climate Research -- Energy System Engineering (ICE-1), 52428 J\"ulich, Germany}
\affiliation{RWTH Aachen University, 52056 Aachen, Germany}

\author{Nil Rodellas-Gràcia\orcidlink{0009-0005-9995-7463}}
    \affiliation{Peter Grünberg Institute for Quantum Computing Analytics (PGI-12), Forschungszentrum J\"ulich, 52428 J\"ulich, Germany}

\author{Christian Wallisch\orcidlink{0009-0007-7670-3721}}
    \affiliation{Forschungszentrum J\"ulich, Institute of Energy and Climate Research -- Energy System Engineering (ICE-1), 52428 J\"ulich, Germany}

\author{Thiemo Pesch\orcidlink{0000-0002-3297-6599}}
  \affiliation{Forschungszentrum J\"ulich, Institute of Energy and Climate Research -- Energy System Engineering (ICE-1), 52428 J\"ulich, Germany}

\author{Frank K. Wilhelm\orcidlink{0000-0003-1034-8476}}
    \affiliation{Peter Grünberg Institute for Quantum Computing Analytics (PGI-12), Forschungszentrum J\"ulich, 52428 J\"ulich, Germany}
    \affiliation{Theoretical Physics, Saarland University, 66123 Saarbrücken, Germany}

\author{Dirk Witthaut\orcidlink{0000-0002-3623-5341}}
  \affiliation{Forschungszentrum J\"ulich, Institute of Energy and Climate Research -- Energy System Engineering (ICE-1), 52428 J\"ulich, Germany}
  \affiliation{Institute for Theoretical Physics, University of Cologne, 50937 K\"oln, Germany}

\author{Tobias Stollenwerk\orcidlink{0000-0001-5445-8082}}
    \affiliation{Peter Grünberg Institute for Quantum Computing Analytics (PGI-12), Forschungszentrum J\"ulich, 52428 J\"ulich, Germany}

\author{Andrea Benigni\orcidlink{0000-0002-2475-7003}}
    \email{a.benigni@fz-juelich.de}
  \affiliation{Forschungszentrum J\"ulich, Institute of Energy and Climate Research -- Energy System Engineering (ICE-1), 52428 J\"ulich, Germany}
  \affiliation{RWTH Aachen University, 52056 Aachen, Germany}
  \affiliation{JARA-Energy, Jülich 52425, Germany}

\begin{abstract}
    Optimizing the topology of networks is an important challenge across engineering disciplines. In energy systems, network reconfiguration can substantially reduce losses and costs and thus support the energy transition. Unfortunately, many related optimization problems are NP hard, restricting practical applications. In this article, we address the problem of minimizing losses in radial networks – a problem that routinely arises in distribution grid operation. We show that even the computation of approximate solutions is computationally hard and propose quantum optimization as a promising alternative. We derive two quantum algorithmic primitives based on the Quantum Alternating Operator Ansatz (QAOA) that differ in the sampling of network topologies: a tailored sampling of radial topologies and simple sampling with penalty terms to suppress non-radial topologies. We show how to apply these algorithmic primitives to distribution grid reconfiguration and quantify the necessary quantum resources. 
\end{abstract}   

\maketitle

\tocless\section{Introduction}
Algorithms for network optimization are extensively employed across diverse domains, including communication networks~\cite{cheng1998topological}, transportation planning~\cite{chen2011transport}, and energy systems~\cite{hao2022comprehensive}, to facilitate cost-effective and reliable system design and operation. In many practical applications, networks are required to maintain radial or tree-like topologies for operational reasons, which introduces additional complexity to the underlying optimization problem \cite{magnanti1995optimal, graham2007history}. Within energy systems engineering, such algorithms are of significant importance for the analysis and design of electrical distribution grids.

Distribution grids play a pivotal role in the decarbonization of energy systems \cite{o2022demand, xie2021toward}, as they form the backbone of renewable energy integration \cite{navidi2023coordinating, celli2004meshed}. Through sector coupling, they also support the decarbonization of other domains, including heating and transportation \cite{orths2019sector}. 
Unlike transmission systems, which are generally meshed, radial operation of distribution networks is preferred due to their simplicity, cost efficiency, and ease of protection \cite{celli2004meshed}.
At the same time, reconfiguration switches are incorporated into distribution networks to minimize losses, balance loads, isolate faults, and enhance voltage profiles, all while maintaining a radial configuration. 
Since the Minimal Loss Network Reconfiguration problem was introduced by Merlin and Back in 1975 \cite{merlin1975search}, various network reconfiguration techniques have been studied, demonstrating significant reductions in power losses and improvements in voltage profiles \cite{civanlar1988distribution, baran1989network, abdelaziz2009distribution, jabr2012minimum}. A comprehensive review can be found in \cite{mishra2017comprehensive}.  

Modeling an electrical distribution grid as a weighted graph, a radial configuration corresponds to a spanning tree. Finding optimal spanning trees is a central yet computationally demanding task: while the classical Minimum Spanning Tree (MST) problem—minimizing total edge weight while ignoring network flows and operational constraints—can be solved in near-linear time using the algorithms of Kruskal~\cite{kruskal1956shortest} and Prim~\cite{prim1957shortest}, many practically relevant problems are NP-hard. These include formulations that optimize or restrict vertex degrees \cite{goemans2006minimum,singh2015approximating}, the diameter~\cite{ho1991minimum}, or the number of leaves~\cite{galbiati1994short,fernandes1998minimal} as well as problems such as the Minimum Routing Cost Spanning Tree~\cite{wu2000polynomial} and the Optimum Communication Spanning Tree~\cite{hu1974optimum}.
These hard problems have in common that the cost function depends non-locally on the configuration, that is, the tree. 

The minimal loss network reconfiguration problem, and its related problem, the Minimum Dissipation Spanning Tree (MDST), are also NP-hard ~\cite{khodabakhsh2018submodular,ito2024loss} due to the non-local change in network flows when switching a line. Hence, these problems quickly become intractable for large networks, forcing system operators to rely on heuristic and approximate methods~\cite{aoki1987normal, civanlar1988distribution, shirmohammadi1989reconfiguration} or to limit optimization to local subproblems or precomputed scenarios~\cite{razavi2021multi}. While computationally efficient, these methods can leave the network in suboptimal grid configurations for extended periods. In distribution grids, this typically leads to higher losses, voltage imbalances, or operational constraint violations. All in all, this computational challenge makes dynamic optimization in support of Dynamic Security Assessment infeasible.

Quantum computers may provide an advantage over classical hardware in solving combinatorial problems and thus boost energy system optimization~\cite{morstyn2024opportunities}.
Hardware based on quantum annealing has been commercially available for several years and first industrial use cases have been demonstrated~\cite{yarkoni2022quantum}.
The key idea of quantum annealing is, loosely speaking, that the annealer will not be stuck in local minima as a classical optimizer based on gradient descent.
Unfortunately, today’s quantum annealers are restricted to a very special class of optimization problems that can be mapped to sparse Ising-type problems, i.e. binary variables with a quadratic objective function~\cite{yarkoni2022quantum}.

\begin{figure*}[t!]
    \centering    \includegraphics{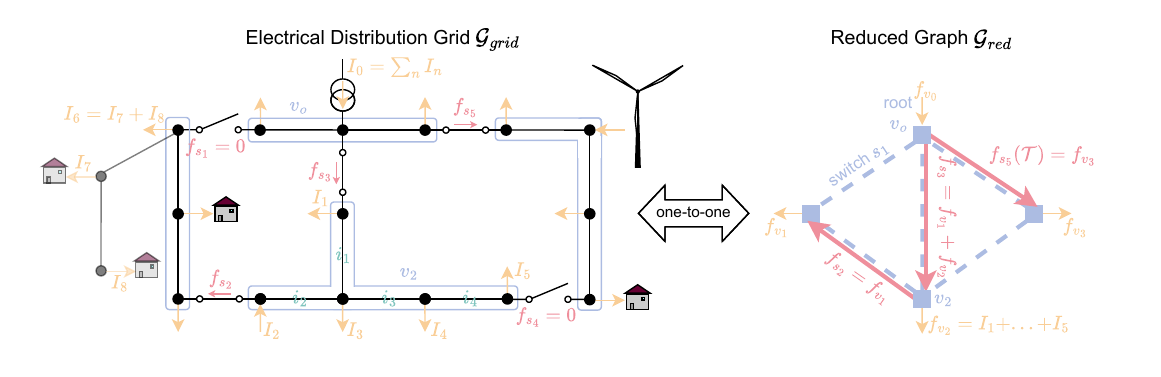}
    \caption{One-to-one correspondence between a feasible configuration of switches in an electrical distribution grid $\GG_{\mathrm{grid}}$ (left) and a spanning tree $\TT$ with root $v_0$ in the reduced graph $\GG_{\mathrm{red}}$ (right), whose edges represent the switches in the distribution grid. Note that buses $7$ and $8$ in $\GG_{\mathrm{grid}}$ can be reduced to bus $6$, since the currents $i_e$ on $e=(6,7)$ and $e=(7,8)$ are nor affected by any reconfiguration.} 
    \label{fig:mapping_distribution_grids}
\end{figure*}

In this article, we propose a heuristic framework to solve optimum spanning tree problems using the Quantum Alternating Operator Ansatz~\cite{hadfield_quantum_2019}.
The fundamental challenge, from a quantum perspective, is to narrow the search space to the set of all spanning trees for a given root node. We develop an algorithm that samples this search space and show how it can be implemented on future quantum hardware.
We discuss potential applications in energy systems, in particular, providing a direct mapping to the distribution grid reconfiguration task. 

In the context of quantum annealing, a formulation as a quadratic unconstrained binary optimization problem (QUBO) for a Minimal Loss Network Reconfiguration variant has been introduced by Silva et al.~\cite{silva_qubo_2023, silva_quantum_2023}. However, the constraints that prevent cycles from being formed are not quadratic and thus require costly polynomial reduction. Moreover, in their formulation, every edge is switchable. Our approach circumvents the costly reduction, and we explicitly provide the construction of the cost function for grids with non-switchable lines. 

We demonstrate the methods in the context of distribution grid reconfiguration; however, it is evident that the proposed approach can be applied to a wide range of other problems.
\\

\tocless\section{Results}

\tocless\subsection{Complexity of Minimum Dissipation Spanning Tree problem}

In this article, we focus on network flow problems, which are particularly important for energy applications. We demonstrate the potential benefits of quantum optimization for the \emph{Minimum Dissipation Spanning Tree} (MDST) problem, which naturally arises in distribution grid operation (Fig.~\ref{fig:mapping_distribution_grids}). 

We start from an undirected connected graph $\GG = (\VV,\EE \subseteq \VV \times \VV)$  with $\lvert \VV \rvert$ nodes, denoted by $n, m, u, v, w \ldots$, and $\lvert \EE \rvert$ edges, denoted by $e, e^\prime, e^{\prime \prime}, \ldots$. For simplicity, we here only consider simple graphs; however, all results can be generalized for multi-graphs, only requiring some additional bookkeeping. 

A \emph{spanning tree} of $\GG$ is a sub-graph $\TT = (\VV, \EE_\TT)$ that contains all nodes, is connected and contains no cycles. In many applications, we have a distinguished root node $n_0$ in the graph, as for instance the feeder in a power distribution grid. The set of all spanning trees with root $n_0$ will be denoted as $ \mathrm{Sp}(\GG, n_0)$. Typically, the number of spanning trees grows exponentially in the system size $\lvert \VV \rvert$, making many spanning tree optimization problems computationally hard.

In flow networks, every node $n \in \VV$ has a fixed in- or outflow $\mathfrak{f}_n$, corresponding to nodal flow demands or injections. In general, we allow multiple sources $\mathfrak{f}_n < 0$ and multiple consumer nodes $\mathfrak{f}_n > 0$.  The flows $f_e$ on the edges $e \in \EE$ are related to the nodal flows $\mathfrak{f}_n$ by Kirchhoff's current laws (KCL). That is, the aggregated flow on the edges connected to a node must equal $\mathfrak{f}_n$ (flow conservation) and consequently, we must have $\sum_n \mathfrak{f}_n=~0$. We now assume that the operating cost due to the dissipation of flow $f_e$ through edge $e$ is given by $c_e = \alpha_e f_e^2$ where $\alpha_e \in \mathbb{R}_{\geq 0}$ is an edge-specific dissipation constant. 

For operational reasons, the network shall be operated as a spanning tree by switching off an appropriate number of edges. To minimize the operational costs, we thus have to solve the MDST optimization problem
\begin{equation}
\label{net_rec:eq:minimal_losses_general}
    \min_{\TT \in \mathrm{Sp}(\GG,n_0)} 
    \sum_{e \in \TT} \alpha_e f_e(\TT)^2,
\end{equation}
where the edge flows $f_e(\TT)$ depend on the topology of the spanning tree $\TT$ via KCL. Optimizing the topology is computationally hard due to the non-local effects on the flows.

To the best of our knowledge, only two papers include contributions regarding the computational hardness of MDST~\cite{khodabakhsh2018submodular,ito2024loss}.
Both papers study the restriction of the problem to distribution networks with only one flow source.
Clearly, the hardness results achieved in this setting transfer to the more general case of multiple flow sources, which, e.g., naturally arise with the introduction of renewable power sources in distribution grids.
In the first of these papers, initially published relatively recently (2017), the authors show that MDST is strongly NP-hard~\cite{khodabakhsh2018submodular}.
Moreover, MDST is NP-hard on lattice graphs even under additional restrictions~\cite{ito2024loss}.

Multiple approximation algorithms for single-source MDST have been proposed~\cite{gupta2022electrical,khodabakhsh2018submodular,ito2024loss}.
However, these approximation algorithms are of limited practical use due to large approximation factors or restrictive assumptions.
For multi-source MDST, an exponential-time algorithm with a guaranteed error bound has been formulated~\cite{inoue2014distribution}.

We present the first approximation hardness result applicable to multi-source MDST (beyond strong NP-hardness).
We say that a minimization problem can be approximated within a factor of $\rho > 1$ if there is an algorithm that, on every possible input, produces a solution that is by at most a factor of $\rho$ more costly than the optimum solution.

\begin{thm}
\label{thm:np_hardness}
    Unless $\mathrm{P} = \mathrm{NP}$, there is a constant $c > 0$ such that MDST cannot be approximated within a factor of $\rho = c\log^2 N$ in polynomial time, where $N$ is the number of nodes. This holds even if integer parameters are polynomially bounded by instance size.
\end{thm}

In particular, the above theorem implies that MDST cannot be approximated within any constant factor in polynomial time. We defer the proof to the Methods.

\tocless\subsection{Network Reconfiguration and MDST}

Network reconfiguration to minimize Ohmic losses in power distribution grids is closely related to the MDST problem. First, network reconfiguration can be reduced to an MDST+ problem by contracting nodes between switches so that all remaining edges are switchable. The resulting cost function mimics MDST \eqref{net_rec:eq:minimal_losses_general}, but is more complex. The cost function involves solving Kirchhoff’s Current Law (KCL) for the contracted subgraphs (see Methods and Fig.~\ref{fig:mapping_distribution_grids}). Second, MDST itself can be seen as a special case of minimal-loss network reconfiguration when every line is switchable. Therefore, existing results on NP-hardness and approximability directly carry over to this important application.

These observations suggest that optimization algorithms developed for MDST, especially those flexible in handling custom cost functions such as for MDST+, can be adapted to tackle network reconfiguration. While classical heuristics for this problem have been extensively studied since the late 1980s~\cite{aoki1987normal,baran1989network,shirmohammadi1989reconfiguration}, giving rise to a wide range of algorithms, including one inspired by the behavior of gut bacteria~\cite{sathishkumar2012power}. Quantum heuristics offer a promising alternative.
\\

\tocless\subsection{Encoding Spanning Trees in a Quantum Register}

\begin{figure*}[t]
    \centering    
    \includegraphics[width=\textwidth]{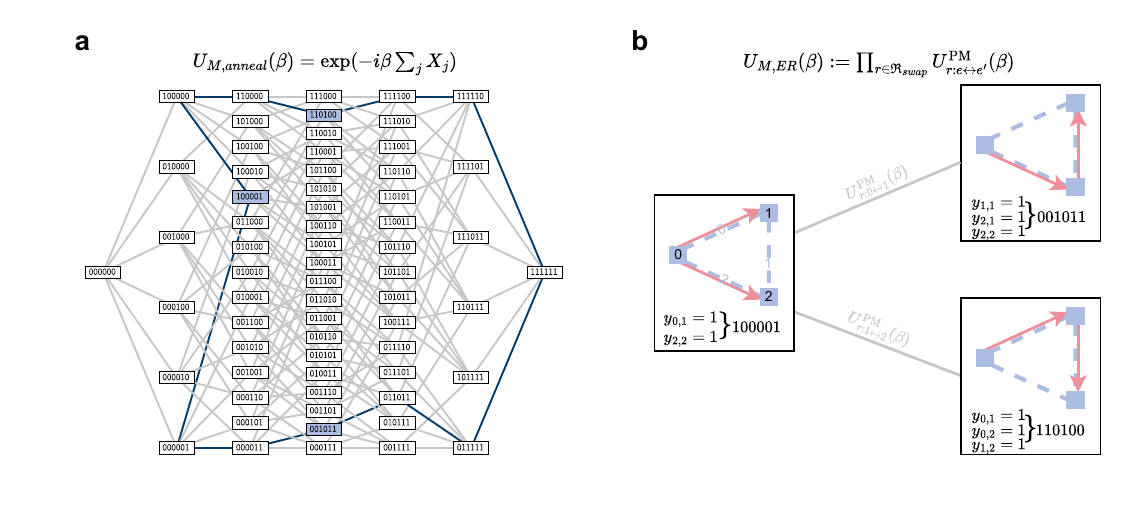}
    \caption{Comparison of the two quantum algorithms for sampling spanning trees for a simple graph with three nodes and three edges. The root is set as $r=0$. 
    \textbf{a}: For the transverse field mixer $U_{\mathrm{TF}}(\beta)$, all 64 configurations can be reached; however, only three of them are feasible (highlighted in light blue). The graph corresponds to the Hamiltonian $H_{\mathrm{TF}} = \sum_j X_j$, the edges correspond to the possible transitions according to the Hamiltonian $H_{\mathrm{TF}}$. Blue edges show (potential) shortest paths between the feasible configurations, 
    \textbf{b}: Partial mixers $U^{\mathrm{PM}}_{r: e \leftrightarrow e^\prime}(\beta)$ implement transition only between two feasible configurations, that is, spanning trees $\TT$. 
    } 
    \label{paper:fig:comparison_ctqw}
\end{figure*}

In this article, we propose a versatile approach to solving optimal spanning tree problems with quantum optimization. We will first formalize the set of problems and introduce a suitable quantum encoding. 

Given an undirected connected graph $\GG = (\VV,\EE \subseteq \VV \times \VV)$, an \emph{orientation} assigns a direction to each edge $e \in \EE$ to keep track of the direction of a flow. For an oriented edge $e = (n,m)$, the node $n$ is called the \emph{tail} and the node $m$ is called the \emph{head} of $e$. The topology and orientation is summarized in the incidence matrix $\matr E \in \mathbb{R}^{\lvert \VV \rvert \times \lvert \EE \rvert}$ as
\begin{equation}
    E_{n,e} = \begin{cases}
        +1 & \text{if } n \text{ is the head of } e,\\
        -1 & \text{if } n \text{ is the tail of } e, \\
        0 & \text{else}. 
    \end{cases}
\end{equation}
For a spanning tree with root node $n_0$, there is a natural orientation where all edges point outwards.

We now introduce an encoding of the optimization variables tailored to flow network problems. An edge is ``active'' if it can contribute to flow transport in a spanning tree and ``inactive'' otherwise. We define the $\lvert \EE \rvert \cdot (\lvert \VV \rvert - 1)$ binary variables
\begin{equation}
    \label{net_rec:eq:def_binaries}
    y_{e,n} = \begin{cases}
        1 & \text{if } n \neq n_0 \text{ is downward of } e \in \TT,  \\
        0 & \text{else.}
    \end{cases}
\end{equation}
and encode them in a quantum state $\ket{y_1} \ket{y_2} \cdots \ket{y_j} \cdots$ by flattening the indices as $j=e \left( \lvert \VV \rvert - 1 \right)  + (n-1)$.
The variables $y_{e,n}$ encode
\begin{enumerate}
    \item whether an edge $e$ is active: For an inactive edge $y_{e,n} = 0$ $\forall n \in \VV \setminus \{n_0\} $. For an active edge, we have that $\sum_{n \in \VV \setminus \{n_0 \}} \lvert E_{n,e} \rvert y_{e,n} = 1$. 
    \item the orientation of an active edge: Given an undirected edge $e=\{n , m\} \in \EE$, we have $y_{e,m} = 1$ if $(n,m) \in \EE_\TT$ and $y_{e,n} = 1$ if $(m,n) \in \EE_\TT$.
\end{enumerate}
Hence, we can directly compute the incidence matrix for a tree $\TT$,
\begin{equation}
     E_{n,e}(\TT) = E_{n,e} ( E_{n,e} y_{e,n} + \sum_{u \in \VV \setminus \{n_0, n\}} E_{u,e} y_{e,u}).
\end{equation}
Furthermore, the global information encoded in the variables $y_{e,n}$ allows to readily calculate network flows, which will be used to encode the objective function as we will discuss below. 

\tocless\subsection{A Primer on Quantum Optimization}
Quantum optimization aims to minimize an objective pseudo-boolean function encoded in a (diagonal) hermitian operator (Hamiltonian) $H_\mathrm{cost}$~\cite{hadfield2021ontherepresentation} with a quantum computing device. 
Typically, these approaches operate with some unitary $U$ (that is derived from the cost Hamiltonian) on a quantum register of size $n$ that, upon measuring, gives bitstring samples of the low-lying solutions.
In addition to the cost Hamiltonian, the near-term approaches adiabatic quantum computation (AQC)~\cite{farhi_quantum_2000} and the Quantum Approximate Optimization Algorithm (QAOA)~\cite{farhi2014quantum}
employ another operator, the co-called mixing operator $H_\mathrm{M}$ that does not commute with the cost Hamiltonian.
This mixer allows us to explore the configuration space by establishing quantum fluctuations, entanglement and tunneling \cite{lanting2014entanglement, denchev2016tunneling}, in analogy to thermal fluctuations in simulated annealing.
In AQC, a unitary operator is applied to the quantum register, which represents the time-evolution of a time-dependent Hamiltonian $H(t)=(1- g(t)) H_\mathrm{cost} + g(t) H_\mathrm{M}$, with $g(0)=0$, $g(1)=1$.
It is given by $U(T_\mathrm{A})=\int_0^{T_\mathrm{A}} \mathrm{d}t \ee^{-\ii t H(t)}$.
Starting in the ground state of the mixer $H_\mathrm{M}$ and for sufficiently large so-called \emph{annealing times} $T_\mathrm{A}$, the system is guaranteed to stay in the ground state due to the adiabatic theorem \cite{born_beweis_1928,albash_adiabatic_2018}.
QAOA is a discretized version of AQC, with a unitary $U=\prod_n \ee^{-\ii \beta_n H_\mathrm{M}} \ee^{-\ii \gamma_n H_\mathrm{cost}}$, where the parameters $\gamma_n, \beta_n$ are either derived from a discretized annealing schedule or freely optimized over.

However, both of these approaches in their original form assume combinatorial optimization problems without constraints.
The standard approach to incorporate constraints is via penalty terms.
Here, one adds terms to the Hamiltonian that penalize infeasible solutions, such that the low-lying eigenstates are all feasible ~\cite{venturelli2015quantum,stollenwerkATM2019}.
For example, one could write
\begin{equation}
\label{eq:penalty_approach}
   H_\mathrm{cost} \to  H_\mathrm{cost} +\lambda_{\mathrm{pen}} H_\mathrm{pen}. 
\end{equation}
This approach has the advantage that the standard mixer $H_\mathrm{M}=-\sum_{i=1}^n X_i$ can be used, whose ground state $\ket{+}^{\otimes n}$ can be prepared efficiently. 
The disadvantages are (i) that we operate on a spectrum that has typically exponentially many more infeasible than feasible states, (ii) the value of the sufficiently large penalty weight is not known a priori, (iii) the enforcement of constraints via penalty terms can introduce higher-order terms that are resource intensive. 

The alternative approach is the invariant feasible subspace approach~\cite{hen2016driver,hadfield_quantum_2019} that starts with an initial feasible state and employs a quantum algorithm that keeps the state in the feasible subspace throughout. This is usually done by constructing advanced mixing operators that map feasible states to feasible states.
In contrast to the penalty-based approach, we exclusively operate on the usually exponentially smaller feasible subspace, which can improve performance greatly. The drawbacks, however, are that complex constraints require resource-intensive mixers and that error correction is needed at least with regard to the feasible subspace~\cite{streif2021errorqaoa}. One of the main contributions of this work is the construction of these advanced mixers.
\\

\tocless\subsection{Sampling Spanning Trees using Penalties}

Following standard annealing procedures \cite{andriyash2016, lanting_experimental_2017}, the mixing unitary to sample all possible qubit configurations is given by the exponential 
\begin{equation}
\label{eq:mixer_transversefield}
    U_{\text{M, penalty}}(\beta) = \exp \left(-\ii \beta \sum \nolimits_j X_j \right),
\end{equation}
where $X_j$ describes the Pauli X operator of the $j$th qubit. This approach is illustrated for an elementary example consisting of three nodes and three edges in Fig.~\ref{paper:fig:comparison_ctqw}a. Only 3 out of 64 possible configurations $y_{e,n}$ are feasible, i.e.~correspond to spanning trees with root $n_0$. Grey lines show transitions due to single-bit flips. At least three bit flips are needed to transfer from one feasible state to another, which affects the respective transition probabilities for $U_{\text{M, penalty}}(\beta)$. 

For larger problems, the share of feasible states is further suppressed as the number of configurations  $2^{(\lvert V \rvert -1) \lvert E \rvert} \gg \lvert \VV \rvert ^{\lvert \VV \rvert - 2} \geq \lvert \mathrm{Sp}(\GG, n_0) \rvert$. Hence, the probabilities for transitions between feasible states are in general suppressed.

The major step in this approach is to formulate equality constraints in the binary variables $y_{e,n}$. These constraints are turned into penalties and added to $H_{pen}$ by squaring the difference between both sides. For spanning trees with root $n_0$ three necessary conditions are that 
\begin{enumerate}
\itemsep-0.1em 
    \item the number edges in $\EE_\TT$ is $\lvert \VV \rvert -1$,
    \item  no cycles are formed,
    \item  all nodes are connected to the root $n_0$. 
\end{enumerate}
Moreover, any two of these three constraints are also sufficient. The following constraints for spanning trees in the binary variables are inspired by necessary conditions 1. and 3.
\begin{align}
    \label{eq:const_edges}
    \sum_e \sum_{n \in \VV \setminus \{n_0\}} \lvert E_{n,e} \rvert y_{e,n} &= \lvert \VV \rvert - 1, \\
    \label{eq:const_KCL}
    \sum_{e \in \EE} \sum_{m \in \VV \setminus \{n_0\}} E_{n,e}(\TT) y_{e,m} & =  1, \quad \forall n \in \VV \setminus \{n_0 \}, \\
\begin{split}
    \label{eq:const_consistency}
    y_{e,n} (1-\lvert E_{n, e} \rvert) &= (1-\lvert E_{n, e} \rvert) \\ \sum_{m \in \VV \setminus \{n_0, n\}} \sum_{e^\prime \in \EE \setminus \{ e\}} &y_{e, m} y_{e^\prime, n} \lvert E_{m,e} \rvert \lvert E_{m,e^\prime} \rvert, \quad \\  &\forall n \in \VV \setminus \{n_0 \}, \forall e \in \EE.
\end{split}
\end{align}
Constraint~\eqref{eq:const_edges} enforces a necessary condition for the number of edges to be $\lvert \VV \rvert -1$. Constraints~\eqref{eq:const_KCL} enforce that every node is connected to the root $n_0$; it can be derived using Kirchhoff's Current Laws (KCLs). 
Constraints~\eqref{eq:const_consistency} establish local consistency between the variables. 
A derivation and discussion of all constraints is provided in the supplementary material.

Both constraints~\eqref{eq:const_KCL}~and~\eqref{eq:const_consistency} are quadratic in the binary variables. Hence, the corresponding penalty terms are quartic and cannot be directly mapped to an Ising Hamiltonian, which is necessary for current quantum annealing hardware. For gate-based implementations of annealing, several approaches have been proposed to address this issue~\cite{campbell2022higherorder}. In general, these approaches substantially increase hardware requirements \cite{dragoi2025approx}.
Alternative constraints ensuring the absence of cycles have been introduced in~\cite{silva_qubo_2023}, but this formulation is also not linear.

\tocless\subsection{Sampling Spanning Trees using the Invariant Feasible Subspace Method}

Instead of sampling all configurations and suppressing the non-radial configurations via penalty terms, we now construct a problem-specific quantum operation that preserves the feasible space spanned by all spanning trees. More concretely, we construct a parameterized unitary $U_{\text{M, feasible}}(\beta)$, such that if our initial state 
encodes a superposition of spanning trees, then the state remains a superposition of spanning trees during the evolution under $U_{\text{M, feasible}}(\beta)$.

We follow Hadfield's approach \cite{hadfield_quantum_2019} and first construct a complete set of local moves that preserve the feasible space $\mathrm{Sp}(\GG, n_0)$, that is, map spanning trees to spanning trees. Let $\TT$ be a spanning tree of $\GG$ with root $n_0$ with the natural orientation implied by the root. We now consider two edges $e=(n,m)\in \TT$ and $e^\prime = (n^\prime, m) \notin \TT$.
We observe that $\TT^\prime = \TT + e^\prime - e$ is another spanning tree with root $n_0$ if and only if the node $n^\prime$ is not downward of edge $e$ in $\TT$, because $\TT^\prime$ would contain a cycle otherwise.
Based on this observation, we define the \emph{edge rotation} $r$ as the local map
\begin{equation}
    r: e = (n,m) \mapsto e^\prime = (n^\prime, m). 
\end{equation}
The edge rotation $r$ is called \emph{valid} if $n^\prime$ is not downward of edge $e$ in $\TT$. 

The following theorem establishes that the set of all edge rotations $\RR$ is complete. Every spanning tree can be \emph{efficiently} reconfigured into any other spanning tree using only these local moves. A derivation and proof can be found in the Methods.
\begin{thm}
    \label{thm:completeness_edge_rotations}
    Let $\TT$ and $\TT^\prime$ be any two spanning trees of $\GG$ with root $n_0$. Let 
    $N_{\TT, \TT^\prime} := \lvert \EE_\TT \setminus \EE_{\TT^\prime} \rvert \leq \lvert \VV \rvert - 1$ be the number of edge mismatches. 
    There exists at least one finite sequence of valid edge rotations $r_{N_{\TT, \TT^\prime}} \circ \ldots \circ r_1$ that maps $\TT$ into $\TT^\prime$. 
\end{thm}

We design a controlled quantum operation that implements valid edge rotations. 
We observe that the two rotations $e \mapsto e^\prime$ and $e^\prime \mapsto e$ are reciprocal: If $e \mapsto e^\prime$ is valid for the spanning tree $\TT$, then $e^\prime \mapsto e$ is valid for the spanning tree $\TT^\prime = \TT+e^\prime-e$. Hence, only one of the two rotations is possible at a time, and validity can be inferred from the binary variables $y_{e,n}$ by evaluating the boolean function $f_{r: e \leftrightarrow e^\prime}=\neg y_{e^\prime, n} \land \neg y_{e, n^\prime}$. We can thus incorporate both edge rotations into one controlled operation.

Based on this classical reasoning, we define a partial controlled edge rotation mixer,
\begin{equation}
    U^\text{PM}_{r: e \leftrightarrow e^\prime}(\beta) := \Lambda_{f_{r: e \leftrightarrow e^\prime}} \big( U_{r: e \leftrightarrow e^\prime}(\beta) \big).
\end{equation}
The notation $\Lambda_{f_{r: e \leftrightarrow e^\prime}}$ specifies that the operation is carried out only if the boolean function $f_{r: e \leftrightarrow e^\prime}$ evaluates True. The operation $ U_{r: e \leftrightarrow e^\prime}(\beta)$ then describes the mixing between the valid configurations $\ket{\bf{y}}$ and $\ket{\bf{y}^\prime}$ encoding the trees $\TT$ and $\TT^\prime$.

We provide the actual design of the partial controlled edge rotation mixer in the supplementary material. Furthermore, we derive the following result regarding the required resources.
\begin{thm}
\label{thm:partial_ mixer}
    The partial mixer $U^\mathrm{PM}_{r: e \leftrightarrow e^\prime}(\beta)$ can be implemented using $\mathcal{O}(\lvert \EE \rvert \lvert \VV \rvert)$ single qubit and CNOT gates. The compiled circuit requires 6 + 2 additional ancillary qubits. The first 6 are required to implement the controlled updating of $\ket{y_{j}}$, and 2 are required for the compilation.  
\end{thm}

The sampling between all feasible configurations is then realized by a full mixer $U_{\text{M, feasible}}(\beta) = \prod_{r \in S} U^\mathrm{PM}_{r: e \leftrightarrow e^\prime}(\beta) $. Depending on the initialization, it is important to use a sequence $S = r_{i_K}, \ldots r_{i_0}$ of edge rotations, such that the whole feasible space can be traversed, that is, we have finite transition probabilities to all possible configurations. Notably, this depends on the initial state. We discuss several approaches in the supplementary material.

\tocless\subsection{Evaluation of both approaches for MDST}

\begin{figure}[t]
    \centering    
    \includegraphics[width=\columnwidth]{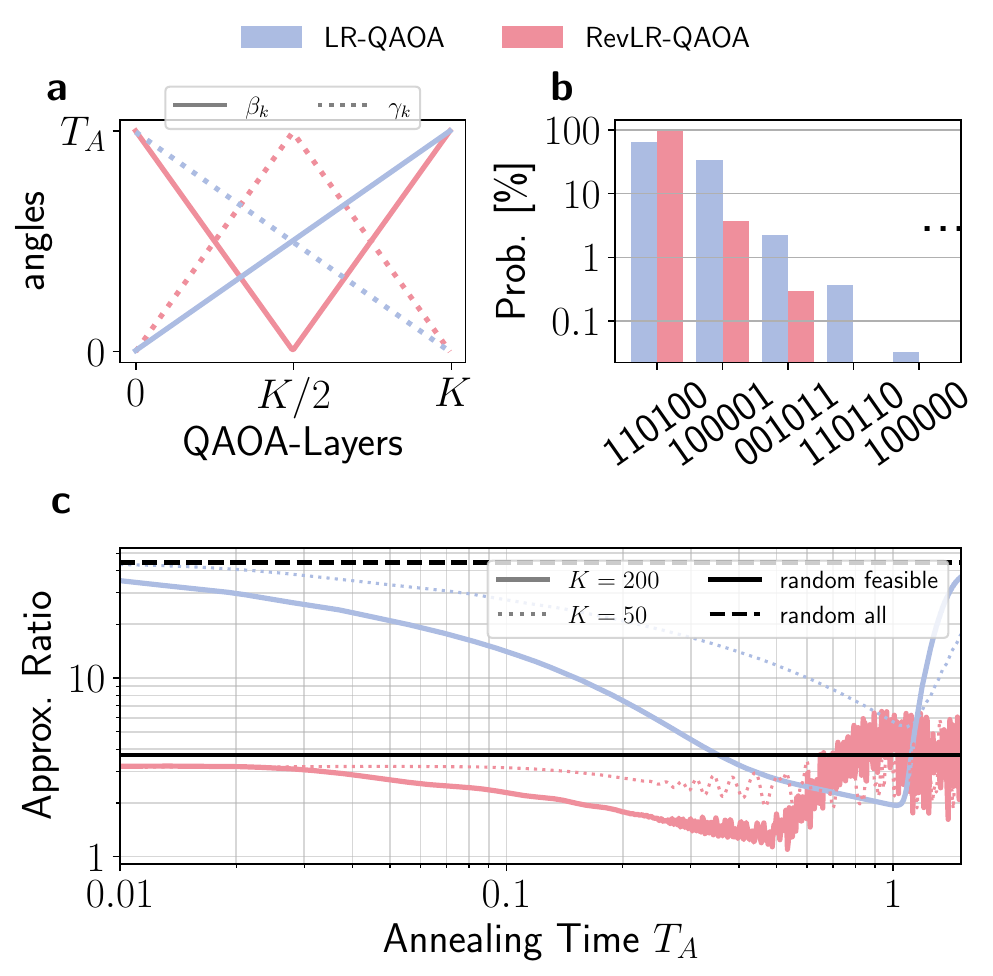}
    \caption{Comparison of the performance of LR-QAOA using the penalty method and RevLR-QAOA employing the invariant feasible subspace approach for a simple MDST instance based on the graph topology shown in Fig.~\ref{paper:fig:comparison_ctqw}. The problem instance is $\alpha_0 = \alpha_1 = 1$ and $\alpha_2=10$; $\mathfrak{f}_0=-3$, $\mathfrak{f}_1=1$ and $\mathfrak{f}_2=2$. The optimal bit string solution is given by $110100$.  \textbf{a}: QAOA schedules $(\gamma_k, \beta_k)$ for LR-QAOA and RevLR-QAOA. The annealing time $T_{\mathrm{A}}$ and the number of layers $K$ are (hyper)-parameters, that define the values of the angles $\beta_k$ and $\gamma_k$ (cf. Methods). \textbf{b}: Final measurement statistics for the best found parameter configurations in a grid search. For LR-QAOA we have $K=200$, $T_{\mathrm{A}}=1$, for RevLR-QAOA $K=200$, $T_{\mathrm{A}} = 0.54$. \textbf{c}: Performance measured by the approximation ratio as a function of the annealing time $T_A$ for fixed $K$. The approximation ratio is defined as the ratio between the energy expectation value of the final state and the ground state (optimal) energy. Lower values of the approximation ratio indicate better performance, with the theoretical lower bound (best achievable value) being $1$. As a benchmark, we compare the performance to picking any feasible state completely at random (vertical black line) and any spin configuration at random (vertical dashed black line).} 
    \label{paper:fig:comparison_qaoa_simulation_results}
\end{figure}

We implement and test both approaches for the MDST problem, in particular for a three-node instance (cf. Fig.~\ref{paper:fig:comparison_ctqw}). The cost function in terms of the binary variables reads
\begin{equation}
   C(\TT) = \sum_{e \in \TT}  \alpha_e f_e(\TT)^2
   = \sum_{e \in \EE} \,\,\, \sum_{n,m \in \VV \backslash {n_0}} \alpha_e y_{e,n} y_{e,m} \mathfrak{f}_n \mathfrak{f}_m.
   \label{eq:total_cost_flows_tree}
\end{equation}
Since the cost function~\eqref{eq:total_cost_flows_tree} is quadratic in the binary variables, it can be mapped to an Ising Hamiltonian and thus be readily implemented for standard quantum optimization techniques such as the Quantum Alternating Operator Ansatz (QAOA) \cite{hadfield_quantum_2019}.

To numerically evaluate the performance of both methods, we simulate two scheduled-QAOA variants tailored to the two approaches: linear ramp QAOA (LR-QAOA) for the penalty and reverse linear ramp QAOA (RevLR-QAOA) for the invariant feasible subspace approach based on the edge-rotation Mixer $U_{\text{M, feasible}}(\beta)$. The corresponding schedules are shown in Fig.~\ref{paper:fig:comparison_qaoa_simulation_results}\textbf{a}. More information on the algorithms and their implementation can be found in the Methods.

We find that for the three-node example, the invariant feasible subspace approach consistently outperforms the penalty method across a wide range of schedule parameters $(T_{\mathrm{A}}, K)$, as shown in panel \textbf{c} and in the supplementary material. While the invariant feasible subspace approach exhibits higher sensitivity to parameter variations at large $T_{\mathrm{A}}$, it achieves a $97.6\%$ probability of sampling the optimal solution at its best setting, with no infeasible states observed. The implementation of the penalty method, by contrast, reaches only a $80.5\%$ success probability under optimal parameters and still samples infeasible configurations with probabilities exceeding $0.1\%$, as seen in panel \textbf{b}. The comparatively higher approximation errors for the penalty methods arise from these infeasible outcomes, particularly at small $T_{\mathrm{A}}$, where the state remains close to a uniform superposition dominated by high-cost, infeasible configurations. Notably, performance does not decrease monotonically for both methods with increasing $T_{\mathrm{A}}$.  Beyond a certain $T_{\mathrm{A}}$ threshold, performance deteriorates because the error, of order 
$\mathcal{O}(T_{\mathrm{A}} / K)$, becomes too large, and the adiabatic evolution is no longer well approximated. For large $T_{\mathrm{A}}$, the invariant feasible subspace approach effectively samples random feasible states, while the penalty method samples random bit strings. 
\\

\tocless\section{Discussion}

The transition toward renewable energy sources fundamentally increases the complexity of power distribution systems. Unlike traditional centralized generation, renewable production is often distributed, with energy injected directly into the grid at multiple points. As a consequence, optimizing power flows and reconfiguring distribution networks with minimal losses has become a central operational challenge. 

The Minimum Loss Network Reconfiguration task in distribution grid operation is closely related to the Minimum Dissipation Spanning Tree (MDST) problem. The key distinction lies in the modeling assumptions: in distribution grids, only a subset of lines are switchable, whereas in MDST, all edges are assumed to be available for switching. In this work, we establish an explicit mapping from network reconfiguration to MDST by constructing an MDST+ cost function on a reduced graph containing only the switchable lines. In light of this formulation, heuristics and optimization strategies developed for MDST can be effectively leveraged to address the network reconfiguration problem. Moreover, this construction establishes that the network reconfiguration problem is at least as computationally hard as MDST.

Like other spanning tree problems whose cost functions depend non-locally on the structure of the tree, MDST is NP-hard \cite{khodabakhsh2018submodular}. Our rigorous results strengthen this understanding by proving strong non-approximability guarantees, indicating that efficient exact or approximation algorithms are unlikely to exist in the general case. Consequently, tackling such computational hardness requires the development and application of alternative strategies, including tailored heuristics and advanced optimization methods \cite{aoki1987normal,civanlar1988distribution,baran1989network,shirmohammadi1989reconfiguration}. 

In this article, we have investigated quantum optimization techniques for the MDST problem. We demonstrate how spanning trees can be efficiently encoded on a quantum register and how they can be sampled. To this end, we compare two approaches. The standard approach incorporates the constraints enforcing a spanning tree directly into the cost function as penalty terms. In contrast, we introduce a set of local moves that enable effective traversal of the search space of all spanning trees. We further show how these local moves can be implemented on a gate-based quantum computer. Numerical simulations on the elementary non-trivial system, using QAOA \cite{hadfield_quantum_2019} and assuming an ideal fault-tolerant quantum computer, indicate that exploring only the feasible space consistently outperforms the penalty-based method across a wide range of hyperparameters. However, the circuits required to implement transitions exclusively between feasible states are significantly deeper than those used to implement transitions between all possible bit strings in the penalty method. Consequently, determining which approach performs better on noisy hardware remains an open question.
\\

\tocless\section{Methods}
\begingroup
\small
\setlength{\baselineskip}{0.9\baselineskip}
Detailed Methods can be found in the supplementary material. 

\paragraph*{Non-approximability of MDST: Theorem \ref{thm:np_hardness}}

We provide a polynomial-time mapping from Minimum Set Cover to MDST that transfers approximation hardness.
Minimum Set Cover is a central problem in the theory of approximation hardness \cite{williamson2011approximation,dinur2014analytical}.
The mapping builds a three-layer network.
The top layer consists of two sources, $y$ and $z$.
The middle layer consists of consumer nodes corresponding to subsets in the Set Cover instance, and the third layer consists of consumer nodes corresponding to elements of the Set Cover instance.
The idea is to design the network such that it is most advantageous to choose a radial configuration that directly connects $y$ to many consumer nodes in the middle layer, satisfying their demands.
Then, source $z$ is left to service the demands of the consumer nodes in the bottom layer.
However, by design, the network does not include any edges directly connecting the top layer with the bottom layer, and hence, flow originating from source $z$ must pass through consumer nodes of the middle layer to reach its destinations in the bottom layer.
Since a radial configuration cannot contain any cycles, source $y$ cannot be directly connected to any of these ``pass-through'' consumer nodes (except for at most one).
Hence, a low-cost radial configuration uses only a few consumer nodes of the middle layer to pass on flow from $z$ to the bottom layer.
This property makes it possible to encode a Set Cover instance: a preferably small number of middle-layer consumer nodes must cover all consumer nodes of the bottom layer.

\paragraph*{Completeness of Local Edge Rotations: Theorem~\ref{thm:completeness_edge_rotations}}

We first observe that for a cyclic graph, there is a sequence of valid edge rotations to map $\TT$ to $\TT^\prime$.
We then prove that there is at least one finite sequence of valid edge rotations that maps $\TT$ to $\TT^\prime$ for arbitrary graphs by induction on the number of edges $\lvert \EE \rvert \ge \lvert \VV \rvert - 1$. For $\lvert \EE \rvert = \lvert \VV \rvert - 1$, the proposition is trivial. For the induction step, we distinguish two cases: (a) If there exists an edge not in either tree, it can be removed and the induction hypothesis applied to the smaller graph. (b) If every edge belongs to at least one tree, we pick an edge $e \in \TT^\prime \setminus \TT$ and construct the fundamental cycle $\CC_e$ in $\TT + e$. Then there exists $\TT^{\prime \prime} = \TT + e - e^\prime$, for any $e^\prime \in \CC_e \setminus \TT^\prime$ and we can reconfigure $\TT \overset{\CC_e}{\mapsto} \TT^{\prime \prime}\overset{\text{(a)}}{\mapsto} \TT^\prime$, using that reconfiguration $\TT \mapsto \TT^{\prime \prime}$ can be achieved by only rotating edges in $\CC_e$, completing the induction. Finally, we prove that there is a sequence of length $N_{\TT, \TT^\prime}$ by induction using the previous results.
\paragraph*{Partial Mixer Implementation and Resource Estimation: Theorem \ref{thm:partial_ mixer}}

A valid edge rotation necessitates the coordinated update of a well-defined subset of variables, specifically those associated with the rotated edges and with the edges along the path between the two tails of the rotation. This update is implemented in two stages by smaller circuits, designed to transform the quantum state prior to the rotation into the corresponding state afterward. Both circuits follow the same principle: they iterate over all variables, mark in an ancilla (via a Boolean function) if the variable is affected by the current rotation, and subsequently apply the required update, controlled by the ancilla. These circuits are embedded within a general mixing circuit, which ensures that instead of overwriting the initial state with the updated one, a quantum rotation of angle $\beta$ is performed between the two.

For the resource estimation, each smaller circuit is decomposed into arbitrary single-qubit gates and CNOTs using standard methods. A systematic count of the number of times each circuit appears then yields the overall resource estimate.

\paragraph*{QAOA Simulation}
QAOA with a fixed schedule is a discretization of unitary in quantum annealing that represents continuous time evolution. 
That is, we discretize the interval $[0, T_\mathrm{A}]$ into $K$ intervals and approximate the evolution as a finite sequence of unitary gates 
\begin{equation*}
   U(0, T_{\mathrm{A}}) = \prod_{k=0}^{K-1} U_\mathrm{M}(\beta_k) e^{-i \gamma_k H_\mathrm{cost}} + \mathcal{O}((T_{\mathrm{A}}/K)^2),
\end{equation*}
where $\beta_k$ and $\gamma_k$ define the schedule and the higher order terms stem from the fact that $H_\mathrm{cost}$ and $U_\mathrm{M}$ do not commute.

For LR-QAOA, the schedule is inspired by standard quantum annealing and given by angles $\beta_k = T_{\mathrm{A}} (1 - k/K), \, \gamma_k = T_{\mathrm{A}} k/K$ (cf.~Fig~\ref{paper:fig:comparison_qaoa_simulation_results}\textbf{a}). Hence, LR-QAOA is suited for the penalty method, so we set $H_\mathrm{C} = H_\mathrm{cost} + \lambda_\mathrm{pen} H_\mathrm{pen}$. We initialize the algorithm in the ground state of the Mixer~\eqref{eq:mixer_transversefield}, which is the uniform superposition over all bit-string configurations, to approximate the ground state of $H_\mathrm{C}$ and thus $H_\mathrm{cost}$. 
Notably, LR-QAOA has been demonstrated to efficiently approximate optimal solutions for a broad class of combinatorial optimization problems \cite{montanez2025toward} and offers good parameter initialization in variational QAOA \cite{sack2021quantum}. 

In contrast, for RevLR-QAOA, the schedule consists of a reverse LR-QAOA schedule for the first $K/2$ Layers followed by a (forward) LR-QAOA schedule for the second half (cf.~Fig~\ref{paper:fig:comparison_qaoa_simulation_results}\textbf{a}). Consequently, in the view of quantum annealing, the initial Hamiltonian is $H_\mathrm{cost}$ and such a schedule was initially suggested to locally refine solutions that have been found using another method \cite{perdomo2011study}. More importantly, RevLR-QAOA allows us to search for the ground state of $H_{cost}$ even though the ground state of the Mixer $U_\mathrm{M,\,ER}$ is not known, and is thus suited for the invariant feasible subspace approach. Particularly, we set 
\begin{equation*}
    H_\mathrm{cost} \to \begin{cases} H_\mathrm{cost,\,init} \quad & \text{for} \, k<K/2, \\  H_\mathrm{cost} \, & \text{else} \end{cases}, 
\end{equation*}
where $H_\mathrm{cost, init}$ is the cost function for another initial problem instance based on the same underlying graph $\GG$, but with other $(\alpha_e^\prime, \, \mathfrak{f}_n^\prime)$, whose optimal solution is known, e.g., by some classical brute force approach. 

In practice, we model MDST instances using Pyomo~\cite{hart2011pyomo, bynum2021pyomo}. For LR-QAOA, the full model (cost and constraints) is then converted into a PUBO using quboify \cite{quobify2025}, which provides automatic $\lambda_\mathrm{pen}$-selection based on a naive upper bound for the cost function. For RevLR-QAOA, only the cost function is converted. Both scheduled-QAOA variants are then simulated in Qiskit using the statevector method, that is, parameterized circuits for the mixer and $H_\mathrm{cost}$ are applied consecutively according to the schedule. The performance of both QAOA variants depends heavily on the hyperparameters $T_\mathrm{A}, \, K$. Thus, we perform a grid search for $K \in \{10, 50 ,100 , 200\}$ and 1000 values for $T_\mathrm{A}$ loguniformly seperated in $[0.01, 1.5]$. 

\paragraph*{MDST+ cost function for Network Reconfiguration}

The essential difference between MDST and distribution grid reconfiguration is that MDST treats all lines as switchable, whereas distribution grids have only a few switches.  A valid configuration of the switches, such that the distribution grid is operated radially, corresponds to the spanning tree of the grid graph $\GG_\mathrm{grid}$, but not every spanning tree represents a feasible operational state.

There is, however, a one-to-one correspondence between valid switch configurations and spanning trees $\TT_{\mathrm{red}}$ of the reduced graph $\GG_{\mathrm{red}}$. The reduced graph is obtained by contracting all nodes $n \in \VV_{grid}$ between switches into single nodes $v \in \VV_{\mathrm{red}}$. We define the flow injections $\mathfrak{f}_v$ as the sum of electrical current injections $I_n$ of all nodes $n$ contracted in $v$, that is, $\mathfrak{f}_v = \sum_{\{n \in \EE_{\mathrm{grid}} | n \in v\}} I_n$. Each edge $s \in \EE_{\mathrm{red}}$ corresponds to a switch in the grid (see Fig.~\ref{fig:mapping_distribution_grids}). For a given $\TT_{\mathrm{red}}$, the switch flows $f_s(\TT_{\mathrm{red}})$ equal the electrical currents on the switches, and by KCL they uniquely determine the line currents $i_e(\TT_{\mathrm{red}})$ for all $e \in \EE_{\mathrm{grid}}$. By solving the resulting system of equations once, we can express the line currents $i_e(\TT_{\mathrm{red}})$ in the binary variables $y_{s,v}$ encoding trees in $\GG_{\mathrm{red}}$. However, these expressions are quadratic, since they involve terms like $E(\TT_{\mathrm{red}})_{s, n(v)}\, f_s(\TT_{\mathrm{red}})$.

Overall, this correspondence allows us to cast reconfiguration as an MDST+ problem on $\GG_{\mathrm{red}}$ with a cost function $\sum_{e \in \EE_{\mathrm{grid}}} R_e\, i_e(\TT_{\mathrm{red}})^2$ that is quartic in $y_{s,v}$.

\endgroup

\tocless\section{Acknowledgements}
This paper was written as part of the project ``Quantum-based Energy Grids (QuGrids)", which is receiving funding from the programme ``Profilbildung 2022", an initiative of the Ministry of Culture and Science of the State of North Rhine-Westphalia. NRG and TS were funded by the German Federal Ministry of Research, Technology and
Space (BMFTR) in the project Quantum Artificial Intelligence for the Automotive Value Chain (QAIAC),
Funding No. 13N17166. The sole responsibility for the content of this publication lies with the authors.

\clearpage
\begin{widetext}

\makeatletter
\hsize=\textwidth
\columnwidth=\textwidth

\tocless\section{Supplementary Notes}
\tableofcontents

\begin{table}[b]
\caption{
List of symbols and variables. Vectors are written as boldface lowercase roman letters, matrices as boldface uppercase roman letters, while Gothic type letters denote sets and graphs.
}
\label{tab:notation}
\begin{tabular}{p{3cm} p{10cm}}
    \hline
    $e, e^{\prime},\ldots$ & (directed) edges \\
    $n, m, u , v, \ldots$ & nodes \\
    $\GG$ & an undirected graph \\
    $\VV$ & set of nodes \\
    $\EE$ & set of edges \\
    $\TT$ & spanning tree of $\GG$ \\
    $\mathrm{Sp}(\GG ,n_0)$ & set of all spanning trees with root $n_0$\\
    $y_{e,n}$ & binary variable, $1$ if node $n$ \\
             &is downward of edge $e$, else $0$. \\
    $\matr{E}$ & the node edge incidence matrix, \\
                &$+1$ if edge $e$ points to node $n$ \\
    $\matr{E}(\TT)$ & the node edge incidence matrix for tree $\TT$\\
    $\alpha_e$ & dissipation constant at edge $e$ \\
    $\mathfrak{f}_n$ & in/out flow at node $n$ \\
    $f_e(\TT)$ & flow through edge $e$ for tree $\TT$ \\
    $r: e \mapsto e^\prime$ & local edge rotation \\
    $U^{\text{PM}}_{r: e \leftrightarrow e^\prime}(\beta)$ & partial edge swap mixer\\
    $X, \, Y, \, Z $ & Pauli gates \\
    $\ket{\bf{y}}$ & quantum state encoding configuration of $y_{e,n}$ \\
    \hline
\end{tabular}
\end{table}

\clearpage
\section{Mathematical Background: Graph Theory}

This section provides an overview of the mathematical background. In \ref{sec:graph_background} we introduce some useful notation and review some important results from the literature on (oriented) spanning trees. In \ref{sec:flow_networks} we review how power grids, such as distribution grids, can be naturally described as flow networks and provide insights into how flows can be easily computed for radial (sub-)graphs, such as spanning trees. 

A general introduction to graph theory, including important results on flow networks, can be found in the textbooks by Bollobás \cite{Bollobas1998Modern} and Newman \cite{Newman2010networks}.

\subsection{Orientation, Spanning Trees and Cycles on Graphs}
\label{sec:graph_background}

Let $\GG = (\VV,\EE \subseteq \VV \times \VV)$ be an undirected graph  with $\lvert \VV \rvert$ nodes, denoted by $n, m, u, v, w \ldots$, and $\lvert \EE \rvert$ edges, denoted by $e, e^\prime, e^{\prime \prime}, \ldots$. Throughout this article, we assume that the network is connected. If not stated otherwise, we discuss simple graphs. However, all results can be generalized to/also hold for multi-graphs, that is, graphs with multiple edges between two nodes. 

An \emph{orientation} of $\GG$ assigns a direction to each edge $e \in \EE$ by turning the edge $e = \{n, m\}$ into a directed edge. Hence, for each edge, there are two choices: $e = (n,m)$ and $e = (m,n)$. For an oriented edge $e = (n,m)$ the node $n$ is called the \emph{tail} and the node $m$ is called the \emph{head} of $e$. 
Once an orientation has been fixed, we can define the edge incidence matrix $\matr E \in \mathbb{R}^{\lvert \VV \rvert \times \lvert \EE \rvert}$ as
\begin{equation}
    E_{n,e} = \begin{cases}
        +1 & \text{if } n \text{ is the head of } e,\\
        -1 & \text{if } n \text{ is the tail of } e, \\
        0  & \text{else}. 
    \end{cases}
\end{equation}
The (unweighted) Laplacian is then defined as $\matr L = \matr E^\top \matr E$ and is independent of the chosen orientation. In flow networks, one often has a distinguished node as for instance the feeder or slack node. Removing the row and column corresponding to this node, one obtains the grounded Laplacian $\matr{\tilde L}$.

A \emph{tree} $\TT = (\VV_\TT, \EE_\TT )$ in $\GG$ is a sub-graph of $\GG$ that has no cycles and is connected, cf. Fig.~\ref{net_rec:fig:cycle_basis_trees} for an elementary example. A \emph{spanning tree} of $\GG$ is a tree, such that all nodes $n \in \VV$ are connected, hence, $\VV_\TT = \VV$. 
Thus, a necessary condition for a spanning tree is that 
\begin{equation}
    \lvert \EE_\TT \rvert = \lvert \VV \rvert - 1.
\end{equation} 
The set of spanning trees will be denoted as $\mathrm{Sp}(\GG)$ in the following. The number of spanning trees of a graph $\GG$ depends on the density of the graph. For a grid with only one cycle, the number of spanning trees is given by the number of nodes in the cycle, whereas for a complete graph, the number of spanning trees is given by $\lvert \VV \rvert ^{\lvert \VV \rvert - 2}$ and thus grows exponentially in the system size. According to Kirchhoff's theorem, the number of spanning trees equals the determinant of the grounded Laplacian $\matr{\tilde L}$ \cite{kirchhoff1958theorem}.

Let $\TT \in \mathrm{Sp}(\GG)$, then any edge $e \notin \EE_\TT$ defines a \emph{fundamental cycle} in $\GG$. To see this, let $p$ be the path from the head of $e$ to the tail of $e$ in $\TT$. There exists exactly one such path, since otherwise there would be a cycle in $\TT$. The path, together with the edge $e$, forms a cycle in $\GG$. Hence, a graph $\GG$ has $\lvert \EE \setminus \EE_\TT \rvert = \lvert \EE \lvert - \lvert \VV \rvert + 1$ fundamental cycles, and the set of all fundamental cycles forms a cycle basis of $\GG$. However, this mapping is not one-to-one. Different spanning trees (can) define the same fundamental cycle basis, see Fig.~\ref{net_rec:fig:cycle_basis_trees}.

\begin{figure}[t!]
    \centering
    \includegraphics[trim={0 0 1.2cm 0},clip]{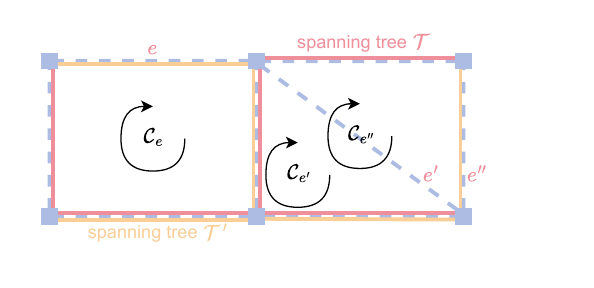}
    \caption{Two spanning trees and fundamental cycle basis for an elementary six node graph $\GG$ (light blue). A spanning tree $\TT$ is a subgraph of $\GG$ such that all nodes are connected and there are no cycles. The edges $e, e^\prime$ and $e^{\prime\prime}$ not in $\TT$ each define a fundamental cycle, and all fundamental cycles together define a cycle basis of $\GG$. The two drawn spanning trees $\TT$ (red) and $\TT^\prime$ (yellow) define the same fundamental cycle basis $\{\CC_e, \CC_{e^\prime}, \CC_{e^{\prime \prime}} \}$.}  
    \label{net_rec:fig:cycle_basis_trees}
\end{figure}

\subsection{Flow Networks}
\label{sec:flow_networks}

\subsubsection*{Graph Representation of Flow Networks}

A natural representation of flow networks is in the form of graphs. Let $\GG = (\VV, \EE)$ be an undirected graph representing a flow network's topology. Let $\mathfrak{f}_n$ denote the supply/demand of flow at node $n$. If $\mathfrak{f}_n > 0$, we define that node $n$ demands $\lvert \mathfrak{f}_n \rvert$ quantities of the flow, e.g. real power in transmission grids or electrical current in distribution grids. Vice versa, the flow $\lvert \mathfrak{f}_n \rvert$ is injected into the network at $n$ if $\mathfrak{f}_n < 0$. We assume that the feeder node always injects/demands flow to the network such that the network is balanced, $\sum_n \mathfrak{f}_n = 0$. 

To describe the direction of a flow along an edge $f_e$, we need to fix an orientation of the graph. For an edge $e=(n,m)$, a positive flow value $f_e > 0$ indicate a flow from the tail $n$ to the head of $m$ and vice versa. Then the flow supplies/demands at each node $\mathfrak{f}_n$ are related to the flows $f_{e}$ on the lines by Kirchhoff's current law (KCL), which states that for each node $n$ the sum over all in and outflows must be equal to $\mathfrak{f}_n$. Algebraically, this conservation law can be written as
\begin{equation}
    \label{net_rec:eq:KCL}
    \mathfrak{f}_n = \sum_e E_{n,e} f_e ,
\end{equation}
where $\matr E$ is the edge-incidence matrix of $\GG$.

\subsubsection*{Application: Representation of Power Grids}

\begin{figure}[b!]
    \centering
    \includegraphics[trim={0 0.4cm 0 0.4cm},clip]{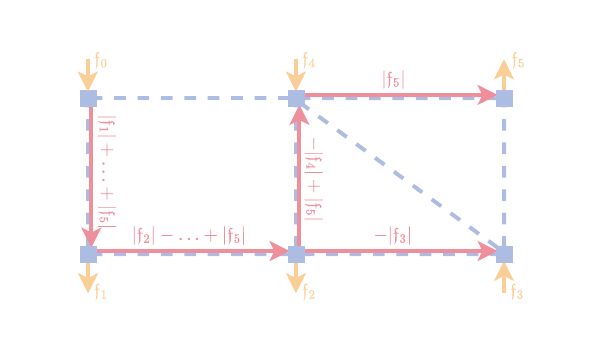}
    \caption{The flows on an edge $e \in \TT$ can be readily computed by summing over all downward demands/supplies, respecting the signs. For illustration, we make the signs that indicate the direction of the flow on an edge explicit. }  
    \label{fig:flows_from_downward}
\end{figure}

Power grids can naturally be represented as flow networks; the flows can describe the physical electrical currents or derived quantities such as (real-) power flows. 

In high-voltage transmission grids, nodes $n \in \VV$ (or buses) represent substations or individual busbars in a substation. Edges $e \in \EE$ (or branches) represent transmission lines or transformers connecting the nodes.

In low-voltage distribution grids, nodes correspond to individual consumers or distributed generators such as photovoltaic panels. The connection to the higher voltage level, the feeder, is represented by a distinguished node $n_0$
The edges represent all possible lines (or cables, or transformers) of the distribution grid. For the major part of this work, we assume that \emph{every} line can be \emph{active}, that is, flow can go through, or \emph{inactive}, that is, no flow is possible along the line. In real-world applications, this assumption is not justified since the network might contain ``non-switchable'' edges that are always active. We discuss the implications on modelling of distribution grids in \ref{sec:mapping_distribution_grids}. 

Under normal operations, distribution grids are operated in a radial configuration to prevent fault propagation. Different radial configurations can be realized by opening or closing switches to isolate faults or optimize costs. Formally, the subgraph $\TT = (\VV, \EE_{\TT})$ given by all active edges in $\EE$ is a spanning tree of $\GG$. Setting the root of $\TT$ as the feeder node $n_0$, there is a unique orientation of $\TT$ where the head of each edge $e \in \EE_\TT$ is pointing downwards from the root. Hence, $\TT$ together with the feeder $n_0$ can be viewed as a directed/oriented spanning tree. We denote the set of all (directed) spanning trees of $\GG$ with root $n_0$ by $ \mathrm{Sp}(\GG, n_0)$.

\subsubsection*{Flows on Spanning Trees}

For a spanning tree $\TT$ of $\GG$ with root $n_0$, an orientation is naturally induced by choosing the orientation of each edge such that the head points downwards, that is, away from the root. 
For this induced orientation, the flow $f_e$ on each edge $e \in \EE_\TT$ can be computed directly from the flow demands and supplies of all nodes $n$ downward of the edge $e$, see Fig.~\ref{fig:flows_from_downward} for a simple example. We thus have
\begin{equation}  \label{net_rec:eq:flows_in_tree}
    f_e(\TT) = \sum_{\substack{n \text{ downward of } e \\ \text{in } \TT}} \mathfrak{f}_n. 
\end{equation}
If $f_e(\TT) > 0$, the downward demands exceed the downward supplies and thus the flow is downward along edge $e$. Vice versa, if $f_e(\TT) < 0$, the flow direction is upwards, to the feeder node. We remark that we use the notation $f_e(T)$ to indicate that the flow on the edge $e$ depends on the topology of the spanning tree $\TT$.

\clearpage
\section{Local Edge Rotations}

This section provides a more thorough introduction to the local edge rotations - constructed to reconfigure spanning trees efficiently. In particular, in \ref{sec:reconfiguration_of_spanning_trees} we rigorously define edge rotation and prove several Lemmas that together provide a complete proof of Theorem 2 in the main paper. Finally, in \ref{sec:edge_rot_graph_Def} we adopt a different viewpoint, representing the configuration space $ \mathrm{Sp}(\GG, n_0)$ as a graph where two spanning trees are adjacent if they can be reconfigured into each other by a single edge rotation.

\subsection{Reconfiguration of Spanning Trees by Local Moves}
\label{sec:reconfiguration_of_spanning_trees}

In this section, we prove that spanning trees can be sampled by sequences of local operations. We start by formally defining the \emph{edge rotations}. We show that under certain conditions, these edge rotations map from one spanning tree to another. Each local move can thus be viewed as a \emph{local reconfiguration}. We show that using these local reconfigurations, we can efficiently explore the search space $ \mathrm{Sp}(\GG, n_0)$. We note that similar results have recently been obtained in Ref.~\cite{behrooznia2024listing}.

Let $\GG$ be an undirected graph. We denote subgraphs of $\GG$ by $\GG^\prime, \, \GG^{\prime \prime}, \, \ldots$. We then choose an orientation for each subgraph $\GG^\prime$, turning $\GG^\prime$ into a directed graph. Then, any node $u \in \VV$ is called downward of an oriented edge $e = (n,m)$ in $\GG^\prime$ if there is a directed path from $m$ to $u$. 
\begin{defi}
A local \textbf{edge rotation} is a single local edge reconfiguration such that the head of the ``rotated'' edge remains at the same node, that is 
\begin{align}
\begin{split}
    r: \EE^\prime & \to \EE^{\prime \prime}, \\
    \{ e_1, ..., e_i^\prime = (u,v), .... , e_m \} & \mapsto \{e_1, ..., e_i^{\prime \prime} = (w, v), .... , e_m\}
\end{split}
\end{align}
To simplify notation we also write
\begin{equation*}
    r: \GG^\prime \to \GG^{\prime \prime}, \, e_i^\prime \mapsto e_i^{\prime \prime}. 
\end{equation*}
\end{defi}

We now only consider subgraphs $\TT, \, \TT^\prime, \, \ldots$ that are spanning trees of $\GG$. For spanning trees with root $n_0$ an orientation is naturally implied. By construction, edge rotations preserve tree structures, i.e.~they map spanning trees to spanning trees, if and only if no cycle is formed. Connectedness then follows from the fact that the number of edges is preserved. Hence, we get the following elementary result. 

\begin{lem}
\label{lem:valid_edge_rotations}
    Let $\TT = (\VV, \EE_\TT \subset \EE)$ be a spanning tree of $\GG$ with root $n_0$. Let $e = (n,m)\in \EE_\TT$ and $e^\prime = (n^\prime, n) \notin \EE_\TT$, then the subgraph $(\VV, \EE_\TT \cup \{e^\prime\} \setminus \{e\})$ obtained by a single edge rotation $r:e=(n,m) \mapsto e^\prime = (n^\prime, m)$ is also a spanning tree with root $n_0$, if and only if the node $n^\prime$ is not downward of the edge $e$ in $\TT$. Such an edge rotation is called \emph{valid}. 
\end{lem}

A sequence of valid edge rotations is equivalent to a general spanning tree reconfiguration. A simple example is given in Fig.~\ref{net_rec:fig:reconfiguration_example}\textbf{a}. 
The following lemmas establish completeness. We show that, starting from an arbitrary spanning tree $\TT$ with root $n_0$, we can move to any other spanning tree with root $n_0$ using only valid local edge rotations. We start with the elementary case of a graph that contains a single cycle only and then generalize the result to arbitrary connected graphs. In this case, there is only one unique (without repetitions) finite sequence which can be trivially constructed, see Fig.~\ref{net_rec:fig:reconfiguration_example}\textbf{b}. Hence, we get the following result. 

\begin{lem}
\label{net_rec:lem:spanning_tree_reconf_cycle}
Let $\GG$ be a graph consisting of a single cycle. Starting from any spanning tree configuration $\TT$ with root $n_0$ any other spanning tree configuration $\TT^\prime$ with root $n_0$ can be realized by a finite sequence of valid edge rotations. 
\end{lem}

\begin{figure}[t!]
    \centering
    \includegraphics{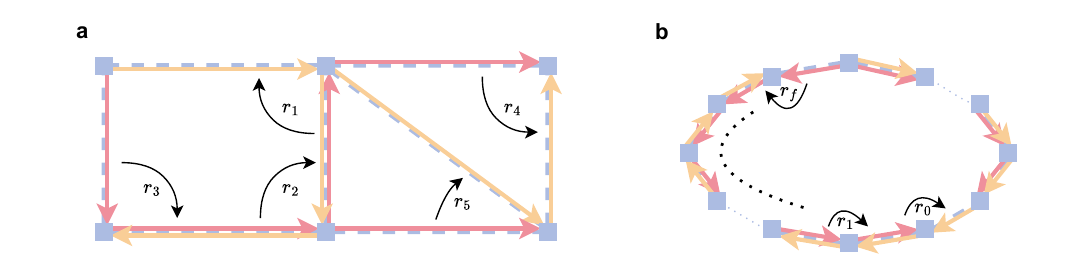}
    \caption{Using a sequence of valid local edge rotations $r_i$ we can reconfigure the spanning tree $\TT$ (red) into the spanning tree $\TT^\prime$ (yellow). \textbf{a} The small 6-node example demonstrates that for non-cyclic graphs the sequence to reconfigure $\TT$ into $\TT^{\prime}$ is not unique: the sequences $r_1, r_2, r_3, r_4, r_5$ and $r_4, r_5, r_1, r_2, r_3$ are both valid. \textbf{b}. For a cycle graph of arbitrary (finite) size, there is one unique (without repetitions) and finite sequence $r_0, r_1, \ldots, r_f$ that can be trivially constructed.}  
    \label{net_rec:fig:reconfiguration_example}
\end{figure}

Using the previous result, we can now prove the general case. 

\begin{lem}
\label{net_rec:thm:spanning_tree_reconf}
Let $\GG$ be a connected graph. Starting from any spanning tree configuration $\TT$ with root $n_0$, any other spanning tree configuration $\TT^\prime$ with root $n_0$ can be realized by a finite sequence of valid edge rotations. 
\end{lem}

\begin{proof}
We prove the Lemma using induction on the number of edges $\lvert \EE \rvert$. For a connected graph, we have that $\lvert \EE \rvert \geq \lvert \VV \rvert -1$. \\ \\
\textit{Base case:} For $\lvert \EE \rvert = \lvert \VV \rvert -1$, there is only one spanning tree configuration, which is given by $\GG$ itself. Hence, the statement is trivial. \\ \\
\textit{Induction Step:} For any $\lvert \EE \rvert >  \lvert \VV \rvert -1$ we distinguish three cases.  
\begin{enumerate}
    \item[(a):] If $\TT = \TT^\prime$, the statement is trivial.
    \item[(b):] Assume there exists an edge $e \in \GG$ such that $e \notin \TT$ and $e \notin \TT^\prime$. Then $\TT$ and $\TT^\prime$ are also spanning trees for the graph $\GG \setminus \{e\}$, the graph with edge $e$ removed. By induction hypothesis, $\TT$ can be reconfigured to $\TT^\prime$ using a finite sequence of valid edge rotations in $\GG \setminus \{e\}$. From the reduction, it follows that $\TT$ can also be reconfigured to $ \TT^\prime$ in $\GG$.
    \item[(c):] If there is no edge $e \in \EE$ such that $e \notin \TT$ and $e \notin \TT^\prime$, then any edge $e \in \GG$ is in at least one of the two trees. Without loss of generality, we take any edge  
        \begin{equation*}
            e \in \GG:  e \notin \TT \text{ and } e \in \TT^\prime.
        \end{equation*}
        The edge $e$ together with $\TT$ defines a fundamental cycle $\CC_e \subset \GG$. Then,
        \begin{equation*}
            \exists \, e^\prime \in \CC_e: \, e^\prime \neq e \text{ and } e^\prime \notin \TT^\prime, 
        \end{equation*}
        since otherwise we would have that $\CC_e \subset \TT^\prime$ which is in contradiction to $\TT^\prime$ being a tree. 
       
        Now we define a third spanning tree configuration $\TT^{\prime \prime}$ as 
        \begin{equation*}
            \TT^{\prime \prime} := \TT \setminus \{e^\prime\} \cup \{e\}.
        \end{equation*}
        $\TT^{\prime \prime}$ is indeed a spanning tree, since the edges $e$ and $e^\prime$ are both in the fundamental cycle $\CC_e$ and thus $\TT^{\prime \prime}$ remains connected. 
        
        Then we have that:
        \begin{enumerate}
            \item[1.] By construction, $\TT^{\prime \prime}$ can be reconfigured into $\TT^\prime$ using the results from case(b) because $e^\prime \notin \TT^\prime$ and $e^\prime \notin \TT^{\prime \prime}$.
            \item[2.] 
            For the trees $\TT$ and $\TT^{\prime \prime}$ restricted to the cycle $ \CC_{e}$, that is, for $\TT \cap \CC_{e}$ and $\TT^{\prime \prime} \cap \CC_{e}$ with root $n_0(\CC_e)$ induced by the direct path from the root to any node in the cycle $\CC_e$, reconfiguration from one to the other using a finite sequence of valid edge rotations is possible according to Lemma  \ref{net_rec:lem:spanning_tree_reconf_cycle}. 
            Hence, using the same sequence of valid edge rotations, we can reconfigure from $\TT$ to $\TT^{\prime \prime},$ leaving all other edges outside of the cycle $\CC_e$ untouched. 
        \end{enumerate}
        Combining these two steps, we have shown that we can reconfigure from $\TT$ to $\TT^\prime$ via $\TT^{\prime \prime}$.
    \end{enumerate}
\end{proof}

Given two spanning tree configurations, we are interested in how many valid edge rotations are needed to reconfigure one into the other. We thus define the \emph{number of edge-mismatches} as 
\begin{align}
\begin{split}
        N_{\TT, \TT^\prime} := \lvert \EE_\TT \setminus \EE_{\TT^\prime} \rvert = \lvert \VV \rvert - 1 - \lvert \EE_\TT \cap \EE_{\TT^\prime} \rvert.
\end{split}
\end{align}
Note that two edges $e = (n,m)$ and $e^\prime=(m,n)$ are not the same, and thus if $e \in \TT$ and $e^\prime \in \TT^\prime$ this would count as a mismatch. The following corollary establishes that the number of valid edge rotations needed to reconfigure between two trees is given by the number of edge mismatches $ N_{\TT, \TT^\prime}$. 

\begin{figure}
    \centering
    \includegraphics[width=0.7\textwidth]{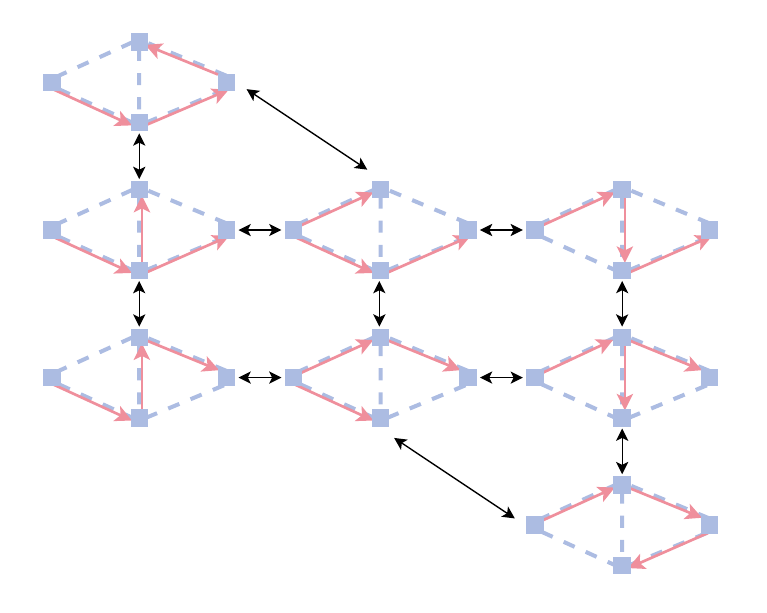}
    \caption{Graph $\GG_{\mathrm{Sp}}$ of all spanning trees for a simple four-node example. Two Spanning Trees are considered adjacent if they can be reconfigured into each other using a single edge rotation, as indicated by the black arrows. }
    \label{fig:graph_of_spanning_trees}
\end{figure}

\begin{cor} 
    \label{cor:number_of_rotations}
    Given two spanning tree configurations $\TT$ and $\TT^\prime$ of $\GG$ with root $n_0$. Then, the number of valid edge rotations needed to reconfigure from one to the other equals the number of edge mismatches $N_{\TT, \TT^\prime}$.
\end{cor}

\begin{proof}
We prove the Corollary using induction on the number of mismatches.\\ \\
\textit{Base Case:} Let $\TT$ and $\TT^\prime$ have one edge mismatch, that is, there exists exactly one edge $e = (m,u) \in \TT$ such that $e \notin \TT^\prime$ and exactly one edge $e^\prime = (n,v) \in \TT^\prime$ such that $e^\prime \notin \TT$ and all other edges $e^{\prime \prime} \in \TT$ are also in $\TT^\prime$. Then, the edges $e$ and $e^\prime$ must have the same head, $u=v$. Otherwise, node $v$ would not be connected to the root in $\TT$ and node $u$ would not be connected to the root in $\TT^\prime$. Since both $\TT$ and $\TT^\prime$ are trees, $r: e \leftrightarrow e^\prime$ is a valid edge rotation.\\ \\
\textit{Induction Step:}
Let $\TT$ and $\TT^\prime$ have $N > 1$ mismatches. Then there exists at least one edge $e = (u,v) \in \TT$ and $e \notin \TT^\prime$ such that there exists another edge $e^\prime = (u^\prime, v) \in \TT^\prime$, since otherwise $\TT^\prime$ would not be connected. Now we consider the tree $\TT^{\prime \prime} = \TT \setminus \{e \} \cup \{e^\prime\}$. Since by construction $\TT$ and $\TT^{\prime \prime}$ have exactly one mismatch, we need one valid edge-rotation to reconfigure from $\TT$ to $\TT^{\prime \prime}$. Furthermore, $\TT^{\prime \prime}$ and $\TT^\prime$ have $N-1$ mismatches. According to Lemma~\ref{net_rec:thm:spanning_tree_reconf} we can reconfigure from $\TT^{\prime \prime}$ to $\TT^\prime$ using a sequence of valid edge rotations. Using the induction hypothesis, this requires $N-1$ valid edge rotations.
\end{proof}

We conclude our analysis of edge rotations by providing an upper limit for the number of mismatches and thus the number of valid edge rotations. The number of mismatches is bounded from above by
\begin{equation}
    N_{\TT, \TT^\prime} \leq  \lvert \VV \rvert - 1 - \lvert \EE_\mathrm{bridge} \rvert,
\end{equation}
where $\EE_\mathrm{bridge}$ be the set of all edges not part of any cycle in $\GG$. Hence, corollary \ref{cor:number_of_rotations} implies that we must perform at most $\lvert \VV \rvert - 1$ local valid edge rotations. 

\subsection{The Adjacency Graph of Spanning Trees}
\label{sec:edge_rot_graph_Def}

We can also adapt a different viewpoint by defining a graph $\GG_{\mathrm{Sp}} = (\VV_{\mathrm{Sp}}, \EE_{Sp})$ where spanning trees are nodes $\TT \in \VV_{Sp}$ and two spanning trees $\TT$ and $\TT^\prime$ are adjacent if there exists a valid edge rotation $r: \TT \to \TT^\prime$. Then Lemma \ref{lem:valid_edge_rotations} is equivalent to the statement that $\GG_{\mathrm{Sp}}$ is connected, and in this picture Lemma \ref{cor:number_of_rotations} says that the shortest path between any two spanning trees has length $N_{\TT, \TT^\prime}$. In general, many sequences of edge rotations, and thus many paths in $\GG_{\mathrm{Sp}}$, exist between $\TT$ and $\TT^\prime$, see Fig.~\ref{fig:graph_of_spanning_trees} for an elementary example. 

\clearpage
\section{Review: Quantum Optimization}

This section provides an introduction to quantum optimization algorithms, starting with the introduction of some quantum computing conventions in \ref{sup:conventions}. We then define Quantum Annealing in \ref{sec:annealing}. Finally, we then turn to QAOA, in \ref{sec:qaoa}.

\subsection{Conventions}
\label{sup:conventions}

We use the following standard convention for the Pauli Operators
\begin{equation*}
    X = \begin{bmatrix} 0 & 1 \\ 1 & 0 \end{bmatrix}, \, Y = \begin{bmatrix} 0 & -i \\ i & 0 \end{bmatrix}, \, Z = \begin{bmatrix} 1 & 0 \\ 0 & -1 \end{bmatrix}. 
\end{equation*}
We follow the usual convention and define the single qubit computational basis $\{ \ket{0}, \ket{1}\}$ as the eigenbasis of the Pauli-$Z$, that is 
\begin{equation}
    \label{eq:pauli_Z_EV}
    Z \ket{0} = +1 \ket{0}, \, Z \ket{1} = -1 \ket{1}.
\end{equation}
Hence, $\ket{1}$ is the ground state of the Pauli-$Z$. We further note that the eigenbasis of the Pauli-$X$ defined by 
\begin{equation*}
    X \ket{+} = +1 \ket{+}, \, X \ket{-} = -1 \ket{-},
\end{equation*}
can be written in the computational basis as 
\begin{equation*}
    \ket{+} = \frac{1}{\sqrt{2}} (\ket{0} + \ket{1}), \, \ket{-} = \frac{1}{\sqrt{2}} (\ket{0} - \ket{1}). 
\end{equation*}

\subsection{Quantum Annealing}
\label{sec:annealing}

Quantum annealing (QA) \cite{farhi_quantum_2000} is based on the principles of Adiabatic Quantum Computing (AQC) \cite{albash_adiabatic_2018}. In AQC, the solution to a computational problem is encoded in the ground state of a problem Hamiltonian $H_\mathrm{P}$. According to the quantum adiabatic theorem \cite{born_beweis_1928}, if a closed quantum system is initialized in the ground state of a mixer Hamiltonian $H_\mathrm{P}$ and evolves slowly under a time-dependent Hamiltonian that interpolates between $H_\mathrm{M}$ and $H_\mathrm{P}$, it will remain in the instantaneous ground state throughout the evolution.

QA implements this idea as a heuristic optimization method. The system evolves under an annealing schedule
\begin{equation}
\label{eq:annealing_schedule}
    H(t) = A(t) H_\mathrm{M} + B(t) H_\mathrm{P},
\end{equation}
where $A(t)$ and $B(t)$ are monotonic functions with $A(0)=B(T_\mathrm{A})=1,\, A(T_\mathrm{A})=B(0)=0$, and $T_\mathrm{A}$ is the annealing time. To satisfy the adiabatic condition, $T_\mathrm{A}$ must be large compared to the inverse of the minimum energy gap between the ground and first excited states during the evolution \cite{jansen_bounds_2007}, among other conditions. Numerical studies have shown that the minimal gap can close exponentially with the problem size \cite{albash_adiabatic_2018}, leading to exponentially large annealing times. In practice, however, the system is open and thus subject to decoherence, noise, and thermal fluctuations, and the annealing time limited by these influences. Consequently, the system can undergo a transition to an excited state of $H_\mathrm{P}$ during anneal. Then the outcome of an anneal is not a globally optimal solution to the original problem, but an approximation thereof.

On current commercially available quantum annealers, like the ones from the company D-Wave, the problem Hamiltonian $H_\mathrm{P}$ is a programmable Ising Hamiltonian of the form
\begin{equation}
\label{eq:Ising_hamiltonian}
H_\mathrm{P} = \sum_{n=1}^N \sum_{\substack{m = 1 \\ m \neq n}}^N J_{n,m} Z_n Z_m + \sum_{n=1}^N h_n Z_n,
\end{equation}
where $Z_n$ denotes the Pauli-$Z$ operator acting on qubit $n$, and $J_{n,m}$, $h_n$ are programmable real coefficients. This Hamiltonian is diagonal in the computational basis --i.e., the joint eigenbasis of all $Z_n$-- and each eigenstate corresponds to a bitstring $s=(s_1,\ldots,s_N) \in \{+1, -1\}^N$. The energy associated with each configuration $s$ reflects the objective value of the optimization problem. 

This structure makes D-Wave's QA implementation particularly suited for solving Quadratic Unconstrained Binary Optimization (QUBO) problems,
\begin{align}
\begin{split}
    &\min_{\{x_n \in \{0,1\}\}} Q \\
    \label{eq:par:QUBO_general}
    &Q = \sum_{n,m = 1}^{N} x_n Q_{n,m} x_m,    
\end{split}
\end{align}
which can be mapped to an classical Ising model via an affine change in variables $s_n = 1 - 2 x_n$, from which the Hamiltonian $H_\mathrm{P}$ can be inferred by replacing the spin variables $s_n$ with $Z_n$\footnote{Since $x_n = 0$ gets mapped to $s_n = 1$ and $x_n = 1$ gets mapped to $s_n = -1$, this convention has the advantage over the commonly used map $s_n = 2x_n - 1$ that computational basis states $\ket{011 \ldots}$ map directly to bit strings $011\ldots$, cf. \ref{eq:pauli_Z_EV} of the QUBO Problem.}. 
The QUBO cost function then becomes equivalent to the energy landscape defined by $H_\mathrm{P}$, with the ground state(s) of $H_\mathrm{P}$ encoding the optimal solution(s).

On the other hand, the Mixer Hamiltonian is fixed as a transverse field Mixer \cite{andriyash2016, lanting_experimental_2017}
\begin{equation}
\label{eq:transverse_field_mixer}
H_\mathrm{M} = - \sum_n X_n,
\end{equation}
where $X_n$ is the Pauli-$X$ operator on qubit $n$. This Hamiltonian is not diagonal in the computational basis. Instead, its ground state $\ket{+}^{\otimes^N}$, and thus the initial state, is the uniform superposition of all $2^N$ computational basis states
\begin{align*}
    \ket{\psi(0)} &= \frac{1}{\sqrt{2}^N} \sum_s \ket{s} \\
    & = \frac{1}{\sqrt{2}^N} \left( \ket{0\ldots00} + \ket{0\ldots01} + \ldots + \ket{1\ldots11}  \right). 
\end{align*}

Quantum annealing can also be understood from a different perspective: during anneal, $H_\mathrm{M}$ causes quantum tunneling between configurations $s$ in analogy to thermal fluctuations in simulated annealing \cite{kadowaki_quantum_1998}.  Over time, the influence of $H_\mathrm{M}$ is reduced while $H_\mathrm{P}$ is increased, ideally driving the system toward a low-energy (ideally ground) state of $H_\mathrm{P}$. 

For a more detailed introduction to QA and a discussion of industry applications, we refer to \cite{yarkoni2022quantum}. 

\subsection{Quantum Alternating Operating Ansatz}
\label{sec:qaoa}

QA and AQC are inherently analog models of quantum computation, where the evolution of the quantum state is governed continuously by a time-dependent Hamiltonian. In contrast, the circuit model of quantum computation is digital, relying on discrete sequences of quantum gates, that is, unitary operators. To simulate analog processes like adiabatic evolution on a digital quantum computer, the continuous time evolution must be discretized.

The time evolution of a quantum system under a time-varying Hamiltonian is described by the Schrödinger Equation, a first-order differential equation. The Schrödinger equation is solved by a time-ordered exponential. In order to discretize, one can approximate the unitary $U(0, T_\mathrm{A})$ describing the evolution from $t=0$ to $t=T_\mathrm{A}$ by slicing the interval $[0, T_\mathrm{A}]$ into $N$-small intervals of duration $\delta t = T_\mathrm{A} / N$ as 
\begin{equation}
    U(0, T_\mathrm{A}) = \prod_{k = 0}^{N-1} e^{- i H(t_k) \delta t} + \mathcal{O}(\delta t^2)
\end{equation}
with $t_k = k \delta t$. In the limit of $N\to \infty$ we recover the exact time-ordered exponential, establishing that AQC is equivalent to the circuit model in the limit $T_\mathrm{A} \to \infty$ \cite{aharonov2008adiabatic}. 

If the Hamiltonian at time $t$ is composed of two non-commuting parts, as in typical quantum annealing protocols (\ref{eq:annealing_schedule}), we can further apply a first-order Trotter decomposition to approximate each time step by a sequence of simpler unitaries
\begin{equation}
e^{- i H(t_k) \delta t} = e^{- i B(t_k) H_\mathrm{P} \delta t}  e^{- i A(t_k) H_\mathrm{M} \delta t} + \mathcal{O}(\delta t^2).
\end{equation}
Thus, the full annealing schedule can be approximated as a discrete sequence of gates
\begin{equation}
\label{eq:schedule_qaoa}
U(0, T_\mathrm{A}) \approx \prod_{k = 0}^{N-1} e^{- i B(t_k) H_\mathrm{P} \delta t} e^{- i A(t_k) H_\mathrm{M} \delta t}, 
\end{equation}
which allows the continuous adiabatic process to be implemented or simulated on gate-based quantum hardware. 

In light of this observation, the Quantum Approximate Optimization Algorithm (QAOA) \cite{farhi2014quantum} can be seen as a digitized version of adiabatic quantum optimization. In its original formulation, QAOA considers a finite number $N$ of discrete layers, replacing the continuous annealing schedule with variational parameters $\vec{\beta}=(\beta_1,…,\beta_N)$ and $\vec{\gamma}=(\gamma_1,…,\gamma_N)$,
Each layer applies a problem unitary $U_\mathrm{P}(\gamma) = e^{-i \gamma_n H_\mathrm{P}}$ followed by a mixing unitary $U_\mathrm{M}(\beta) = e^{-i \beta_n H_\mathrm{M}}$, using the same Hamiltonians $H_\mathrm{P}$ and $H_\mathrm{M}$ as in D-Wave’s analog quantum annealing framework. QAOA is a hybrid algorithm; the angles $\vec{\beta}$ and $\vec{\gamma}$ are optimized in an outer classical loop. That is, if the final state is denoted by $\ket{\psi(\vec{\beta}, \vec{\gamma})}$ we seek to minimize the expectation value
\begin{equation*}
    \braket{\psi(\vec{\beta}, \vec{\gamma}) | H_\mathrm{P} | \psi(\vec{\beta}, \vec{\gamma} )}
\end{equation*}
Furthermore, the initial state $\ket{\psi(0)}$ can be easily prepared by a Hadamard transform.

Since the unitaries $U_\mathrm{M}$ and $U_\mathrm{P}$ can be implemented using shallow quantum circuits and improvements in the approximation error have been analytically guaranteed for $N=1$ \cite{farhi2014quantum}, QAOA has been considered as a promising candidate for NISQ-devices \cite{zhou2020quantum, willsch2020benchmarking}. However, for $N>0$ the energy landscape becomes more complex, and QAOA is known to be prone to converging into suboptimal local minima \cite{willsch2020benchmarking}. Hence, initialization of angles $\vec{\beta}$ and $\vec{\gamma}$ determines the performance of QAOA. Numerical experiments show that initializing the angles according to a linear annealing schedule avoids sampling suboptimal local minima frequently \cite{sack2021quantum}. 

While the original QAOA used a transverse-field Mixer (\ref{eq:transverse_field_mixer}), later work introduced problem-aware mixers Hamiltonians that incorporate hard constraints directly \cite{hen2016quantum, hen2016driver}. That is, the time evolution of the Mixer Hamiltonian preserves the feasible space. If the circuit is initialized in any superposition of feasible states, only feasible states are sampled on a fully fault-tolerant machine. However, to be able to apply QAOA with a discretized annealing schedule, thus without freely optimizing the angles $\vec{\beta}$ and $\vec{\gamma}$, the initial state has to be the ground state of the Mixer Hamiltonian in the feasible space. This idea was then formalized in the Quantum Alternating Operator Ansatz (also abbreviated QAOA) \cite{hadfield_quantum_2017, hadfield_quantum_2019}, which extends the original framework by allowing general unitaries as mixers $U_\mathrm{M}(\beta)$. These Mixers need to preserve the feasible subspace defined by the problem’s constraints and need not commute with the cost unitaries $[U_\mathrm{M}(\beta), e^{-i \gamma H_\mathrm{P}}] \neq 0$. 
This generalization is crucial for many real-world combinatorial optimization problems, where constraints (e.g., scheduling, routing, matching) are nontrivial and can not be mapped to Ising penalty terms directly. However, the design of feasibility-preserving mixers for real-world problems becomes significantly more complex \cite{stollenwerk_toward_2020}.

A critical challenge in the context of problem-aware Mixers is that quantum errors can cause the system to leave the feasible subspace. Once an error occurs, the QAOA evolution does not offer a mechanism to return to the feasible space. This sensitivity to errors makes error mitigation and correction particularly important for constrained QAOA variants. Fortunately, for certain classes of mixers, symmetry-based error correction schemes can be constructed to correct such deviations \cite{streif2021errorqaoa}. These schemes exploit the underlying algebraic or combinatorial symmetries of the problem and require fewer resources than standard error correction codes. 

Beyond the design of feasibility-preserving mixers, another crucial aspect of gate-based quantum computing is the ability to handle higher-order cost Hamiltonians directly. On gate-based architectures, there are two strategies to deal with higher-order cost functions (or penalty terms). First, quadratization refers to the reduction to a quadratic Hamiltonian~\cite{dattani2019quadratization}, generally requiring the introduction of auxiliary binary variables and additional penalty parameters $\lambda$. Thus, quadratization increases the search space and the number of interactions, increasing the gate complexity of the Hamiltonian simulation. In particular, for a penalty term with $x$ quartic expressions in the binary variables, one needs to introduce at most $x$ auxiliary binary variables \cite{dattani2019quarticreduction}. Second, the higher-order terms can be simulated directly, which leads to an increase in gate complexity, cf. \cite{dragoi2025approx}.

\newpage
\section{Mixer Construction and Resource Estimation}

In this section, we construct the partial mixers $U^\text{PM}_{r: e \leftrightarrow e^\prime}(\beta)$ and decompose them into elementary single- and two-qubit gates, thereby providing a proof of Theorem~\ref{thm:partial_ mixer} of the main paper. 

Our design relies on a novel circuit that performs general (partial) mixing between two feasible states $\ket{\phi_A}$ and $\ket{\phi_B}$ whenever the transitions $\ket{\phi_A} \to \ket{\phi_B}$ and $\ket{\phi_B} \to \ket{\phi_A}$ can be implemented by the same sequence of unitaries, as formalized in~\ref{sec:syncrot}. 
We then build $U^\text{PM}_{r: e \leftrightarrow e^\prime}(\beta)$ in~\ref{sec:partial_mixers} by combining this general construction with a controlled operation on ancillary qubits to verify the validity of the edge swap $r: e \leftrightarrow e^\prime$ and update the relevant register qubits $\ket{y_{e,n}}$ according to the graph topology. 
In~\ref{sec:ressources_pm}, we decompose the circuit into arbitrary single-qubit and CNOT gates to estimate resource requirements, and in~\ref{sec:full_mixer} we discuss the full mixer designs, including a Qiskit implementation for a simple three-node example in~\ref{sec:simple_example_mixer}.

\subsection{Mixing between two feasible states: A general perspective }
\label{sec:syncrot}

In its general form, our partial Mixer needs to implement the rotation from one feasible state $\ket{\phi_A}$ to another feasible state $\ket{\phi_B}$ in the plane spanned by both states. That is, its action of the rotation can be described as a parameterized unitary acting as
\begin{equation}
    \label{eq:gen_transition}
    U(\beta) \ket{\phi_A} := \cos(\frac{\beta}{2}) \ket{\phi_A} + i \sin(\frac{\beta}{2})\ket{\phi_B}.
\end{equation}
Under certain conditions, such rotations can be implemented indirectly by applying a phase gate
\begin{equation*}
    P(\beta) = \begin{bmatrix} 1 & 0 \\ 0 & e^{i \beta} \end{bmatrix}
\end{equation*}
to an ancillary qubit $\ket{anc}$, initialized as $\ket{0}$, and leveraging phase kickback to transfer the effect to the target system, as the following lemma shows.

\begin{lem}
    \label{lem:synchronized_rotations}
    Let $(U_1, \ldots, U_N)$ be a sequence of gates such that 
    \begin{align}
    \begin{split}
    \label{eq:conditions_synchronized_rotations}
    & \ket{\phi_B} = U_N \ldots U_1 \ket{\phi_A}, \\
    & \ket{\phi_A} = U_N \ldots U_1 \ket{\phi_B}.
    \end{split}
    \end{align}
    Then the state after applying the following circuit 
    \begin{equation*}
    \begin{quantikz}[scale=0.50, classical gap=0.1cm]
       \ket{0} \, & \gate{H} & \ctrl{1} & \ \ldots\ & \ctrl{1} & \gate{H} 
        & \gate{P(\beta)}\slice{$\ket{\psi_1}$} 
        & \gate{H} & \ctrl{1} & \ \ldots\ & \ctrl{1} & \gate{H} & \\
       \ket{\phi_A} \, \setwiretype{b} & \qw  & \gate{U_1} & \ \ldots\ & \gate{U_N}      & \qw     
        & \qw      & \qw      & \gate{U_1}  & \ \ldots\   & \gate{U_N} & \qw & \qw
    \end{quantikz}
    \end{equation*}
    is $e^{i\beta/2}U(\beta)\ket{0}\ket{\phi_A}$, which differs from~\eqref{eq:gen_transition} only by a global phase $e^{i \beta/2}$. That is, the circuit produces the same measurement statistics as~\eqref{eq:gen_transition}.
\end{lem}

\begin{proof}

Let's briefly analyze this circuit. After the phase gate, the intermediate state is given by 
\begin{equation*}
    \ket{\psi_1} = \frac{1}{2} \left[ \ket{0} (\ket{\phi_A} + \ket{\phi_B}) + e^{i\beta} \ket{1} (\ket{\phi_A} - \ket{\phi_B}) \right].
\end{equation*}
Then applying the Hadamard gate to the first qubit and rearranging terms yields 
\begin{equation*}
    \frac{1}{2 \sqrt{2}} \left[ \ket{0} \left( (1 +e^{i \beta}) \ket{\phi_A} + (1 -e^{i \beta}) \ket{\phi_B}  \right)   + \ket{1} \left( (1 -e^{i \beta}) \ket{\phi_A} + (1 +e^{i \beta}) \ket{\phi_B}  \right)\right].
\end{equation*}
Using conditions~(\ref{eq:conditions_synchronized_rotations}) we see that the final before the last Hadamard gate is
\begin{equation*}
    \frac{1}{2 \sqrt{2}} \left[ \ket{0} \left( (1 +e^{i \beta}) \ket{\phi_A} + (1 -e^{i \beta}) \ket{\phi_B}  \right)   + \ket{1} \left( (1 -e^{i \beta}) \ket{\phi_B} + (1 +e^{i \beta}) \ket{\phi_A}  \right)\right].
\end{equation*}
After applying the last Hadamard, one obtains
\begin{equation*}
    \frac{1}{2} \ket{0} \left[ (1 +e^{i \beta}) \ket{\phi_A} + (1 -e^{i \beta}) \ket{\phi_B}  \right], 
\end{equation*}
which shows that the ancillary qubit is not entangled with the other qubits and is uncomputed at the end. Hence, for the register initialized as $\ket{\phi_A}$, the circuit implements the rotation operation~(\ref{eq:gen_transition}), with a global phase $e^{i\beta}$, which can be seen by multiplying out such phase 
\begin{align*}
    \frac{1}{2}(1 +e^{i \beta}) \ket{\phi_A} +  \frac{1}{2} (1 -e^{i \beta}) \ket{\phi_B} &= e^{i \frac{\beta}{2}} \left( \frac{e^{i\frac{\beta}{2}} + e^{-i\frac{\beta}{2}}}{2} \ket{\phi_A} + i \frac{e^{i\frac{\beta}{2}} - e^{-i\frac{\beta}{2}}}{2i} \ket{\phi_B}\right) \\
    &= e^{i \frac{\beta}{2}} \left( \cos(\frac{\beta}{2}) \ket{\phi_A} + i \sin(\frac{\beta}{2}) \ket{\phi_B}\right). 
\end{align*}
\end{proof}

We get the following corollary:
\begin{cor}
    Let $(U_1, \ldots, U_N)$ be a sequence of Hermitian ($U_i = U_i^\dagger$ $\forall i$) mutually commuting gates, such that $\ket{\phi_B} = U_N \ldots U_1 \ket{\phi_A}$. Then the conditions~(\ref{eq:conditions_synchronized_rotations}) are satisfied and the rotation can be implemented as in Lemma \ref{lem:synchronized_rotations}. \label{cor:syncrhonized_rotations} 
\end{cor}

This occurs for example, if the $U_i$ are Pauli Operators or SWAP gates acting on different qubits. A simple example for such a sequence can be found in Fig. \ref{sup:fig:synchronized_rotation_and_swap}, where the sequence $(X_1, X_2, \text{SWAP}(3,4))$ maps the state $\ket{0010}$ to $\ket{1101}$ and vice-versa.
\begin{figure}[t]
    \centering
    \includegraphics[width=\columnwidth]{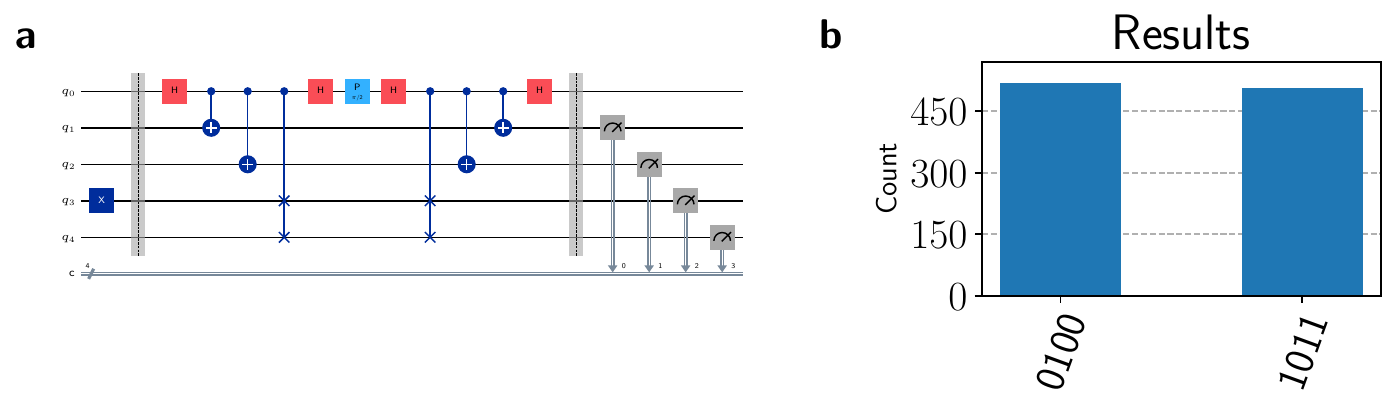}
    \caption{Simple example for a parameterized rotation in the $\ket{0010}-\ket{1101}$ plane. The rotation can implemented using Lemma \ref{lem:synchronized_rotations} with the sequence $(X_1, X_2, \text{SWAP}(3,4))$. \textbf{a:} Implementation in Qiskit. The ancillary qubit $\ket{anc}$ is the first qubit $q_0$. The initial state $\ket{0010}$ is prepared by applying a single $X$ gate to the fourth qubit.  \textbf{b:} Histogram of the results when simulating the circuit 1000 times for $\beta = \pi /2$. An ideal quantum computer with no noise is assumed at this stage.}
    \label{sup:fig:synchronized_rotation_and_swap}
\end{figure}

\subsection{Explicit Construction of Partial Mixers $U^\text{PM}_{r: e \leftrightarrow e^\prime}(\beta)$}
\label{sec:partial_mixers}

\begin{figure}[t!]
    \centering
   	\includegraphics[width=0.5\textwidth]{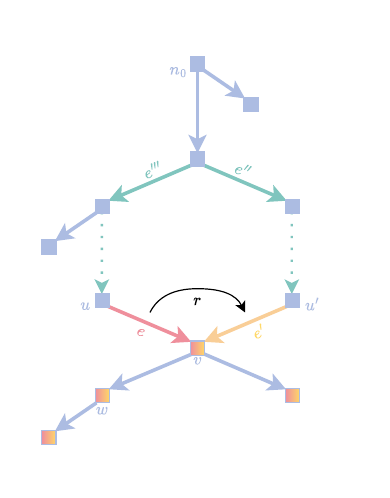}
    \caption{Depiction of a valid edge rotation $r:e=(u,v)\mapsto e^\prime=(u^\prime,v^\prime)$, highlighting all variables that are relevant. In red, the current active edge $e$, which is rotated maintaining the head $v$, to the edge $e^\prime$, in yellow. The gradient colored nodes, such as node $w$, represent all nodes downward of $e$. Finally, the green edges, such as $e^{\prime\prime}$, $e^{\prime\prime\prime}$, represent the edges laying in the unique path between $u$ and $u^\prime$.}
    \label{sup:fig:pm_update_var}
\end{figure}

The goal of this section is to construct a unitary that mixes between the two states $\ket{\bf{y}}$ and $\ket{\bf{y}^\prime}$ that encode the two trees $\TT$ and $\TT^\prime$ before and after a valid local edge rotation. That is, we construct a mixing operation that preserves the feasible space $\mathrm{Sp}(\GG, n_0)$. 

Let $\TT$ be a spanning tree of $\GG$ with root $n_0$ and the natural orientation implied by the root. As stated in the main text, such tree is represented with a set of $|\EE|(|\VV|-1)$ binary variables, $\{y_{e,v}\}_{e\in\EE, v\in \VV\setminus\{n_0\}}$, which only evaluate to 1 if the node $v$ is downwards of the edge $e$, and this edge is present in $\TT$, that is,
\begin{align}
\label{sup:eq:binary_variables}
    y_{e,n} = \begin{cases}
        1 & \text{if } n \neq n_0 \text{ is downward of } e \in \TT,  \\
        0 & \text{else.}
    \end{cases}
\end{align}
On the quantum computer, each binary variable is encoded in a qubit. By considering a collection of such qubits, we represent the tree $\TT$ as a quantum state, denoted by $\ket{\bf{y}}=\bigotimes_{e\in \EE,n\in \VV\setminus\{n_0\}}\ket{y_{e,n}}$.

We look at two edges $e,e^\prime\in \GG$ with their respective tails $u$, $u^\prime$, and a shared head $v$, that is, $e=(u,v)$ and $e^\prime = (u^\prime, v)$ (cf. Fig.~\ref{sup:fig:pm_update_var}), and assume that $e \in \TT$. This implies that $e^\prime \notin \TT$, as otherwise $\TT$ would not be a tree. Consider the \emph{edge rotation} $r$ between the two edges, described previously as the map
\begin{align*}
    r:\TT &\to \TT^\prime:=\TT+e^\prime-e\\
    e = (u,v) &\mapsto e^\prime = (u^\prime, v), 
\end{align*}
which we called \emph{valid} if $u^\prime$ is not downward of the edge $e$ in $\TT$ (cf.  Lemma~\ref{lem:valid_edge_rotations}). Note that $\TT^\prime$ will only be a tree if the rotation is valid. 

We aim to define a quantum operation that is able to rotate between the state $\ket{\bf{y}}$ and $\ket{\bf{y}^\prime}$, representing $\TT$ and $\TT^\prime$, before and after the valid edge rotation, respectively.
Based on the previous discussion on mixing between two states in \ref{sec:syncrot}, we first need to design a sequence of unitaries $U_1, U_2, \ldots$ that transforms $\ket{\bf{y}}$ to $\ket{\bf{y}^\prime}$, and vice-versa, and that oblige the two conditions~(\ref{eq:conditions_synchronized_rotations}). 
Note that these conditions do not clash with the classical move in the following sense: if $e \mapsto e^\prime$ is valid for the spanning tree $\TT$, then $e^\prime \mapsto e$ is valid for the spanning tree $\TT^\prime = \TT+e^\prime-e$, hence we call the two rotations $e \mapsto e^\prime$ and $e^\prime \mapsto e$ \emph{reciprocal}. Consequently we incorporate both edge rotations into one quantum operation $U_{r: e \leftrightarrow e^\prime} = U_1 U_2 \cdots$, which we call an \emph{edge swap}. 

To this end, we observe that performing an edge swap (or edge rotation) affects a large number of the variables $y_{e,n}$. We can identify two distinct contributions as we will explain in the following. Accordingly, we can decompose the overall operation $U_{r: e \leftrightarrow e^\prime}$ into two corresponding sub-operations, that is 
\begin{equation*}
    U_{r: e \leftrightarrow e^\prime} := U^{\text{swap}}_{r: e \leftrightarrow e^\prime} U^{\text{path}}_{r: e \leftrightarrow e^\prime}. 
\end{equation*}

First, all nodes originally downward of $e$, like the node $w$ and all other gradient colored nodes in Fig.~\ref{sup:fig:pm_update_var}, become downward of $e^\prime$ after the rotation $r: \TT \to \TT^\prime$. This means that for such nodes, before the rotation, we have $y_{e,w}=1$ and $y_{e^\prime,w}=0$ and after the rotation, we have $y_{e,w}=0$ and $y_{e^\prime,w}=1$. The reverse is true for the reciprocal rotation. Hence, updating these variables for the combined edge swap can be achieved by swapping or interchanging their values. All actions are summarized in the unitary $U^{\text{swap}}_{r: e \leftrightarrow e^\prime}$ with details of the implementation given below. 

Second, also for the edges $e^{\prime \prime}$ and $e^{\prime \prime \prime}$ on the path from tail $u$ to tail $u^\prime$ (green in Fig.~\ref{sup:fig:pm_update_var}) the downward relations to the nodes $v, w, \ldots$ (gradient colored nodes) are affected. For example, for the edge $e^{\prime \prime \prime}$, the nodes $v, w, \ldots$ won't be downward anymore after the edge rotation $e\mapsto e^\prime$ and thus the variables $y_{e^{\prime \prime \prime}, w}$ need to be updated as $1 \to 0$. Similarly, for the edge $e^{\prime \prime}$ the nodes $v, w, \dots$ become downward, and hence $y_{e^{\prime \prime}, w}$ are updated as $0 \to 1$. Since for the reciprocal rotation $e^\prime \mapsto e$ we observe similar updating rules, we conclude that updating the variables for the edges on the path from $u$ to $u^\prime$ can be achieved by simple negations, which we all incorporate in the unitary $U^{\text{path}}_{r: e \leftrightarrow e^\prime}$.

Note that for all other edges (blue in Fig.~\ref{sup:fig:pm_update_var}), the corresponding variables remain unaffected by the edge swap. This is because the nodes $v, w, \ldots$ are either not located below these edges in either $\TT$ or $\TT^\prime$, or they are below them in both trees, as in the case of edges shared by the paths from the root $n_0$ to the nodes $u$ and $u^{\prime}$.

Together, $U_1=U^{\text{path}}_{r: e \leftrightarrow e^\prime}$ and $U_2=U^{\text{swap}}_{r: e \leftrightarrow e^\prime}$ comprise a set of unitaries that establish the transitions $\ket{\bf{y}}\to \ket{\bf{y}^\prime}$ and $\ket{\bf{y}^\prime}\to \ket{\bf{y}}$ if and only if the edge swap is valid. Then and only then, Lemma~\ref{lem:synchronized_rotations} can be applied to mix between these states. Hence, one crucial step is to check the validity of the edge swap and apply the transition $U_1=U^{\text{path}}_{r: e \leftrightarrow e^\prime}$ and $U_2=U^{\text{swap}}_{r: e \leftrightarrow e^\prime}$ only in the case the validity is evaluated as \emph{true}. If, for a given configuration $\ket{\bf{y}}$, an edge swap is not valid, the operation needs to become the identity. Consequently, we can not apply Lemma~\ref{lem:synchronized_rotations} immediately, but need to add one extra layer. 

Recall that the conditions for a valid edge swap are
\begin{enumerate}
    \item one of the edges $e, e^\prime$ is active, 
    \item both edge rotations in the edge swap are valid, that is, they do not create cycles.
\end{enumerate}

The first condition will be directly incorporated in the design of the unitaries $U^{\text{path}}_{r: e \leftrightarrow e^\prime}$ and $U^{\text{swap}}_{r: e \leftrightarrow e^\prime}$, which will become the identity if both edges are not active (cf. the implementation below). Additionally, note that both edges cannot be active simultaneously, as for any tree, no two edges can point to the same node.

The second condition needs to be inferred from the binary variables $y_{e,n}$. According to Lemma~\ref{lem:valid_edge_rotations}, when rotating from $e$ to $e^\prime$, the node $u^\prime$ must not be downward of the edge $e$. Similarly, for the reciprocal rotation, the node $u$ must not be downward of the edge $e^\prime$. That is, we must have that $y_{e^\prime, u} = 0$ and $y_{e, u^\prime} = 0$. 
Hence, for configurations $\ket{\bf{y}}$ where Boolean function 
\begin{equation}
    f_{r: e \leftrightarrow e^\prime} = \neg y_{e^\prime, u} \land \neg y_{e, u^\prime}
\end{equation}
evaluates to \emph{true}, the edge swap unitary can be applied. The function can be implemented as a unitary $U_{f_{r: e \leftrightarrow e^\prime}}$, acting as 
$U_{f_{r: e \leftrightarrow e^\prime}} \ket{\mathbf{y}}\ket{0} = \ket{\mathbf{y}}\ket{f_r}$, 
using a zero-controlled Toffoli gate with an ancilla initialized in $\ket{0}$. 
The ancilla stores the function outcome and serves as the control qubit $\ket{f_r} $ for the following operations.

Finally, combining Lemma~\ref{lem:synchronized_rotations} with the validity checking operation, the \emph{partial controlled edge swap mixer} $U^{\text{PM}}_{r: e \leftrightarrow e^\prime}(\beta)$ can be schematically represented as
\begin{equation}
\begin{quantikz}[scale=0.5, classical gap=0.1cm]
    \lstick[5]{$\ket{\bf{y}}$} \quad \quad & \wireoverride{n} & \setwiretype{b} & \gategroup[7,steps=2,style={dashed,rounded
    corners,fill=blue!20, inner
    xsep=2pt},background,label style={label
    position=above,anchor=north,yshift=0.45cm}]{{$U_{f_{r: e \leftrightarrow e^\prime}}$}}  & &  \gate[6]{U^\text{path}_{r: e \leftrightarrow e^\prime}}\gategroup[6,steps=2,style={dashed,rounded
    corners,fill=red!20, inner
    xsep=2pt},background,label style={label
    position=above,anchor=north,yshift=0.45cm}]{$U_{r:e \leftrightarrow e^\prime}$} & \gate[5]{U^\text{swap}_{r: e \leftrightarrow e^\prime}} & & & & \gate[5]{U^\text{swap}_{r: e \leftrightarrow e^\prime}} & \gate[6]{U^\text{path}_{r: e \leftrightarrow e^\prime}} & & \\ 
    & \wireoverride{n} \lstick[1]{$\ket{y_{e^\prime, u}} \,$ } & & \ctrl[open]{5} & & & &  & & & & & \ctrl[open]{5} & \\
    & \wireoverride{n} & \setwiretype{b} & & & & & & & & & & &   \\
    &\wireoverride{n} \lstick[1]{$\ket{y_{e, u^\prime}} \, $} & & \ctrl[open]{3} & & & & &  & & & & \ctrl[open]{3} &  \\
    &\wireoverride{n} & \setwiretype{b} & & & & & & & & & & & \\
    &\lstick[1]{$\ket{\text{0}}$}\wireoverride{n}   & \qwbundle{4} & & & & & & & & & & &  \\
    &\lstick[1]{$\ket{0}$ }\wireoverride{n}& & \gate[1]{X} & \ket{f_{r}} & \ctrl{-1} &  \ctrl{-2}& & \ctrl{1} & & \ctrl{-2} &  \ctrl{-1} & \gate[1]{X} & \\
    &\lstick[1]{$\ket{anc}$}\wireoverride{n}& & \gate[1]{H} & & \ctrl{-2} & \ctrl{-3} & \gate[1]{H} & \gate[1]{P(\beta)} & \gate[1]{H} & \ctrl{-2} & \ctrl{-3} & \gate[1]{ H}& 
\end{quantikz}
\label{sup:eq:partialmixer}
\end{equation}
The key difference from~\ref{sec:syncrot} is that the unitaries $U_1$, $U_2$, and the phase gate $P(\beta)$ are all controlled by the qubit $\ket{f_r}$, which encodes whether the swap $r: e \leftrightarrow e^\prime$ is valid. 
In particular, also the phase gate $P(\beta)$, implementing the mixing rotation, must act on $\ket{\text{anc}}$ only when the swap is valid; otherwise, the partial mixer would not reduce to the identity. Moreover, the ancilla can be uncomputed by reapplying the same zero-controlled Toffoli gate (since it is self-inverse). 
This is possible because the states $\ket{y_{e^{\prime},u}}$ and $\ket{y_{e,u^{\prime}}}$ remain unchanged during the variable updates associated with the edge swap $r: e \leftrightarrow e^\prime$. 
We will demonstrate this explicitly in the construction of 
$U^{\text{path}}_{r: e \leftrightarrow e^{\prime}}$ and $U^{\text{swap}}_{r: e \leftrightarrow e^{\prime}}$, whose detailed implementation is discussed in the following.

\subsubsection*{Update downward variables for all edges $e^{\prime \prime}$ on the path between the tails:}

We now provide detailed construction of the unitary $U^{\text{path}}_{r: e \leftrightarrow e^{\prime},}$, which implements the update of the downward variables $y_{e^{\prime \prime}, w}$, for the edges $e^{\prime \prime}, e^{\prime \prime \prime}, \ldots$ (cf.~Fig.~\ref{sup:fig:pm_update_var}) lying on the path $p$ between $u$ and $u^\prime$ (the tails of $e$ and $e^\prime$) and nodes $w$ downwards of either $e$ or $e^\prime$ before the rotation. As established earlier, this can be achieved by negating the affected variables. 

Overall, to realize this unitary, we iterate over all edge-node pairs, marking in one ancillary qubit whether an edge lies on the path and in another ancillary qubit whether a node is affected. A single Toffoli gate, controlled by these two ancillas among others, then applies the required negation.

To determine whether an edge $e^{\prime\prime}$ lies on the path $p$, we note that for edges along this path, exactly one of the nodes $u$ or $u^{\prime}$ is downward: either $u$ is downward and $u^{\prime}$ is not, or vice versa. Thus, we can evaluate the Boolean function 
\begin{equation}
    g_{e^{\prime \prime},  r: e \leftrightarrow e^\prime} = \left( y_{e^{\prime \prime}, u^\prime} \land \neg y_{e^{\prime \prime}, u} \right) \lor \left( \neg y_{e^{\prime \prime}, u^\prime} \land y_{e^{\prime \prime}, u} \right)
\end{equation}
to identify all edges $e^{\prime \prime}$ on the path $p$. The function can be implemented as a unitary with an action 
\begin{equation}
    U_{g_{e^{\prime \prime},  r: e \leftrightarrow e^\prime}} \ket{\bf{y}} \ket{0} \ket{0} \ket{0} = \ket{\bf{y}} \ket{y_{e^{\prime \prime}, u^\prime} \land \neg y_{e^{\prime \prime}, u}} \ket{\neg y_{e^{\prime \prime}, u^\prime} \land y_{e^{\prime \prime}, u}} \ket{g_{e^{\prime \prime}}}
\end{equation}
which in a circuit reads
\begin{align}
\label{eq:circ_u_g_e_prime_prime}
    U_{g_{e^{\prime \prime},  r: e \leftrightarrow e^\prime}} = &\begin{quantikz}[scale=1, classical gap=0.1cm]
             & \setwiretype{b} & & & & &\\ 
             \lstick[1]{$\ket{y_{e^{\prime \prime}, u}} \, \,$ }& & \ctrl[closed]{4}  & \ctrl[open]{5} & & &\\ 
             & \setwiretype{b} & & & & &\\ 
             \lstick[1]{$\ket{y_{e^{\prime \prime}, u^\prime}} \, \,$} & & \ctrl[open]{2} & \ctrl[closed]{3} &  & &\\ 
             & \setwiretype{b} & & & & & \\ 
            \lstick[1]{$\ket{\text{0}}\, \,$} & & \gate[1]{X} &  & \ctrl[closed]{2} & &\\
            \lstick[1]{$\ket{\text{0}}\, \,$}& & & \gate[1]{X} & & \ctrl[closed]{1} & \\
            \lstick[1]{$\ket{\text{0}}\, \,$}& & & & \gate[1]{X} & \gate[1]{X} &  \rstick[1]{$\ket{g_{e^{\prime \prime}}}$}\\
            & & & & & & \\
            \lstick[1]{$\ket{f_r}\, \,$}& & & & & & \\
            \lstick[1]{$\ket{anc}\, \,$}& & & & & & 
    \end{quantikz}. 
\end{align}
The nodes $w$ affected by the rotation are the nodes downward of either $e$ or $e^\prime$, depending on whether $e$ or $e^\prime$ is active at the beginning. For example, in Figure~\ref{sup:fig:pm_update_var} $e$ is active at the beginning and thus the nodes $w$ under question are downward of $e$, not of $e^\prime$. Thus, the nodes $w$ can be identified by evaluating the Boolean function 
\begin{equation}
    h_{w,  r: e \leftrightarrow e^\prime} = y_{e, w} \lor y_{e^{\prime}, w}
\end{equation}
which can be implemented as a unitary with action 
\begin{equation}
	U_{h_{w,  r: e \leftrightarrow e^\prime}}  \ket{\bf{y}} \ket{0} = \ket{\bf{y}} \ket{h_{w}},
\end{equation} 
via the following circuit: 
\begin{align}
\label{eq:circ_u_h_w}
    U_{h_{w,  r: e \leftrightarrow e^\prime}}=&\begin{quantikz}[scale=1, classical gap=0.1cm]
             & \setwiretype{b} & & & \\ 
             \lstick[1]{$\ket{y_{e, w}} \, \,$ }& & \ctrl[closed]{5}  & & \\ 
             & \setwiretype{b} & & & \\ 
             \lstick[1]{$\ket{y_{e^{\prime}, w}} \, \,$} & &  & \ctrl[closed]{3} &   \\ 
             & \setwiretype{b} & & &  \\ 
            & \qwbundle{3} & & &   \\
            \lstick[1]{$\ket{0} \, \,$}& & \gate[1]{X} & \gate[1]{X} &  \rstick[1]{$\ket{h_{w}}$} \\
            \lstick[1]{$\ket{f_r} \, \,$} & & & & \\
            \lstick[1]{$\ket{anc} \, \,$} & & & & \\
    \end{quantikz}.
\end{align}
The full unitary $U^\text{path}_{r: e \leftrightarrow e^\prime} $ to update the downward variables for all edges on the path is then schematically given by 
\begin{equation}
\begin{quantikz}[scale=0.50, classical gap=0.1cm]
    \lstick[3]{$\ket{\bf{y}}$} \setwiretype{b} & \gate[7]{U^\text{path}_{r: e \leftrightarrow e^\prime}} & \\
    && \\
    &\setwiretype{b}& \\
    \lstick[1]{$\ket{0}$}&& \\
    \lstick[1]{$\ket{0}$}&& \\
    \lstick[1]{$\ket{0}$}&& \\
    \lstick[1]{$\ket{0}$}&& \\
    \lstick[1]{$\ket{f_r}$}& \ctrl{-1} & \\
    \lstick[1]{$\ket{anc}$}& \ctrl{-2} & \\
\end{quantikz}= \prod_{e^{\prime \prime} \in \EE \setminus\{e, e^\prime\}} U_{g_{e^{\prime \prime},  r: e \leftrightarrow e^\prime}}
\left( \prod_{w \in \VV \setminus\{u, u^\prime, n_0\}} U_{h_{w,  r: e \leftrightarrow e^\prime}}\cdot 
\begin{quantikz}[scale=0.50, classical gap=0.1cm]
    \setwiretype{b} & & \\
    \lstick[1]{$\ket{y_{e^{\prime \prime},w}}$}&\gate[1]{X}& \\
    &\setwiretype{b}& \\
    && \\
    && \\
    \lstick[1]{$\ket{g_{e^{\prime \prime}}}$}&\ctrl{-4}& \\
    \lstick[1]{$\ket{h_w}$}& \ctrl{-5}& \\
    \lstick[1]{$\ket{f_r}$}& \ctrl{-6} & \\
    \lstick[1]{$\ket{anc}$}& \ctrl{-7} & \\
\end{quantikz}
\cdot U_{h_{w,  r: e \leftrightarrow e^\prime}}
\right) \cdot U_{g_{e^{\prime \prime},  r: e \leftrightarrow e^\prime}}^{-1}
\label{sup:eq:Upath}
\end{equation}
The central 4-fold controlled X gate implements the negation of the variables $y_{e^{\prime\prime},w}$ if $\ket{g_{e^{\prime\prime},}}$, $\ket{h_w}$ and $\ket{f_r}$ are in state $\ket{1}$, that is, according to whether the edge $e^{\prime \prime}$ lies on the path $p$, whether the node $w$ needs to be updated, and whether the edge swap is valid. The control on the ancilla $\ket{anc}$ is needed for the mixing rotation, cf. \ref{sec:syncrot}. Before the negation by the 4-fold controlled X gate, computation of the ancillary qubits $\ket{g_{e^{\prime\prime}}}$ and $\ket{h_w}$ is done using the circuits \eqref{eq:circ_u_g_e_prime_prime} and \eqref{eq:circ_u_h_w}. Importantly, uncomputation of $\ket{g_{e^{\prime\prime},}}$ and $\ket{h_w}$ can be achieved by applying the inverse circuits since none of the qubits involved in the two unitaries is $\ket{y_{e^{\prime\prime},w}}$, the one modified by the X gate. As $U_{h_{{w,  r: e \leftrightarrow e^\prime}}}$ is its own inverse, we can just apply it again, for $U_{g_{e^{\prime \prime},  r: e \leftrightarrow e^\prime}}$ we just need to apply the gates in the inverse order. Uncomputation allows us to reuse the ancillary qubits. 

Observe that if no edge $e$ or $e^\prime$ is active, then $y_{e,n}$ and $y_{e^\prime,n}$ are zero for any node $n$, thus $h_w$ evaluates always to zero and as a consequence the 4-fold controlled X gate is never applied, so $U^\text{path}_{r: e \leftrightarrow e^\prime}$ is the identity as required. 

Finally, it is worth noting that under some conditions $U^{\text{path}}_{r: e \leftrightarrow e^{\prime},}$ simplifies. First, if $u$ or $u^\prime$ is the root $n_0$ and since variables $y_{e^{\prime \prime}, n_0}$ are not defined (since the root can never be downward of any edge), $U_{g_{e^{\prime \prime}, r: e \leftrightarrow e^\prime}}$ reduces to only one CNOT. Second, in the case that $e$ and $e^\prime$ are multiedges, that is, they have the same head and tail, there is no path $p$ between the tails, and $U^{\text{path}}_{r: e \leftrightarrow e^{\prime}}$ can be omitted completely. 

\subsubsection*{Update downward variables for the edges $e$ and $e^\prime$:}

The other effect an edge rotation has is that the nodes $w$ that have been previously downward of edge $e$ are now downward of the edge $e^\prime$, or vice-versa in the reciprocal rotation. To update the downward variables, it suffices to interchange the values of $y_{e,w}$ and $y_{e^\prime, w}$ for all $w \in \VV \setminus \{u, u^\prime \}$. On the quantum computer, this is achieved with a SWAP gate for each node, controlled by the ancillary qubit $\ket{f_r}$, assessing the validity of the rotation and the ancilla $\ket{anc}$ needed for the mixing. Schematically, the controlled unitary $U^\text{swap}_{r: e \leftrightarrow e^\prime}$ can be implemented as 

\begin{equation}
\begin{quantikz}[scale=0.50, classical gap=0.1cm]
\lstick[5]{$\ket{\bf{y}}$}& [20 pt]\lstick{}\setwiretype{n}& \gate[5]{U^\text{swap}_{r: e \leftrightarrow e^\prime}}\setwiretype{b} & \midstick[8,brackets=none]{$\displaystyle = \prod_{w \in \VV \setminus\{u, u^\prime, n_0\}}$ } & & \\
    &\lstick[1]{$\ket{y_{e,w}}$}\setwiretype{n}& \setwiretype{q} &         & \swap{2} & \\
    &\lstick[1]{               }\setwiretype{n}& \setwiretype{b} &         &          &  \\
    &\lstick[1]{$\ket{y_{e^\prime,w}}$}\setwiretype{n}& \setwiretype{q} &         & \targX{} &\\
    &\lstick[1]{               }\setwiretype{n}& \setwiretype{b} &         &          & \\
    &\lstick[1]{               }\setwiretype{n}& \setwiretype{q}\qwbundle{4}&& \qwbundle{4} &\\
    &\lstick[1]{$\ket{f_r}$}    \setwiretype{n}& \ctrl{-2}\setwiretype{q} & & \ctrl{-3} & \\
    &\lstick[1]{$\ket{anc}$}    \setwiretype{n}& \ctrl{-3}\setwiretype{q} & & \ctrl{-4} &\\
\end{quantikz}\label{sup:eq:Uswap}
\end{equation}

Similarly as before, if no edge $e$ or $e^\prime$ is active, then $y_{e,n}$ and $y_{e^\prime,n}$ are zero for any node $n$, thus the SWAP interchanges zeros. Hence $U^\text{swap}_{r: e \leftrightarrow e^\prime}$ behaves as the identity.

\subsubsection*{Satisfaction of Lemma~\ref{lem:synchronized_rotations} conditions}

Now that the partial mixer construction has been laid out, it remains to argue that $U^\text{swap}_{r: e \leftrightarrow e^\prime}$ and $U^\text{path}_{r: e \leftrightarrow e^\prime}$ do not violate the conditions of Lemma~\ref{lem:synchronized_rotations}. By constructing the unitaries as a direct implementation of the variable updates that need to occur after an edge rotation, we have ensured that $\ket{\bf{y^\prime}}=U^\text{swap}_{r: e \leftrightarrow e^\prime}U^\text{path}_{r: e \leftrightarrow e^\prime}\ket{\bf{y}}$, as well as $\ket{\bf{y}}=U^\text{swap}_{r: e \leftrightarrow e^\prime}U^\text{path}_{r: e \leftrightarrow e^\prime}\ket{\bf{y^\prime}}$ because of the built in symmetry.

It is however easier to show that the hypothesis of Corollary~\ref{cor:syncrhonized_rotations} are fulfilled. The unitary $U^\text{swap}_{r: e \leftrightarrow e^\prime}$ consists of a product of SWAP gates, which are Hermitian, acting on disjoint qubit pairs $\ket{y_{e, w}},\ket{y_{e^\prime, w}}$. Hence $(U^\text{swap}_{r: e \leftrightarrow e^\prime})^2=I$, and thus is Hermitian. The Hermiticity of $U^\text{path}_{r: e \leftrightarrow e^\prime}$ is also argued in a similar way. By realising that any two terms $A_{e^{\prime\prime}}:=U_{g_{e^{\prime \prime},  r: e \leftrightarrow e^\prime}} \left( \prod_{w \in \VV \setminus\{u, u^\prime, n_0\}} U_{h_{w,  r: e \leftrightarrow e^\prime}}  \cdot \mathrm{CCNOT}(\ket{g_{e^{\prime\prime}}}, \ket{h_{w}},\ket{y_{e^{\prime\prime}, w}}) U_{h_{w,  r: e \leftrightarrow e^\prime}} \right)  (U_{g_{e^{\prime \prime},  r: e \leftrightarrow e^\prime}})^{-1}$ in the product among edges commute,  it is possible to rearrange $(U^\text{path}_{r: e \leftrightarrow e^\prime})^2 = (\prod_{e^{\prime \prime} \in \EE \setminus\{e, e^\prime\}} A_{e^{\prime\prime}})^2=\prod_{e^{\prime \prime} \in \EE \setminus\{e, e^\prime\}} A_{e^{\prime\prime}}^2$ into a product of squares. Then, each $A_{e^{\prime\prime}}^2$ term becomes the identity, as the terms $U_{h_{w,  r: e \leftrightarrow e^\prime}}  \cdot \mathrm{CCNOT}(\ket{g_{e^{\prime\prime}}}, \ket{h_{w}},\ket{y_{e^{\prime\prime}, w}}) U_{h_{w,  r: e \leftrightarrow e^\prime}}$ also commute among them. Thus Hermiticity is shown. 

Finally, $U^\text{swap}_{r: e \leftrightarrow e^\prime}$ and $U^\text{path}_{r: e \leftrightarrow e^\prime}$ mutually commute as the set of qubits affected by the first is only used in the second to compute $h_w$, which remains invariant under the swap.

\subsection{Resource Estimation}
\label{sec:ressources_pm}

In this section we produce an estimation of the quantum resources needed to implement the partial mixer shown in \ref{sec:partial_mixers}.
We do not look for optimal circuit compilations, as those are hardware dependent. Instead, we concentrate on decomposing the mixer into arbitrary single qubit gates and CNOTs, which form an universal set of gates~\cite{barenco1995gates}, an approach that is hardware agnostic. We follow a procedure similar to the one in Ref.~\cite{stollenwerk_toward_2020}.

Let $\mathcal{N}_S$ be the number of single qubit gates, $\mathcal{N}_C$ the number of CNOT gates, and $\mathcal{N}_Q$ the number of qubits needed.

The most straightforward one to compute is $\mathcal{N}_Q$. We use one qubit for each variable $y_{e,n}$, which makes up for $|\EE|(|\VV|-1)$. 
Moreover, from (\ref{sup:eq:partialmixer}) one can already observe that six qubits more are needed: three where the control information $f_r$, $h_w$, $g_{e^{\prime\prime}}$ is stored (repeteadly computed and un-computed), one needed to implement the mixing step and two more used to compute $g_{e^{\prime\prime}}$). This makes up for a total of $|\EE|(|\VV|-1)+6$. To get $\mathcal{N}_Q$ just remains to count how many more ancillas will be needed to decompose everything in terms of single qubits and CNOTs, in what follows, which will only add a small offset factor. So we can conclude $\mathcal{N}_Q=\mathcal{O}(|\EE||\VV|)$.

For the other two quantities, we need to carefully decompose each unitary in (\ref{sup:eq:partialmixer}). It will be helpful to use the tuple $\mathcal{R} = (\mathcal{N}_S, \mathcal{N}_C)$ to keep track of the resources needed.

Let's start with $U_{f_{r: e \leftrightarrow e^\prime}}$, a zero-controlled Toffoli gate. As controlling on zero is equivalent to a normal control surrounded by two X gates, we have
\begin{equation}
\begin{quantikz}[scale=0.50, classical gap=0.1cm]
    & \wireoverride{n} \lstick[1]{$\ket{y_{e^\prime, u}} \,$ }  & \ctrl[open]{1} & \ghost{X}\\
    &\wireoverride{n} \lstick[1]{$\ket{y_{e, u^\prime}} \, $}  & \ctrl[open]{1} & \ghost{X}  \\
    &\lstick[1]{$\ket{0}$ }\wireoverride{n}& \gate[1]{X} & \rstick[1]{$\ket{f_{r}}$}  
\end{quantikz} = 
\begin{quantikz}[scale=0.50, classical gap=0.1cm]
    &\gate[1]{X}& \ctrl{1} &\gate[1]{X} &\\
    &\gate[1]{X}& \ctrl{1} & \gate[1]{X} & \\
    && \gate[1]{X} & &   
\end{quantikz}.
\end{equation}
In turn, the 3-Toffoli gate can be decomposed into CNOT and 1-qubit gates as follows:
\begin{equation}
\begin{quantikz}[scale=0.50, classical gap=0.1cm]
    & \ctrl{1} & \ghost{T}\\
    & \ctrl{1} &  \ghost{T^\dagger} \\
    & \gate[1]{X} &  \ghost{T^\dagger}
\end{quantikz} = 
\begin{quantikz}[scale=0.50, classical gap=0.1cm]
    & & & & \ctrl{2} & &&& \ctrl{2} & & \ctrl{1} & \gate{T} & \ctrl{1} &\\
    & & \ctrl{1} & & & &\ctrl{1} & & &\gate{T} & \targ{} & \gate{T^\dagger} & \targ{} &\\
    &\gate{H}& \targ{} & \gate{T^\dagger} & \targ{} & \gate{T} & \targ{}& \gate{T^\dagger} &\targ{} & \gate{T} & \gate{H} &&&
\end{quantikz},
\label{sup:eq:3toffoli} 
\end{equation}
where $T=\exp(-i\frac{\pi}{8} Z)$. This gives a count of 6 CNOT gates (which are necessary~\cite{shende2009cnot}) and 9 single qubit gates. Thus, $\mathcal{R}(U_{f_{r: e \leftrightarrow e^\prime}}) = (4,0) + \mathcal{R}(3-\text{Toffoli}) = (4,0)+(9,6) = (13,6)$ and we note that no extra ancilla is required.

Turning now to the controlled $U^\text{path}_{r: e \leftrightarrow e^\prime}$, recalling its definition in equation (\ref{sup:eq:Upath}) we observe that we need to compute the resources for the main 5-Toffoli gate, for $U_{g_{e^{\prime\prime}, r: e \leftrightarrow e^\prime}}$, and for $U_{h_{w, r: e \leftrightarrow e^\prime}}$, as well as determine how many times each of them is applied due to the product.

A 5-Toffoli can be decomposed into eight 3-Toffoli gates as follows~\cite{barenco1995gates}
\begin{equation}
\begin{quantikz}[scale=0.50, classical gap=0.1cm]
    \lstick[1]{$\ket{y_{e^{\prime \prime},w}}$} &\gate{X} &\midstick[7, brackets=none]{$=$}&\targ{}  &\ghost{X}&         &         &\targ{}  &         &         &         &\\
    \lstick[1]{$\ket{anc_{\text{extra1}}}$}     &         &                                &\ctrl{-1}&\targ{}  &         &\targ{}  &\ctrl{-1}&\targ{}  &         &\targ{}  &\\
    \lstick[1]{$\ket{anc_{\text{extra2}}}$}     &         &                                &         &\ctrl{-1}&\targ{}  &\ctrl{-1}&         &\ctrl{-1}&\targ{}  &\ctrl{-1}&\\
    \lstick[1]{$\ket{g_{e^{\prime \prime}}}$}   &\ctrl{-3}&                                &\ctrl{-2}&         &         &         &\ctrl{-2}&         &         &         &\\
    \lstick[1]{$\ket{h_w}$}                     &\ctrl{-1}&                                &         &\ctrl{-2}&         &\ctrl{-2}&         &\ctrl{-2}&         &\ctrl{-2}&\\
    \lstick[1]{$\ket{f_r}$}                     &\ctrl{-1}&                                &         &         &\ctrl{-3}&         &         &         &\ctrl{-3}&         &\\
    \lstick[1]{$\ket{anc}$}                     &\ctrl{-1}&                                &         &         &\ctrl{-1}&         &         &         &\ctrl{-1}&         &
\end{quantikz},
\label{sup:eq:5toffoli} 
\end{equation}
and hence $\mathcal{R}(5-\text{Toffoli}) = 8 \cdot \mathcal{R}(3-\text{Toffoli}) = 8\cdot (9,6) = (72, 48)$. Note two extra ancillas are required.
For $U_{g_{e^{\prime\prime}, r: e \leftrightarrow e^\prime}}$ we transform the zero-controls from (\ref{eq:circ_u_g_e_prime_prime}) into normal controls,
\begin{equation}
\begin{quantikz}[scale=0.5, classical gap=0.1cm]
         \lstick[1]{$\ket{y_{e^{\prime \prime}, u}} \, \,$ }& & \ctrl[closed]{1}  & \ctrl[open]{1} & &\ghost{X} &\\ 
         \lstick[1]{$\ket{y_{e^{\prime \prime}, u^\prime}} \, \,$} & & \ctrl[open]{1} & \ctrl[closed]{2} &  &\ghost{X} &\\ 
        \lstick[1]{$\ket{\text{0}}\, \,$} & & \gate{X} &  & \ctrl[closed]{2} & &\\
        \lstick[1]{$\ket{\text{0}}\, \,$}& & & \gate{X} & & \ctrl[closed]{1} & \\
        \lstick[1]{$\ket{\text{0}}\, \,$}& & & & \gate{X} & \gate{X} &  \rstick[1]{$\ket{g_{e^{\prime \prime}}}$}
\end{quantikz} = 
\begin{quantikz}[scale=0.5, classical gap=0.1cm]
         & & \ctrl{1}& \gate{X}   & \ctrl{1} & \gate{X} & &\\ 
         & \gate{X} & \ctrl{1} & \gate{X} & \ctrl{2} &  & &\\ 
         & &  \gate{X} &  & &\ctrl{2} & &\\
         & & & & \gate{X} & &\ctrl{1} & \\
         & & & && \gate{X} & \gate{X} &
\end{quantikz}
\end{equation}
and thus obtaining $\mathcal{R}(U_{g_{e^{\prime\prime}, r: e \leftrightarrow e^\prime}})= (4,2) + 2\cdot\mathcal{R}(3-\text{Toffoli}) = (4,2) + 2\cdot (9,6)=(22,14)$.
Regarding $U_{h_{w, r: e \leftrightarrow e^\prime}}$, its resources $\mathcal{R}(U_{h_{w, r: e \leftrightarrow e^\prime}})=(0,2)$ are straightforward to obtain from (\ref{eq:circ_u_h_w}).

To get the resources of the whole controlled $U^\text{path}_{r: e \leftrightarrow e^\prime}$, we need to sum over the products $\prod_{e^{\prime \prime} \in \EE \setminus\{e, e^\prime\}}$ and $\prod_{w \in \VV \setminus\{u, u^\prime, n_0\}}$:
\begin{align}
    \mathcal{R}(U^\text{path}_{r: e \leftrightarrow e^\prime}) &= (|\EE|-2)\left[2\mathcal{R}(U_{g_{e^{\prime\prime}, r: e \leftrightarrow e^\prime}}) + (|\VV|-3) \left(2\mathcal{R}(U_{h_{w, r: e \leftrightarrow e^\prime}})+\mathcal{R}(5-\text{Toffoli})\right)\right]\\
    & = (|\EE|-2) (44, 28) + (|\EE|-2)(|\VV|-3) ((0,4)+(72,48))\\
    & =  (|\EE|-2) (44, 28) + (|\EE|-2)(|\VV|-3) (72,52).
\end{align}
Note that this is an upper bound, as some unitaries can simplify in specific cases, such as when $e$ and $e^\prime$ are a multiedge, or if $u$ or $u^\prime$ is the root.

Regarding $U^\text{swap}_{r: e \leftrightarrow e^\prime}$, each of the CCSWAP terms in (\ref{sup:eq:Uswap}) can be decomposed as follows:
\begin{equation}
\begin{quantikz}[scale=0.50, classical gap=0.1cm]
\lstick[1]{$\ket{y_{e,w}}$}             &\swap{1} &\midstick[5,brackets=none]{$=$}&\ctrl{1} & \targ{} &\ctrl{1} &\\
\lstick[1]{$\ket{y_{e^\prime,w}}$}      &\targX{} &                               &\targ{}  &\ctrl{-1}& \targ{} &\\
\lstick[1]{$\ket{anc_{\text{extra3}}}$} &         &                               &\targ{}  &\ctrl{-1}& \targ{} &\\
\lstick[1]{$\ket{f_r}$}                 &\ctrl{-2}&                               &\ctrl{-1}&         &\ctrl{-1}&\\
\lstick[1]{$\ket{anc}$}                 &\ctrl{-1}&                               &\ctrl{-1}&         &\ctrl{-1}&
\end{quantikz}
\end{equation}
Thus $\mathcal{R}(\mathrm{CCSWAP}) = (0,2)+3\cdot \mathcal{R}(3-\text{Toffoli}) = (0,2) + 3(9,6) = (27, 20)$. Taking into account the product $\prod_{w \in \VV \setminus\{u, u^\prime, n_0\}}$ this results in
\begin{equation}
    \mathcal{R}(U^\text{swap}_{r: e \leftrightarrow e^\prime}) = (|\VV|-3)(27,20),
\end{equation}
and requires one extra ancilla.

It remains to determine the resources required for the mixing rotation and to combine them with the previous contributions. From equation (\ref{sup:eq:partialmixer}), 4 Hadamard gates are required, and the controlled phase gate, which can be decomposed as follows:
\begin{equation}
    \begin{quantikz}[scale=0.50, classical gap=0.1cm]
    &\ctrl{1}&\ghost{P}\\
    &\gate{P(\beta)}&
    \end{quantikz} = 
    \begin{quantikz}[scale=0.50, classical gap=0.1cm]
    &&\ctrl{1}&\gate{P(\beta/2)}&\ctrl{1}&\\
    &\gate{P(\beta/2)}& \targ{}&\gate{P(-\beta/2)}&\targ{}&
    \end{quantikz}
\end{equation}
Hence $\mathcal{R}(\text{mixing rotation}) = 4\cdot\mathcal{R}(\mathrm{H})+ \mathcal{R}(\mathrm{CPhase}) = 4\cdot (1,0)+ (3,2) = (7,2)$.

We now have all the ingredients needed to determine the total resource requirements,
\begin{align*}
    \mathcal{R}(U^\text{PM}_{r: e \leftrightarrow e^\prime}(\beta)) &= \mathcal{R}(\text{mixing rotation}) + 2 \mathcal{R}(U_{f_{r: e \leftrightarrow e^\prime}}) + 2\mathcal{R}(U^\text{path}_{r: e \leftrightarrow e^\prime}) + 2\mathcal{R}(U^\text{swap}_{r: e \leftrightarrow e^\prime})\\
    &=(7,2) + 2(13,6) + 2 [(|\EE|-2) (44, 28) + (|\EE|-2)(|\VV|-3) (72,52)]+ 2 (|\VV|-3)(27,20)\\
    & = (33,14) + (|\EE|-2) (88, 56) + (|\EE||\VV|-2|\VV|-3|\EE|+6) (144,104)+ (|\VV|-3)(54,40)\\
    & = (559, 406)  - |\EE| (344, 256) - |\VV| (234,168) + |\EE||\VV|(144,104).
\end{align*}
The total number of qubits needed, by realising that either $\ket{anc_{\mathrm{extra1}}}$ or $\ket{anc_{\mathrm{extra2}}}$ can be used as $\ket{anc_{\mathrm{extra3}}}$, as they are employed at different points and are always uncomputed, is $\mathcal{N}_Q=|\EE|(|\VV|-1)+8$.

    \subsection{Full Mixer and Initial State Preparation}
\label{sec:full_mixer}

\begin{figure}[b!]
    \centering
    \includegraphics[width=\columnwidth]{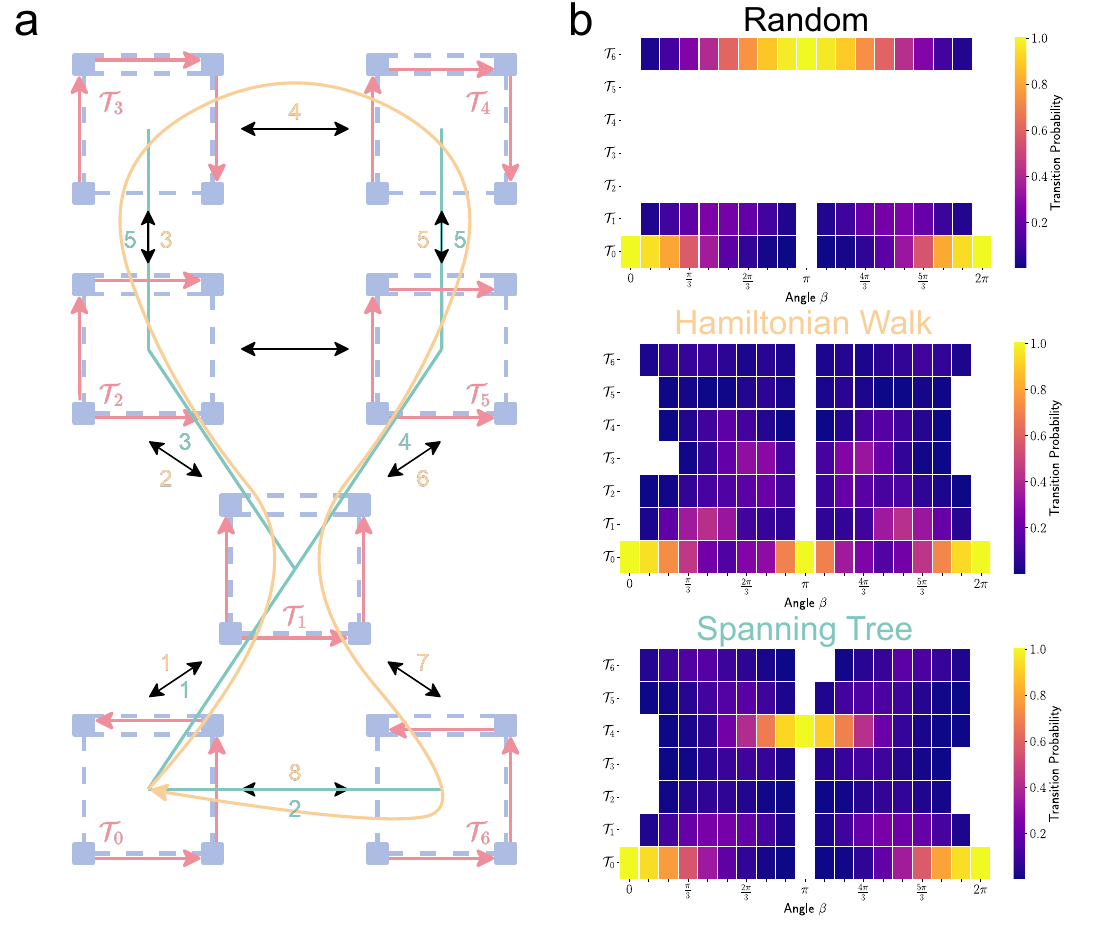}
    \caption{Comparison of Random-Order and graph-aware Mixer designs if applied once to an initial state $\ket{\TT_i}$. \textbf{a}: Graph of all spanning trees $\GG_{\mathrm{Sp}}$ for a simple four-node example, with multiedges. We set $\TT_0$ as the initial state/configuration of our algorithm. Then, the yellow path represents a Hamiltonian Walk, and the green path is a Minimum Spanning Tree in $\GG_{\mathrm{Sp}}$, which allows us to explore the entire graph $\GG_{Sp.}$. \textbf{b}. Transition probabilities $\TT_0 \to \TT_k$ for the three presented Mixer designs in dependence on the variational parameter $\beta$. For the Random-Order Mixer, we representatively use the sequence of edge rotations that corresponds to the sequence $\TT_1 \leftrightarrow \TT_2,  \, \TT_1 \leftrightarrow \TT_5, \, \TT_1 \leftrightarrow \TT_6, \,  \TT_0 \leftrightarrow \TT_1, \, \TT_3 \leftrightarrow \TT_4$ and $\TT_2 \leftrightarrow \TT_3 / \TT_4 \leftrightarrow \TT_5$, which are both implemented by the same edge rotation. since they are mulitedges. The Hamiltonian-Walk and Spanning-Tree Mixer uses the sequence shown in \textbf{a}. Note again that for the Spanning-Tree Mixer, the transitions $\TT_2 \leftrightarrow \TT_3$ and $\TT_4 \leftrightarrow \TT_5$ are based on the same edge-rotation, so only one Partial Mixer is needed for this level of the Tree. 
    Probabilities are approximated by sampling the Mixer circuits 1000 times using Qiskit-Aer with the ``matrix\_product\_state'' method, cf. \cite{vidal2003efficient}.}  
    \label{sup:fig:full_mixers}
\end{figure}

Each partial Mixer implements a single edge rotation; therefore, a full Mixer can be constructed as a composition of partial Mixers. The key question is: which composition should we choose? The answer depends intrinsically on the chosen initial state, i.e., some superposition of feasible states $\ket{\bf{y}_i}$ that encodes the tree $\TT_i$.

By analogy with standard QAOA, where the algorithm is initialized in the uniform superposition of all bit strings, one natural choice would be to start in the uniform superposition of all feasible states—in this case, all bit strings encoding a spanning tree rooted at $n_0$. However, preparing such a superposition is nontrivial; in the worst case, it requires an exponential-size circuit \cite{Plesch2011}\footnote{A detailed discussion of whether efficient methods for preparing feasible superpositions exist is beyond the scope of this work, but would be an interesting direction for future research.}. Nevertheless, a straightforward full Mixer design applicable in this situation is a Random-Order-Mixer
\begin{equation}
\label{eq:random_order_mixer}
    U_{\text{M, feasible}}(\beta, \sigma) = \prod_{r} U^\text{PM}_{ \sigma(r)}(\beta)
\end{equation}
where $\sigma$ is an element of the permutation group $S_{\lvert \RR_{\mathrm{SWAP}} \rvert}$ where $\RR_{\mathrm{SWAP}}$ is the set of all edge swaps. Hence, we apply all edge swaps in a random order. 

Another approach would be to initialize the algorithm in any of the feasible states. This initialization has two advantages. First, such a state can be efficiently prepared with $h$ Pauli-$X$ gates, where $h$ is the hamming weight of the state (cf. Fig.~\ref{sup:fig:mixers_simple_example}b). Second, initializing the circuit in one feasible state is necessary for Rev-QAOA (cf. the discussion in \ref{sec:qaoa_simulation_methods}). 

However, applying the Random-Order Mixer~\eqref{eq:random_order_mixer} (once) to a feasible state does not guarantee transitions to all other feasible states. The reason is that the order of edge swaps matters: if an intermediate configuration (spanning tree) is reached in which the next edge swap is invalid, the corresponding partial Mixer acts as the identity, effectively skipping that operation. As a result, certain regions of the spanning-tree graph $\GG_{\mathrm{Sp}}$ (cf. Sec.~\ref{sec:edge_rot_graph_Def}) may remain unexplored. Consider the example in Fig.~\ref{sup:fig:full_mixers}a. If the algorithm is initialized in the state describing $\TT_0$ and the first two edge rotations in the permutation $\sigma$ correspond to transitions $\TT_1 \leftrightarrow \TT_2$ and $\TT_1 \leftrightarrow \TT_5$, then configurations $\TT_2, \, \TT_3, \, \TT_4$ and $\TT_5$ cannot be reached; the corresponding transition probabilities are zero (cf. upper panel in Fig.\ref{sup:fig:full_mixers}b).

There is, however, a straightforward solution. We choose an order of edge swaps such that the full graph $\GG_{Sp.}$ is explored from this initial configuration. We present two strategies: the Hamiltonian-Walk Mixer and the (Minimum-) Spanning-Tree Mixer. 

First, let's consider a Hamiltonian walk on the graph $\GG_{\mathrm{Sp}}$ starting (and ending) at $\TT_0$, the spanning tree corresponding to the initial state (cf. the yellow path in Fig.~\ref{sup:fig:full_mixers}a).  We then choose the sequence of edge rotations corresponding to the transition between $\TT_k$ in the Hamiltonian walk. The resulting Hamiltonian-Walk Mixer then has finite transition probabilities from $\TT_0$ to all feasible states for most values of $\beta$ (cf. middle panel in Fig.~\ref{sup:fig:full_mixers}b). Depending on the graph $\GG_{\mathrm{Sp}}$, the Hamiltonian-Walk Mixer requires more or fewer edge swaps than the Random-Order Mixer. For the example in Fig.~\ref{sup:fig:full_mixers}a, we need a sequence of 8 edge swaps, which is one more than $\lvert \RR_{\mathrm{SWAP}} \rvert$. Note that, for example, the transitions 3 and 5 correspond to the same edge swap. However, for the graph depicted in Fig.~\ref{fig:graph_of_spanning_trees}, the Hamiltonian-Walk Mixer can be, depending on the choice of walk, more efficient than the Random-Order Mixer.  However, we note that finding the shortest Hamiltonian Walk is again an NP-hard Problem. 

Second, we can build a Mixer based on a (Minimum-) Spanning Tree on $\GG_{\mathrm{Sp}}$. We again choose a sequence of edge rotations according to the distance to the root in the Minimum Spanning Tree (cf. the green tree in Fig.~\ref{sup:fig:full_mixers}a). Again, this Mixer design has finite transition probabilities from $\TT_0$ to all feasible states (cf. lower panel in Fig.~\ref{sup:fig:full_mixers}b). Spanning-Tree Mixers require fewer edge swaps to be implemented and thus provide the shallowest full Mixer circuits. Moreover, finding a Minimum Spanning Tree can be achieved in $\mathcal{O} (\lvert \EE_{\mathrm{Sp}} \rvert \log \lvert \VV_{\mathrm{Sp}} \rvert)$. 

It is important to note that any of these full-Mixer designs must ultimately be evaluated in the context of a concrete optimization algorithm, such as QAOA. 
In particular, their performance depends on how effectively they mix across all feasible configurations when applied repeatedly within the algorithm. 
The discussion above focused on transition probabilities starting from the state $\mathcal{T}_0$. 
However, during the execution of an optimization algorithm, the intermediate states are superpositions of feasible configurations, which in turn affects the occupation probabilities after the application of a Mixer.

\subsection{Implementation of the Partial Mixers for a Simple Example}
\label{sec:simple_example_mixer}

To ``validate'' the Partial Mixer design and show when simplifications arise we now discuss the implementation of the two Partial Mixers $U^\text{PM}_{r: 0 \leftrightarrow 1^\prime}(\beta)$ for the simple three nodes and three edge example discussed in the main document. The simplified Partial Mixer circuits can be found in Fig.~\ref{sup:fig:mixers_simple_example}a. The indices $e,n$ are flattened by $j=e \lvert \VV - 1 \rvert + (n-1)$. For this minimal example, there are only two edge swaps $r: 0 \leftrightarrow 1$ and $r: 1 \leftrightarrow 2$. Hence, we only need two partial mixers. Since the underlying problem is quite simple, the Boolean conditions $f_{r, w, e^{\prime \prime}}$ simplify. Fewer controlled operations and ancillary qubits are necessary to implement the Boolean functions. Furthermore, for each edge swap, there is only one $w, e^{\prime \prime}$, that is, we don't need products in these variables. Hence, the Boolean functions need to be evaluated only once. Their computation in the ancillary can be moved outside the synchronized rotation. The resulting partial mixers are thus more tractable. Here, the simplest full mixer design is a sequential application of the partial mixers ( cf. Fig.~\ref{sup:fig:mixers_simple_example}b). 

Simulating the full mixer circuit in Qiskit for an ideal quantum computer with no noise, we see that the feasible space is protected. Starting from a feasible configuration such as $\ket{100001}$, we have non-zero occupation possibilities only for the three feasible states $\{\ket{10001}, \ket{001011}, \ket{110100}\}$. However, for $\beta = \frac{\pi}{2}$ we do not get a uniform superposition of the feasible states. The probabilities depend on the Mixer Ansatz, such as the order of the partial mixers.  

\begin{figure}[h!]
    \centering
    \includegraphics[width=0.9\columnwidth]{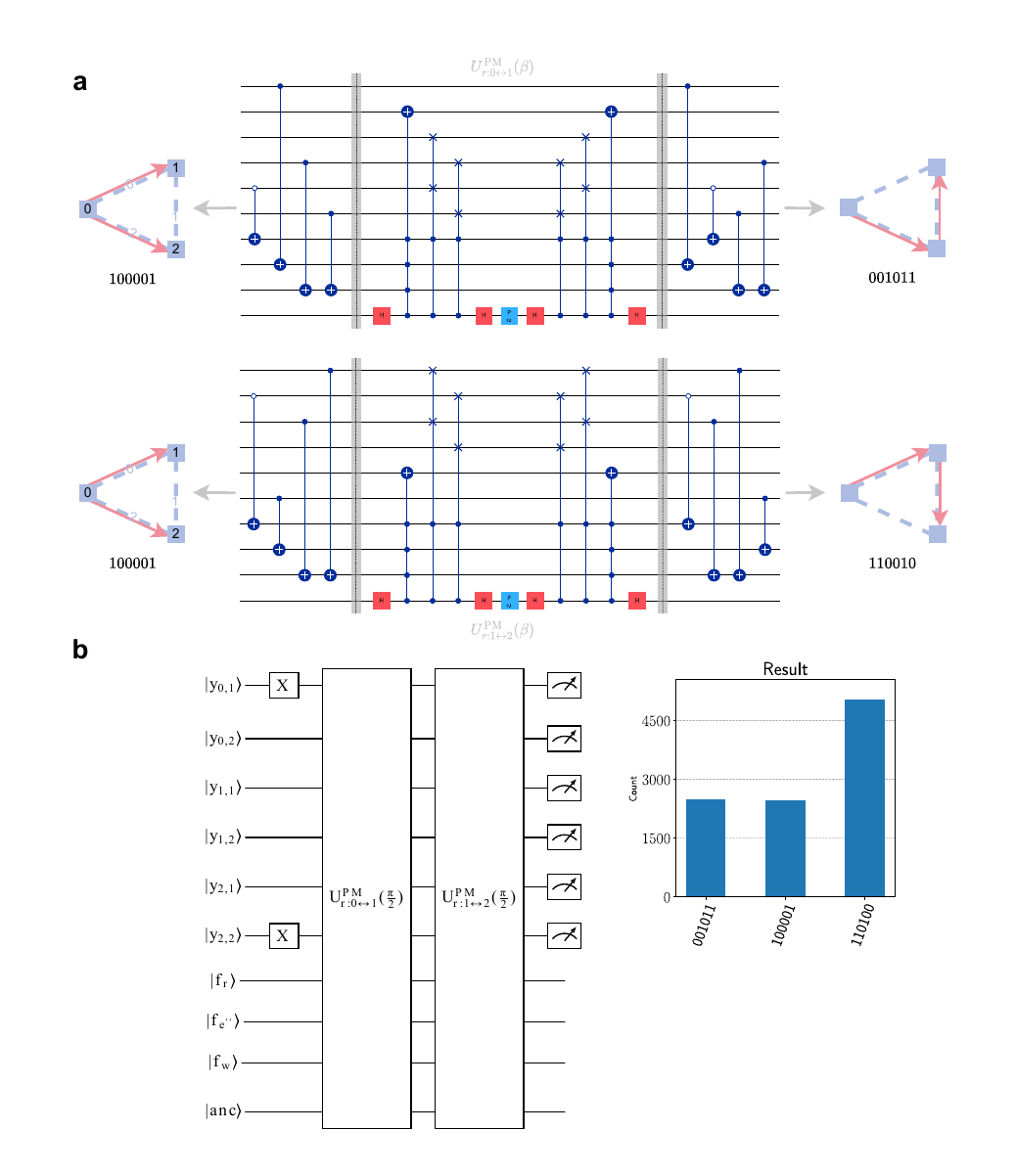}
    \caption{Edge rotation based mixer for the simple Example from Fig. . \textbf{a}: Implementation of the two partial mixers $U^\text{PM}_{r: 0 \leftrightarrow 1^\prime}(\beta)$ and $U^\text{PM}_{r: 1 \leftrightarrow 2^\prime}(\beta)$ in Qiskit. Since the example is small, fewer ancillary variables and controlled operations are needed than for the general construction. \textbf{b:} Verification that the full mixer $U^\text{PM}_{r: 0 \leftrightarrow 1^\prime}(\beta) U^\text{PM}_{r: 0 \leftrightarrow 1^\prime}(\beta)$ protects the feasible space, here for $\beta = \frac{\pi}{2}$. The circuit shows the experiment conducted in Qiskit, assuming an ideal Quantum computer with no noise. First, the initial state $\ket{100001}$ is prepared using two $X$-gates, and then the two partial mixers are sequentially applied before the qubits $\ket{y_{e,n}}$ are measured. The outcome of the experiment in Qiskit is shown as a histogram. Only feasible states have non zero probability.    }  
    \label{sup:fig:mixers_simple_example}
\end{figure}

\newpage
\section{The Minimum Dissipation Spanning Tree Problem}
\label{sec:mdst}

In this section, we discuss the Minimum Dissipation Spanning Tree (MDST) problem from a theoretical perspective.
In \ref{sec:sup:comp_hardness}, we first present a short NP-hardness proof that is intended to be accessible to readers unfamiliar with the topic and to offer an instructive example of a polynomial-time reduction.
We then proceed to show that MDST is NP-hard to approximate, a new result that provides further motivation for the study of (quantum) heuristic algorithms for MDST. 
In \ref{sec:sup:red_graph_deg_2} we discuss a simplification scheme for instances of MDST. 
Finally, we give a mixed-integer programming (MIP) formulation in \ref{sec:sup:mip_formulation}.

\subsection{Computational Hardness}
\label{sec:sup:comp_hardness}

\subsubsection*{A Brief Introduction to Complexity Theory}

In computational complexity theory (cf.\ the book by Arora and Barak~\cite{arora2009computational} for a comprehensive introduction), an algorithm is considered efficient if it runs in \emph{polynomial time}: for an input of size~$n$, the algorithm terminates after at most~$cn^d$ basic operations, where~$c$ and~$d$ are constants.
The available basic operations, e.g., fundamental arithmetic computations and simple control steps, depend on the computational model, such as the \emph{Turing machine}. However, all reasonable classical models are polynomially equivalent and we omit details here. The input size~$n$ can be measured in any reasonable way; for instance, if the input is a graph, the number of nodes plus edges is a natural choice.
Although the constants~$c$ and~$d$ may be large, they are typically small enough for satisfactory performance in practice.

A \emph{problem} is a collection of \emph{instances}: for example, each Sudoku puzzle is an instance, while the set of all Sudoku puzzles constitutes the Sudoku problem.
An algorithm solves a problem if it produces the correct answer on every instance.
Some problems are known not to be solvable in polynomial time, such as determining the winning player in (generalized) chess~\cite{fraenkel1981computing}, and are therefore considered intractable: they behave poorly and are difficult to deal with.
However, for many important practical problems the situation is not so clear: they are merely conjectured not to be solvable in polynomial time.
Most of these conjectures are implied by the widely-believed conjecture that~$\mathrm{P} \neq \mathrm{NP}$, also known as the~$\mathrm{P}$ versus~$\mathrm{NP}$ problem.

The symbols~$\mathrm{P}$ and $\mathrm{NP}$ each denote a class of \emph{decision problems}, i.e., problems that ask yes--no questions, such as ``Does this Sudoku puzzle have a valid solution?'' or ``Does this electrical network have a configuration producing a power loss of at most~$k$?''
The class~$\mathrm{P}$ contains all decision problems that can be solved in polynomial time.
The class~$\mathrm{NP}$ is more elusive, with the symbol standing for ``nondeterministic polynomial''.
Loosely speaking,~$\mathrm{NP}$ contains all decision problems that allow to check solutions for correctness in polynomial time: it might be challenging to solve a Sudoku puzzle, but checking whether a solution is correct is manageable.
Hence, even though it is unclear whether (generalized) Sudoku is in~$\mathrm{P}$ , it is included in~$\mathrm{NP}$~\cite{yato2003complexity}.

There are problems in~$\mathrm{NP}$ that appear to be particularly hard.
To compare the hardness of problems in~$\mathrm{NP}$, the following concept is crucial: a \emph{polynomial-time Karp reduction} from decision problem~$A$ to decision problem~$B$ is an algorithm that, given any instance~$\mathcal{I}$ of~$A$, produces an instance~$\mathcal{I}'$ of~$B$ in polynomial time such that~$\mathcal{I}$ is a ``yes''-instance if and only if~$\mathcal{I}'$ is a ``yes''-instance.
If we have such a reduction, then we can conclude that problem~$B$ is at least as hard as problem~$A$, apart from the overhead of polynomial translation, which is of secondary importance when dealing with problems that presumably require exponential time to solve.

If a problem is, in the above sense, at least as hard as all problems in~$\mathrm{NP}$, then it is \textit{$\mathrm{NP}$-hard}.
If a problem is~$\mathrm{NP}$-hard and contained in~$\mathrm{NP}$, then it is \textit{$\mathrm{NP}$-complete}.
The most intriguing feature of NP-complete problems is that a polynomial-time algorithm for a single one of them would yield a polynomial-time algorithm for every problem in NP: we could simply Karp-reduce every problem in~$\mathrm{NP}$ to the one $\mathrm{NP}$-complete problem that we know how to solve in polynomial time.
This would contradict the~$\mathrm{P} \neq \mathrm{NP}$ conjecture, and hence we believe that no $\mathrm{NP}$-complete problem can be solved in polynomial time.

Many very general and highly important problems are $\mathrm{NP}$-complete.
The first natural problem shown to be NP-complete is Boolean Satisfiability.
This was achieved by Cook~\cite{cook1971complexity} and Levin~\cite{levin1973universal} in the early 1970s and sparked the discovery of many more $\mathrm{NP}$-complete problems: to prove a problem~$\Pi$ in~$\mathrm{NP}$ to be $\mathrm{NP}$-complete, provide a polynomial-time Karp reduction from a problem already known to be $\mathrm{NP}$-complete to~$\Pi$.
Fundamental decision problems like Boolean Satisfiability, Integer Linear Programming, Traveling Salesperson, Vertex Cover, and Subset Sum were shown to be $\mathrm{NP}$-complete by Karp-reductions~\cite{garey1979computers}.
The aforementioned Sudoku problem is $\mathrm{NP}$-complete as well~\cite{yato2003complexity}.
In general, because $\mathrm{NP}$-complete problems appear hard to solve yet have easily verifiable solutions, they are closely associated with puzzles.

\subsubsection*{Proof of NP-hardness for MDST}

We show that the MDST problem is NP-hard, illustrating the general principle of Karp reductions.
We do not claim originality; our goal is a clear, self-contained presentation.
Presumably, the ideas we use here first appeared in a Japanese-language master's thesis by Shion Chiba, which is not accessible to us (cf.~Ref.~\cite{ito2024loss}).

As is standard in hardness reductions, we restrict numbers (demands and dissipation constants) to integers.
If the problem is hard in this setting, then it remains hard in the general case.
Here, we work with the decision version of MDST.
If the decision version is hard, then the optimization version is also hard.

\decisionproblem{Minimum Dissipation Spanning Tree (Decision Version)}{prob:mdst}
{A connected undirected graph~$\GG = (\VV,\EE)$, a \emph{demand}~$\mathfrak{f}_v \in \mathbb{Z}$ for every node~$v \in \VV$ with~${\sum_{v \in \VV} \mathfrak{f}_v = 0}$, a \emph{dissipation constant}~$\alpha_e \in \mathbb{N}_0$ for every edge~$e \in \EE$, and an integer~$k \in \mathbb{N}_0$.}
{Is there a spanning tree~$\TT$ of~$\GG$ such that~$C(\TT) \leq k$? (See the definition of~$C$ below.)}

The cost function~$C$ is given by~$C(\TT) = \sum_{e \in \TT} \alpha_e f_e(\TT)^2$ with the flow $f_e(\TT)$ over edge~$e$ uniquely determined by Kirchhoff’s current law (flow conservation) for every~$e \in \EE$ (cf.~Eq.~\eqref{net_rec:eq:KCL}).
Here, we may pick a root node arbitrarily, as we do not explicitly require it.

\begin{thm}
    The Minimum Dissipation Spanning Tree problem is NP-hard.
\end{thm}

We reduce from the NP-hard \cite{garey1979computers} \textsc{Partition} problem to MDST.
Given a nonempty set~$S$, nonempty subsets~$A, B \subseteq S$ form a partition of~$S$ if~$A \cup B = S$ and~$A \cap B = \emptyset$.

\decisionproblem{Partition}{prob:partition}
{A set~$S$ of positive integers.}
{Is there a partition of~$S$ into nonempty subsets~$A$ and~$B$ with equal sum, i.e.,~$\sum_{s \in A} s = \sum_{s \in B} s$?}

Let us consider an instance of \textsc{Partition} with set~$S$.
The idea is to build an MDST instance with one source node and, for each~$s \in S$, one consumer node such that any low-cost spanning tree must connect each consumer to the source in exactly one of two ways: either via a node~$a$ representing set~$A$, or via a node~$b$ representing set~$B$. A spanning tree then encodes a partition of~$S$. 
The precise construction is as follows.

\begin{figure}[t]
    \centering
    \includegraphics[width=0.8\columnwidth]{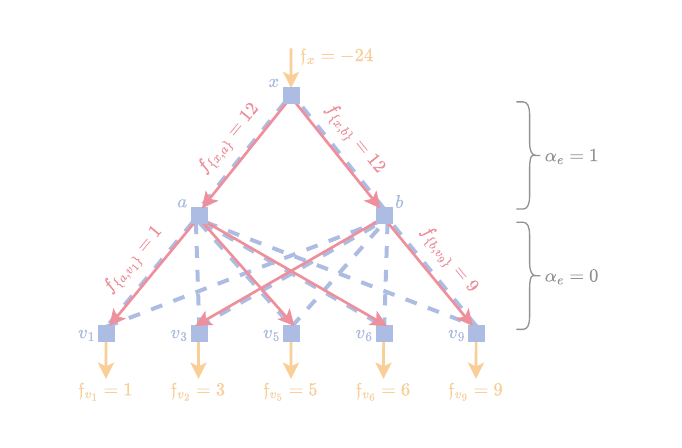}
    \caption{Example of an MDST-instance constructed by Mapping~\ref{cons:partition2mdst}. The depicted instance results from a \textsc{Partition}-instance with integer set~$S=\{1,3,5,6,9\}$. The nodes and edges of the graph~$\GG$ are blue, and the demands are yellow.
    The dissipation constants of the two upper edges~($\{x,a\}$ and~$\{x,b\}$) are~$1$, and the dissipation constants of the lower edges are~$0$. 
    In addition, the figure shows a solution (a spanning tree) to the MDST-instance in red, including some of the flow values over active edges, where~$x$ is taken as the root node. This solution corresponds to partitioning~$S$ into the sets~$A = \{1,5,6\}$ and~$B = \{3,9\}$.}
    \label{sup:fig:np-hardness}
\end{figure}

\begin{cons}\label{cons:partition2mdst}
    Given any instance~$\mathcal{I}$ of \textsc{Partition}, we construct an instance~$\mathcal{I}'$ of MDST (see Fig.~\ref{sup:fig:np-hardness}).
    The instance~$\mathcal{I}$ consists of a set~$S$ of positive integers, and~$\mathcal{I}'$ consists of a graph, demands, dissipation constants, and an integer threshold~$k$.
    Let the graph~$\GG = (\VV, \EE)$ contain nodes named~$x$, $a$, $b$, and, for each~$s \in S$, a node named~$v_s$.
    Include the edges~$\{x,a\}$, $\{x,b\}$, and, for each~$s \in S$, the edges~$\{a,v_s\}$ and~$\{b,v_s\}$.
    Set the demands of node~$a$ and node~$b$ to zero, i.e., set~$\mathfrak{f}_a \coloneqq 0$ and $\mathfrak{f}_b \coloneqq 0$.
    For each~$s \in S$, set the demand of~$v_s$ to~$s$, i.e., set~$\mathfrak{f}_{v_s} \coloneqq s$.
    Let~$Q \coloneqq \sum_{s \in S} s$ be an auxiliary value holding the sum of the elements of~$S$.
    Set~$\mathfrak{f}_x \coloneqq -Q$, which balances power injection with power consumption.
    Moreover, set the dissipation constants of both~$\{x,a\}$ and~$\{x,b\}$ to one~($\alpha_{\{x,a\}} \coloneqq 1$ and $ \alpha_{\{x,b\}} \coloneqq 1$), and the dissipation constants of all other edges to zero.
    Finally, set~$k \coloneqq \lfloor Q^2 / 2 \rfloor$, where~$\lfloor \cdot \rfloor$ denotes the operation of rounding down to the largest integer not exceeding the argument value.
    (Note that~$\lfloor \cdot \rfloor$ does not modify the argument value if it already is an integer.)
\end{cons}

The construction runs in polynomial time.
It remains to show that~$\mathcal{I}$ and~$\mathcal{I}'$ are equivalent.
Before the formal proof, we informally discuss the network produced by the construction and intuitively argue how the theorem follows. We assume~$Q > 0$ in the following.

The node~$x$ is the single source of flow, whereas the~$v_s$-nodes are consumers of flow. The nodes~$a$ and~$b$ only serve as transit nodes, neither supplying nor consuming flow.
The flow originating from~$x$ can split on its way to the consumers: some of the flow reaches the consumers through node~$a$ and some reaches them through node~$b$. Since most edges have a dissipation constant of zero, with the edges~$\{x,a\}$ and~$\{x,b\}$ being the only exceptions, the division of the flow between node~$a$ and node~$b$ is crucial. As the dissipation on a single edge grows quadratically with the flow, it is best to spread the flow as evenly as possible between~$a$ and~$b$.
The threshold~$k$ is chosen such that the total dissipation is at most~$k$ only if the flows from~$x$ to~$a$ and from~$x$ to~$b$ are equal.
In a spanning tree including both edge~$\{x,a\}$ and edge~$\{x,b\}$, each consumer is adjacent to either $a$ or $b$. Equal flows thus imply a partition of consumers into two subsets with equal total demand -- directly corresponding to the \textsc{Partition} problem. Hence, the constructed MDST instance is equivalent to the original \textsc{Partition} instance. We now proceed to the formal proof.

\begin{lem}
    Let~$\mathcal{I}$ be an instance of \textsc{Partition} and let~$\mathcal{I}'$ be the corresponding instance of MDST (Mapping~\ref{cons:partition2mdst}).
    Then~$\mathcal{I}$ is a yes-instance if and only if~$\mathcal{I}'$ is a yes-instance.
\end{lem}
\begin{proof}
    ($\Rightarrow$) Assume that~$\mathcal{I}$ is a yes-instance.
    We show that~$\mathcal{I}'$ is a yes-instance.
    By assumption, there is a partition of~$S$ into disjoint subsets~$A$ and~$B$ with equal sum, meaning that~$\sum_{s \in A} s = \sum_{s \in B} s = Q / 2$.
    Hence,~$Q$ is an even number.
    Take the spanning tree~$\TT$ of~$\GG$ that includes the edges~$\{x,a\}$, $\{x,b\}$, and, for each~$s \in A$, the edge~$\{a,v_s\}$, and, for each~$s \in B$, the edge~$\{b,v_s\}$.
    Observe that~$\TT$ is indeed a spanning tree.
    Moreover, the amount of flow on edge~$\{x,a\}$ is~$Q/2$, and the amount of flow on edge~$\{x,b\}$ is also~$Q/2$.
    Only these two edges have nonzero resistance, so~$C(\TT) = 2 \cdot (Q/2)^2 = Q^2/2 = \lfloor Q^2 / 2 \rfloor = k$.
    Thus,~$\mathcal{I}'$ is a yes-instance.

    ($\Leftarrow$) Assume that~$\mathcal{I}'$ is a yes-instance.
    We show that~$\mathcal{I}$ is a yes-instance.
    By assumption, there is a spanning tree~$\TT$ of~$\GG$ with cost~$C(\TT) \leq \lfloor Q^2 / 2 \rfloor$.
    We first consider the case that~$\TT$ does not contain the edge~$\{x,a\}$.
    Then,~$\TT$ contains the edge~$\{x,b\}$.
    Since~$\{x,b\}$ is the only edge incident to node~$x$ in~$\TT$, we have~$f_{\{x,b\}}(\TT) = \mathfrak{f}_x = Q$.
    It follows that~$C(\TT) = Q^2$, a contradiction to~$C(\TT) \leq \lfloor Q^2 / 2 \rfloor$.
    Analogously, we get a contradiction if~$\TT$ does not contain~$\{x,b\}$.
    Hence, assume that~$\TT$ contains both~$\{x,a\}$ and~$\{x,b\}$.
    Let~$F_a \coloneqq f_{\{x,a\}}(\TT)$ and~$F_b \coloneqq f_{\{x,b\}}(\TT)$.
    We have~$F_a^2 + F_b^2 \leq \lfloor Q^2 / 2 \rfloor$.
    Moreover, as node~$x$ is the only source and injects~$Q$ units of flow, we additionally have~$F_a + F_b = Q$.
    Using the rearrangement~$F_b = Q - F_a$, we obtain~$F_a^2 + (Q - F_a)^2 \leq \lfloor Q^2 / 2 \rfloor$.
    We can verify that~$F_a^2 + (Q - F_a)^2 = Q^2/2 + 2(F_a - Q/2)^2$ by algebraic simplification.
    Thus, we have~$Q^2 / 2 + 2(F_a - Q/2)^2 \leq \lfloor Q^2 / 2 \rfloor$.
    It follows that~$F_a - Q/2 = 0$ (and that~$Q$ is even).
    Hence, we have~$F_a = F_b = Q/2$.
    Since~$\{x,a\}$ and~$\{x,b\}$ are both contained in~$\TT$ and~$\TT$ is a spanning tree, we have that~$v_s$ is adjacent to either~$a$ or~$b$ in~$\TT$ for each~$s \in S$.
    Thus, the sets~$A \coloneqq \{s \mid \{a,v_s\} \in \TT\}$ and~$B \coloneqq \{s \mid \{b,v_s\} \in \TT\}$ form a partition of~$S$, and their sums are equal since~$F_a = F_b$.
    It follows that~$\mathcal{I}$ is a yes-instance.
\end{proof}

\subsubsection*{Hardness of Approximation for MDST}

The NP-hardness of MDST means that, unless~$\mathrm{P} = \mathrm{NP}$, no polynomial-time algorithm exists that finds an optimum solution on every instance.
A common approach is to allow solutions that are not necessarily optimal but can be proven to lie within a bounded distance from the optimum, giving rise to the field of approximation algorithms.

Let~$\mathcal{I}$ be an instance of some minimization problem $\Pi$.
We write~$\OPT(\mathcal{I})$ to refer to the cost of an optimal solution to~$\mathcal{I}$.
A $\rho$-approximation algorithm (with~$\rho \geq 1$) for $\Pi$ is a polynomial-time algorithm that, given any instance~$\mathcal{I}$ of~$\Pi$, produces a solution of cost at most~$\rho \cdot \OPT(\mathcal{I})$.
We say that a $\rho$-approximation algorithm approximates~$\Pi$ within a factor of~$\rho$.
If no $\rho$-approximation algorithm exists for a problem~$\Pi$, then we say that~$\Pi$ cannot be approximated within a factor of~$\rho$ in polynomial time.

\begin{thm}\label{sup:thm:mdst-log2-new}
    Unless~$\mathrm{P} = \mathrm{NP}$, there is a constant~$c > 0$ such that MDST cannot be approximated within a factor of~$\rho = c\log^2 N$ in polynomial time, where~$N$ is the number of nodes. This holds even if integer parameters are polynomially bounded by instance size.
\end{thm}

We devote the remainder of this section to the proof, in which we reduce the \textsc{Set Cover} problem to MDST.
(Under Crescenzi's taxonomy~\cite{crescenzi1997short}, we may classify the reduction as an A-reduction.)

\optimizationproblem{Set Cover}{prob:msc}
{A ground set (``universe'')~$U= \{ u_1, \dots, u_\nu \}$ and a collection of subsets~$\mathcal{S} = \{S_1, \dots, S_\mu\}$ with~$S_i \subseteq U$ for every~$i \in \{1, \dots, \mu\}$.}
{A subcollection~$\mathcal{S}' \subseteq \mathcal{S}$ that covers every element of the ground set, i.e.,~$\bigcup_{S \in \mathcal{S}'} S = U$.}
{Minimize~$|\mathcal{S}'|$.}

\textsc{Set Cover} is known to be hard to approximate.
For every~$\epsilon > 0$, \textsc{Set Cover} cannot be approximated within a factor of~$(1-\epsilon) \ln \nu$ in polynomial time unless~$\mathrm{P} = \mathrm{NP}$, where~$\nu = \lvert U \rvert$ \cite{moshkovitz2015projection, dinur2014analytical}.
This holds even if the number of subsets~$\mu$ is bounded by some polynomial in~$\nu$, and hence there exists a constant~$c > 0$ such that \textsc{Set Cover} cannot be approximated within a factor of~$c \log (\nu + \mu)$ in polynomial time, unless~$\mathrm{P} = \mathrm{NP}$~\cite{escoffier2006completeness,chlebik2008approximation}.
We restrict our attention to the nontrivial instances of \textsc{Set Cover}. These are instances with $U \neq \emptyset$ that have a solution (each element of $U$ appears in at least one subset in $\mathcal{S}$). 
\\
\paragraph*{\underline{Proof Outline.}} The proof proceeds by contradiction. 
Assume the theorem does not hold, i.e., MDST can be approximated ``well'', admitting near-optimal solutions in polynomial time.
We then perform a reduction from \textsc{Set Cover} to MDST that implies that \textsc{Set Cover} can also be approximated ``well'', contradicting its known hardness unless~$\mathrm{P} = \mathrm{NP}$.
We organize the proof into four main steps.
For an illustration of the first three steps, see Fig.~\ref{sup:fig:scheme-new}.

\begin{figure}[t]
    \centering
    
    \begin{tikzpicture}[yscale=-1]
        \tikzmath{
            \width = 16;
            \height = 2.25;
            \step = 0.15;
            \separation = 1;
            \offset = 3;
        }
    
        \begin{scope}
            \coordinate (a) at (0, 0);
            \coordinate (b) at (\width, 0);
            \coordinate (c) at (\width, \height);
            \coordinate (d) at (0, \height);
        
            \node[draw, anchor=north west] at (0,0) {Step 1};
        
            \node at (0.5*\width, 0.5*\height) {
                \begin{tikzpicture}[yscale=-1]
                    \node[align=center] (l1) at (\offset, 0) {Instance of\\ \textsc{Set Cover}};
                    \node[align=center] (l2) at (\offset, \separation) {\LARGE $\mathcal{I}$\vphantom{$\mathcal{I}'$}};
            
                    \node[align=center] (r1) at (\width-\offset, 0) {Instance of\\ MDST};
                    \node[align=center] (r2) at (\width-\offset, \separation) {\LARGE $\mathcal{I}'$};
        
                    \draw[thick,-{Stealth[length=2mm]}] (l2.east) -- node[above=2pt,align=center] {Polynomial-time forward mapping (Mapping~\ref{sup:cons:sc2mdst-new})} (r2.west);
                \end{tikzpicture}
            };
            \draw (a) -- (b) -- (c) -- (d) -- cycle;
        \end{scope}

        \begin{scope}[yshift={(\height+\step)*1cm}]
            \coordinate (a) at (0, 0);
            \coordinate (b) at (\width, 0);
            \coordinate (c) at (\width, \height);
            \coordinate (d) at (0, \height);
        
            \node[draw, anchor=north west] at (0,0) {Step 2};
        
            \node at (0.5*\width, 0.5*\height) {
                \begin{tikzpicture}[yscale=-1]
                    \node[align=center] (l1) at (\offset, 0) {A solution to $\mathcal{I}$\\ (a collection of subsets)};
                    \node[align=center] (l2) at (\offset, \separation) {\LARGE $x$\vphantom{$x'$}};
            
                    \node[align=center] (r1) at (\width-\offset, 0) {A solution to $\mathcal{I}'$\\ (a spanning tree)};
                    \node[align=center] (r2) at (\width-\offset, \separation) {\LARGE $x'$};
        
                    \draw[thick,{Stealth[length=2mm]}-] (l2.east) -- node[above=2pt,align=center] {Polynomial-time backward mapping (Mapping~\ref{sup:cons:mdst2sc-new})} (r2.west);
                \end{tikzpicture}
            };
            \draw (a) -- (b) -- (c) -- (d) -- cycle;
        \end{scope}
    
        \begin{scope}[yshift={(2*\height+2*\step)*1cm}]
            \coordinate (a) at (0, 0);
            \coordinate (b) at (\width, 0);
            \coordinate (c) at (\width, \height);
            \coordinate (d) at (0, \height);
        
            \node[draw, anchor=north west] at (0,0) {Step 3};
        
            \node at (0.5*\width, 0.5*\height) {
                \begin{tikzpicture}[yscale=-1]
                    \node[align=center] (l2) at (\offset, 0) {\large $|x| \leq \sqrt{2 \rho} \OPT(\mathcal{I})$};
            
                    \node[align=center] (r2) at (\width-\offset, 0) {\large $C(x') \leq \rho \OPT(\mathcal{I}')$\vphantom{$\sqrt{2 \rho}$}};
        
                    \draw[thick,implies-,double equal sign distance] (l2.east) -- node[above=2pt,align=center] {Implies (Lemma~\ref{sup:lem:final-new})} (r2.west);
                \end{tikzpicture}
            };
            \draw (a) -- (b) -- (c) -- (d) -- cycle;
        \end{scope}
    \end{tikzpicture}

    \caption{
        Illustration of the first three steps of the proof of Theorem~\ref{sup:thm:mdst-log2-new}.
    }
    \label{sup:fig:scheme-new}
\end{figure}
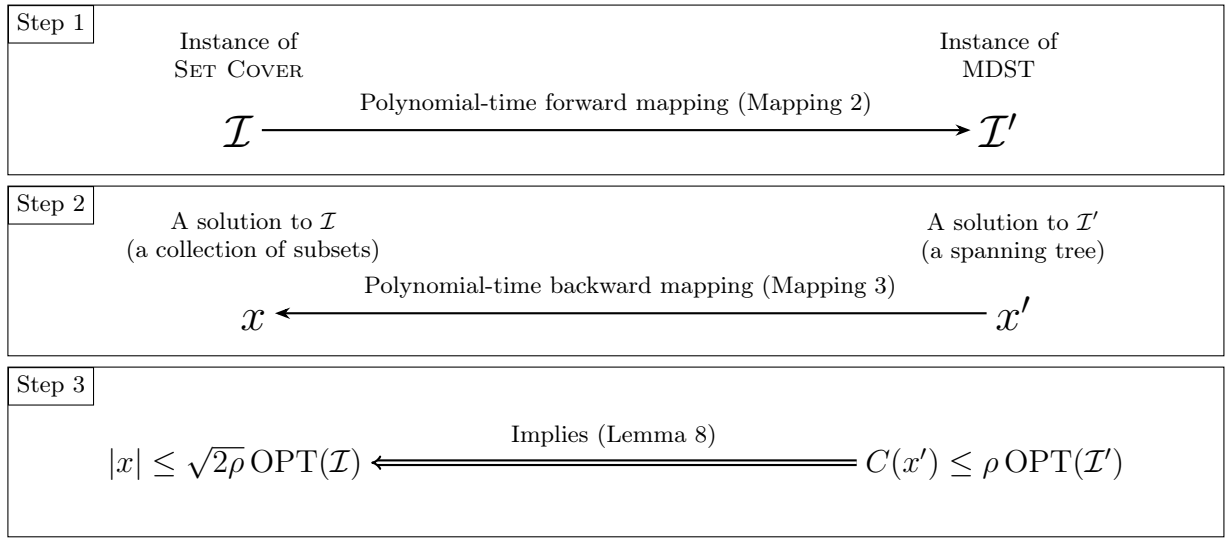

\begin{enumerate}
    \item We provide a polynomial-time ``forward'' mapping from any instance $\mathcal{I}$ of \textsc{Set Cover} to an instance $\mathcal{I}^\prime$ of MDST. 
    The idea is that, by assumption, we can approximate MDST ``well'', which allows us to efficiently compute a ``good'' solution~$x'$ to~$\mathcal{I}'$.
    \item We provide a polynomial-time ``backward'' mapping between solutions: given any solution~$x'$ to~$\mathcal{I}'$, it produces a solution~$x$ to~$\mathcal{I}$. The idea is to transform a ``good'' solution~$x'$ for MDST into a ``good'' solution~$x$ for \textsc{Set Cover}.
    \item We formally prove that a ``good'' solution~$x'$ results in a ``good'' solution~$x$. 
    Specifically, we show that, if~$C(x') \leq \rho \OPT(\mathcal{I}')$ holds for some~$\rho \geq 1$, then~$|x| \leq \sqrt{2\rho} \OPT(\mathcal{I})$.
    \item We derive a contradiction using the insights from the previous steps. More precisely, we argue that sequential application of the forward mapping, a ``good'' approximation algorithm for MDST, and the backward mapping yields a ``good'' approximation algorithm for \textsc{Set Cover}. But such an algorithm cannot exist unless~$\mathrm{P} = \mathrm{NP}$.
\end{enumerate}

\paragraph*{\underline{Step 1.}}
We now present the forward mapping that transforms any instance~$\mathcal{I}$ of \textsc{Set Cover} into an instance~$\mathcal{I}'$ of MDST.
For technical reasons, the graph of the constructed instance~$\mathcal{I}'$ contains~$\mu$ ``copies'', indexed by $\ell$, of identical structure.
Note that the two source nodes~$y$ and~$z$ do not belong to any particular copy.

\begin{cons}[Forward Mapping]\label{sup:cons:sc2mdst-new}
    Let~$\mathcal{I}$ be an instance of \textsc{Set Cover} with ground set~$U = \{u_1, \dots, u_\nu\}$ and a collection of subsets~$\mathcal{S} = \{S_1, \dots, S_\mu\}$.
    We construct an instance~$\mathcal{I}'$ of MDST with graph~$\GG = (\VV, \EE)$ as follows (see Fig.~\ref{sup:fig:log_inapprox-new} for an illustration).

    \begin{figure}[t]
        \centering
        \includegraphics[width=1\columnwidth]{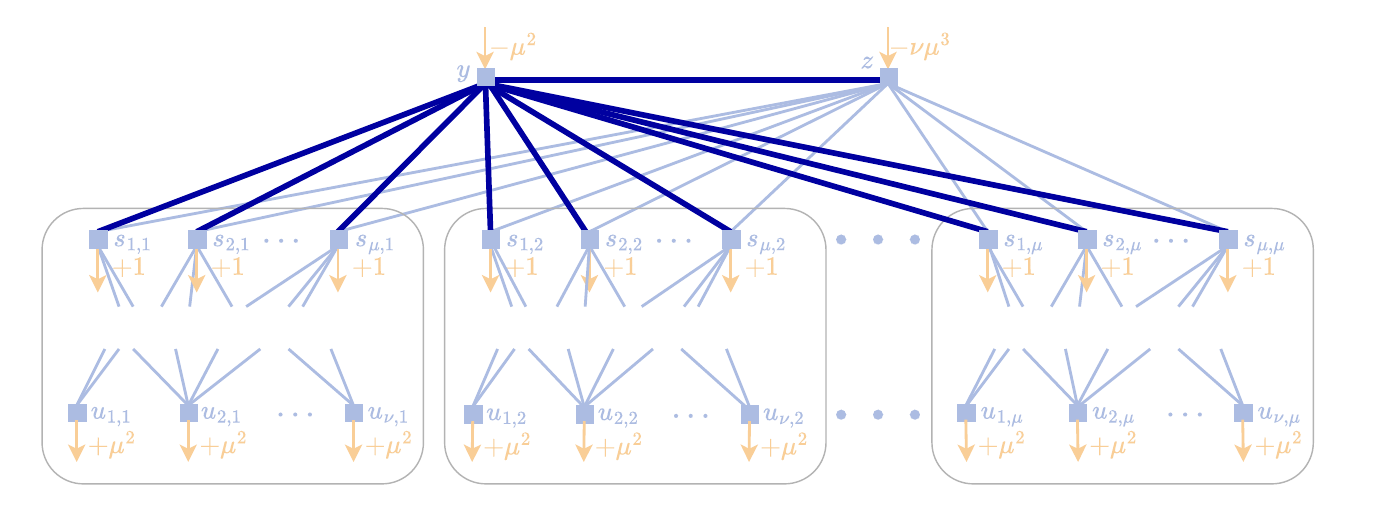}
        \caption{The graph~$\GG$ is drawn in blue. Thin edges have a dissipation constant of~$0$, whereas darker thick edges have a dissipation constant of~$1$. Demands are shown in yellow. The frames mark copies with equal structure.}
        \label{sup:fig:log_inapprox-new}
    \end{figure}

    We define a set~$\VV_\mathcal{S}$ of nodes corresponding to the elements of~$\mathcal{S}$, and a set~$\VV_U$ of nodes corresponding to the elements of~$U$.
    Precisely, let~$\VV_{\mathcal{S},\ell} \coloneqq \bigcup_{i=1}^\mu \{s_{i,\ell}\}$ for each~$\ell \in \{1, \dots, \mu\}$, and let~$\VV_\mathcal{S} \coloneqq \bigcup_{\ell=1}^\mu \VV_{\mathcal{S},\ell}$.
    Similarly, let~$\VV_{U,\ell} \coloneqq \bigcup_{j=1}^\nu \{u_{j,\ell}\}$ for each~$\ell \in \{1, \dots, \mu\}$, and let~$\VV_U \coloneqq \bigcup_{\ell=1}^\mu \VV_{U,\ell}$.
    The full node set is then~$\VV \coloneqq \{y,z\} \cup \VV_\mathcal{S} \cup \VV_U$.
    We assign the following demands to the nodes: $\mathfrak{f}_y \coloneqq - \mu^2$, $\mathfrak{f}_z \coloneqq - \nu \mu^3$, $\mathfrak{f}_v \coloneqq 1$ for every~$v \in \VV_\mathcal{S}$, and~$\mathfrak{f}_v \coloneqq \mu^2$ for every~$v \in \VV_U$.
    
    We then connect the nodes as follows: the edge sets~$\EE_y \coloneqq \{ \{y, v\} \mid v \in \VV_\mathcal{S} \}$ and~$\EE_z \coloneqq \{ \{z, v\} \mid v \in \VV_\mathcal{S} \}$ connect~$y$ and~$z$, respectively, to the vertices of~$\VV_\mathcal{S}$.
    The edge set~$\EE_\SU \coloneqq \bigcup_{\ell=1}^\mu \{ \{ s_{i, \ell}, u_{j, \ell} \} \mid u_j \in S_i \}$ encodes membership of the elements of the ground set~$U$ within the subsets of the collection~$\mathcal{S}$. 
    The full edge set is then ~$\EE \coloneqq \{\{y,z\}\} \, \cup \, \EE_y \, \cup \, \EE_z \, \cup \, \EE_\SU$. Finally, for each edge~$e \in \EE$, we choose the dissipation constant~$\alpha_e \coloneqq 1$ if~$e$ is incident to~$y$, and~$\alpha_e \coloneqq 0$ otherwise.
\end{cons}

Note that this forward mapping can be performed in polynomial time, and that~$\GG$ is connected.
Further note that~$|\VV_\mathcal{S}| = \mu^2$, $|\VV_U| = \nu \mu$, and~$|\VV| = \mu^2 + \nu \mu + 2$.
In particular, all integer parameters are polynomially bounded by $|\VV|$ and hence also by instance size.
For~$\mathcal{I}'$ to be a valid instance of MDST, the equation~$\sum_{v \in \VV} \mathfrak{f}_v = 0$ must hold, which we verify by a simple computation
\begin{displaymath}
    \sum_{v \in \VV} \mathfrak{f}_v = \mathfrak{f}_y + \mathfrak{f}_z + |\VV_\mathcal{S}| \cdot 1 + |\VV_U| \cdot \mu^2 = -\mu^2 - \nu \mu^3 + \mu^2 + \nu \mu^3 = 0.
\end{displaymath}

We discuss the instance~$\mathcal{I}'$ that the forward mapping produces and its optimal solution $x^\prime$ to provide some intuition for the remainder of the proof. Recall that the objective of MDST is to find a spanning tree of minimum cost, meaning that the corresponding radial network configuration causes minimum dissipation.
In pursuing this objective, the edges incident to node~$y$ are critical, as only these edges have nonzero dissipation constants.
At source node~$y$, there is an influx of~$\mu^2$ units of flow, which we must distribute to the other nodes of the network through the edges incident to~$y$.
Since, on each edge, dissipation grows quadratically with the flow, we would like to divide the flow originating at node~$y$ as evenly as possible across the edges incident to~$y$ to minimize cost.
Note that the node set~$\VV_\mathcal{S}$ contains exactly~$\mu^2$ nodes, each with demand~$1$ and adjacent to~$y$.
A natural idea is to send one unit of flow from~$y$ to each node of~$\VV_\mathcal{S}$, creating an essentially perfect division of the flow originating at node~$y$.

However, this naive approach does not capture the full picture and does not produce feasible solutions, since we must also distribute the flow originating at the other source node~$z$ of the network.
At node~$z$, a much larger flow of~$\nu \mu^3$ units enters the network.
Since the nodes in~$\VV_\mathcal{S}$ only consume~$\mu^2$ units of flow in total, most of the flow originating at node~$z$ must ultimately go to the nodes in~$\VV_U$.
However, none of the nodes in~$\VV_U$ is adjacent to~$z$, which means that the flow from node~$z$ must pass through some nodes in~$\VV_\mathcal{S}$ to reach its destinations.
To prevent cycles, none of these ``pass-through'' nodes can be adjacent to node~$y$ (except for at most one).
To see why we would otherwise have cycles, first note that routing any relevant amount of flow originating at node~$z$ through node~$y$ would involve edges with nonzero dissipation factors and hence dramatically increase cost, which implies that each pass-through node is connected to~$z$ by a path that does not contain~$y$.
Thus, if two pass-through nodes are both connected to~$y$ by a direct edge, then there is a cycle because they also have~$y$-independent paths to~$z$.

The central observation is that minimizing cost requires minimizing the number of pass-through nodes.
If there are few pass-through nodes, then many nodes in~$\VV_\mathcal{S}$ remain available to each receive one unit of flow from node~$y$, allowing us to get close to the naive perfect division of the flow originating at node~$y$ that we pondered earlier.
We conclude that we can frame the objective as minimizing the number of pass-through nodes: we seek a minimum number of nodes in~$\VV_\mathcal{S}$ to ``cover'' all nodes in~$\VV_U$.
This establishes the correspondence to the \textsc{Set Cover} problem.
The pass-through nodes take a central role in the backward mapping and the formal analysis.

Lastly, we mention that the purpose of using multiple copies of the same structure is to amplify the penalty incurred by pass-through nodes.
This sharpens the cost gap between solutions of different quality, especially when the number of pass-through nodes is small, and is crucial for the proof.
\\
\paragraph*{\underline{Step 2.}}
We now present the backward mapping that transforms any solution~$x'$ to~$\mathcal{I}'$ into a solution~$x$ to~$\mathcal{I}$.
To simplify the discussion, we introduce some additional notation.
Given a graph~$\mathcal{H} = (\mathcal{W},\mathcal{F})$ and a node~$v \in \mathcal{W}$, the \emph{degree} of~$v$ in~$\mathcal{H}$, written~$\deg_\mathcal{H}(v)$, is the number of nodes adjacent to~$v$ in~$\mathcal{H}$.
Given a node subset~$\mathcal{W}' \subseteq \mathcal{W}$, let~$V_{\geq 2}^\mathcal{H}(\mathcal{W}') \coloneqq \{v \in \mathcal{W}' \mid \deg_\mathcal{H}(v) \geq 2\}$.
In other words, the set~$V_{\geq 2}^\mathcal{H}(\mathcal{W}')$ contains all nodes of~$\mathcal{W}'$ that have a degree of at least two in graph~$\mathcal{H}$.
Recall again that any solution~$x'$ to~$\mathcal{I}'$ is a spanning tree and, in particular, a graph. Essentially, 
the set~$V_{\geq 2}^{x'}(\VV_\mathcal{S})$ contains the pass-through nodes of the solution~$x'$.
This is because all nodes of~$\VV_\mathcal{S}$ that are only adjacent to node~$y$ have a degree of one in~$x'$.

\begin{cons}[Backward Mapping]\label{sup:cons:mdst2sc-new}
    Let~$x'$ be a solution to~$\mathcal{I}'$.
    We construct a solution~$x$ to~$\mathcal{I}$.
    Choose any index~$\ell \in \{1, \dots, \mu\}$ such that the size of~$V_{\geq 2}^{x'}(\VV_{\mathcal{S}, \ell})$ is minimum.
    Then, we select~$x \coloneqq \{S_i \subseteq \mathcal{S} \mid s_{i,\ell} \in V_{\geq 2}^{x'}(\VV_{\mathcal{S}, \ell})\}$ as a solution to~$\mathcal{I}$.
\end{cons}

Note that the mapping can be performed in polynomial time.
Intuitively, the mapping constructs the solution~$x$ by translating pass-through nodes of the solution~$x'$ to elements of the collection of subsets~$\mathcal{S}$.

We now show that~$x$ is indeed a solution to~$\mathcal{I}$ by contradiction. Suppose that~$x$ is not a solution to~$\mathcal{I}$.
Then, there is an element~$u_j \in U$ that is not covered by~$x$, i.e., we have~$u_j \notin S_i$ for every~$S_i \in x$.
Let~$u_j$ be such an element.
Moreover, let~$\ell \in \{1,\dots,\mu\}$ be the index chosen in the computation of~$x$ by the backward mapping.
We recall that~$\EE_\SU = \bigcup_{\ell=1}^\mu \{ \{ s_{i, \ell}, u_{j, \ell} \} \mid u_j \in S_i \}$ (defined within the forward mapping) is part of the edge set~$\EE$ of the graph~$\GG$ of instance~$\mathcal{I}'$.
We observe that for every node~$s_{i,\ell} \in \VV_{\mathcal{S},\ell}$ with~$\{s_{i,\ell},u_{j,\ell}\} \in \EE$ we have~$u_j \in S_i$ and hence~$S_i \notin x$.
By definition of~$x$, this further implies that for every node~$s_{i,\ell} \in \VV_{\mathcal{S},\ell}$ with~$\{s_{i,\ell},u_{j,\ell}\} \in \EE$ we have~$s_{i,\ell} \notin V_{\geq 2}^{x'}(\VV_{\mathcal{S},\ell})$ and hence~$\deg_{x'}(s_{i,\ell}) \leq 1$.
Every path in the tree~$x'$ starting at~$u_{j,\ell}$ begins with some edge~$\{u_{j,\ell}, s_{i,\ell}\}$.
However, every such path ends at~$s_{i,\ell}$ because~$\deg_{x'}(s_{i,\ell}) \leq 1$.
In particular, there is no path from~$u_{j,\ell}$ to~$y$ in~$x'$, implying that~$x'$ is not a spanning tree of~$\GG$.
This contradicts that~$x'$ is a solution to~$\mathcal{I}'$.
\\
\paragraph*{\underline{Step 3.}}
The goal of this step is to prove the implication depicted in Step~3 of Fig.~\ref{sup:fig:scheme-new}, which is formalized in Lemma~\ref{sup:lem:final-new}.
Informally, Lemma~\ref{sup:lem:final-new} asserts that if~$x'$ is a low-cost solution, then so is~$x$.
In preparation for Lemma~\ref{sup:lem:final-new}, we prove Lemma~\ref{sup:lem:deg-new} and Lemma~\ref{sup:lem:opt-cmp-new}.

Lemma~\ref{sup:lem:deg-new} formalizes the idea that pass-through nodes are expensive. Specifically, the square of the number of pass-through nodes is a lower bound on the solution cost for~$\mathcal{I}'$.
Hence, each additional pass-through node becomes progressively more costly, creating an incentive to minimize their number.
We use the notation~$V(\mathcal{H})$ to denote the vertex set of a graph~$\mathcal{H}$.

\begin{lem}\label{sup:lem:deg-new}
    Let~$x'$ be a solution to~$\mathcal{I}'$.
    Then,~$C(x') \geq \big|V_{\geq 2}^{x'}(\VV_\mathcal{S})\big|^2$.
\end{lem}
\begin{proof}
    Let~$\TT_1, \dots, \TT_k$ be the connected subgraphs (\textit{components}) obtained if we were to delete~$y$ from~$x'$, and let~$\mathcal{C} \coloneqq \{\TT_1, \dots, \TT_k\}$.
    Since exactly the edges incident to~$y$ have nonzero dissipation constant, we have~$C(x') = \sum_{\TT \in \mathcal{C}} (\sum_{v \in V(\TT)} \mathfrak{f}_v)^2$.
    Let~$\TT_z \in \mathcal{C}$ be the component containing~$z$.
    We distinguish between two cases.

    \begin{enumerate}
        \item For the first case, let there be a component~$\TT^\star \in \mathcal{C} \setminus \{ \TT_z \}$ and a node~$v^\star \in \VV_U$ such that~$v^\star \in V(\TT^\star)$.
        Then, using that~$y,z \notin V(\TT^\star)$ and that all nodes except for~$y$ and~$z$ have positive demands, we get
        \begin{displaymath}
        \sum_{v \in V(\TT^\star)} \mathfrak{f}_v
        = \mathfrak{f}_{v^\star} + \sum_{\substack{v \in V(\TT^\star) \\ v \notin \{y,z,v^\star\}}} \mathfrak{f}_v
        \geq \mathfrak{f}_{v^\star}
        = \mu^2.
        \end{displaymath}
        With this, we have
        \begin{displaymath}
        C(x')
        = \sum_{\TT \in \mathcal{C}} \Bigl( \sum_{v \in V(\TT)} \mathfrak{f}_v \Bigr)^2
        \geq \Bigl( \sum_{v \in V(\TT^\star)} \mathfrak{f}_v \Bigr)^2
        \geq ( \mu^2 )^2
        = |\VV_\mathcal{S}|^2
        \geq \big|V_{\geq 2}^{x'}(\VV_\mathcal{S})\big|^2.
        \end{displaymath}
        \item In the second case, every~$v \in \VV_U$ is contained in~$\TT_z$.
        By construction, each node in~$\VV_\mathcal{S}$ only has the following neighbors in~$\GG$: the nodes~$y$,~$z$, and some nodes from~$\VV_U$.
        Hence, when we delete~$y$ from~$x'$, then every node in~$V_{\geq 2}^{x'}(\VV_\mathcal{S})$ is adjacent to~$z$ or some node of~$\VV_U$ in the resulting graph.
        Since~$z$ and all nodes in~$\VV_U$ are included in~$\TT_z$, it follows that all nodes in~$V_{\geq 2}^{x'}(\VV_\mathcal{S})$ are also included in~$\TT_z$. 
        Some additional nodes from~$\VV_\mathcal{S}$ may also be included in~$\TT_z$.
        We get 
        \begin{displaymath}
        \sum_{v \in V(\TT_z)} \mathfrak{f}_v
        \geq \mathfrak{f}_z + \big|V_{\geq 2}^{x'} (\VV_\mathcal{S})\big| \cdot 1 + |\VV_U| \cdot \mu^2
        = -\nu \mu^3 + \big|V_{\geq 2}^{x'}(\VV_\mathcal{S})\big| + \nu \mu^3
        = \big|V_{\geq 2}^{x'}(\VV_\mathcal{S})\big|,
        \end{displaymath}
        where we use that nodes in~$\VV_\mathcal{S}$ are consumer nodes, meaning they have positive flow demand.
        
        Hence, we find that
        \begin{displaymath}
        C(x')
        = \sum_{\TT \in \mathcal{C}} \Bigl( \sum_{v \in V(\TT)} \mathfrak{f}_v \Bigr)^2
        \geq \Bigl( \sum_{v \in V(\TT_z)} \mathfrak{f}_v \Bigr)^2
        = \big|V_{\geq 2}^{x'}(\VV_\mathcal{S})\big|^2.
    \end{displaymath}
    \end{enumerate}
\end{proof}

The next lemma provides an upper bound on how much larger the optimum of the constructed instance~$\mathcal{I}'$ can be relative to the optimum of the original instance~$\mathcal{I}$.
This allows us to relate the cost of solutions to~$\mathcal{I}'$ to solutions to~$\mathcal{I}$.

\begin{lem}\label{sup:lem:opt-cmp-new}
    The optimal values satisfy the inequality $\OPT(\mathcal{I}') \leq 2\mu^2 \OPT(\mathcal{I})^2$.
\end{lem}
\begin{proof}
    Let~$x$ be an optimal solution to~$\mathcal{I}$.
    We prove the claim by constructing a solution~$x'$ to~$\mathcal{I}'$ such that~$C(x') \le 2\mu^2 |x|^2$.

    Let~$\EE_y' \coloneqq \bigcup_{\ell = 1}^\mu \{ \{ y, s_{i,\ell} \} \mid S_i \notin x \}$ and~$\EE_z' \coloneqq \bigcup_{\ell = 1}^\mu \{ \{ z, s_{i,\ell} \} \mid S_i \in x \}$.
    Let~$g : U \rightarrow x$ be an auxiliary function that maps each element~$u \in U$ to a subset~$S \in x$.
    Such a function exists since~$x$ is a solution to~$\mathcal{I}$.
    Then, let~$\EE_\SU' \coloneqq \bigcup_{\ell = 1}^\mu \{ \{ s_{i,\ell}, u_{j, \ell} \} \mid S_i = g(u_j) \}$.
    Finally, we set~$\EE' \coloneqq \{\{y,z\}\} \cup \EE_y' \cup \EE_z' \cup \EE_\SU'$.  

    By construction,~
    the resulting subgraph $x' = (\VV, \EE')$ is a spanning tree.
    First, we observe that~$x'$ is connected: node~$y$ is adjacent to node~$z$, every node in~$\VV_\mathcal{S}$ is adjacent to~$y$ or~$z$, and every node in~$\VV_U$ is adjacent to a node in~$\VV_\mathcal{S}$.
    Second, we argue that~$x'$ has~$|\VV| - 1$ edges:
    \begin{displaymath}
        |\EE'| = 1 + |\EE_y'| + |\EE_z'| + |\EE_\SU'| = 1 + |\VV_\mathcal{S}| + |\VV_U| = 1 + \mu^2 + \nu \mu = |\VV| - 1.
    \end{displaymath}
    Hence,~$x'$ is a spanning tree and a solution to $\mathcal{I}'$. 

    Next, we show that~$C(x') \leq 2 \mu^2 |x|^2$.
    For determining~$C(x')$, it suffices to consider the edges in~$\EE_y'$ and the edge~$\{y,z\}$ because only edges incident to~$y$ have nonzero dissipation constant.
    We start with~$\EE_y'$.
    Let~$\{ y, s_{i,\ell} \} \in \EE_y'$.
    By definition of~$\EE_y'$, we have~$S_i \notin x$.
    From the definitions of~$\EE_z'$ and~$\EE_\SU'$, we see that~$s_{i,\ell}$ has degree one (is a \emph{leaf}) in~$x'$.
    Since~$s_{i,\ell} \in \VV_\mathcal{S}$, we additionally have~$\mathfrak{f}_{s_{i,\ell}} = 1$.
    It follows that~$\{ y, s_{i,\ell} \}$ carries one unit of flow, and hence every edge in~$\EE_y'$ carries one unit of flow.
    Then, the edge~$\{y,z\}$ carries~$|\mathfrak{f}_y| - |\EE_y'|$ units of flow due to Kirchhoff's current law (flow conservation).
    This yields the following result
    \begin{displaymath}
        C(x')
        = |\EE_y'| \cdot 1^2 + 1 \cdot (|\mathfrak{f}_y| - |\EE_y'|)^2
        \, \overset{(\ast)}{=} \, \mu^2 - \mu |x| + (\mu^2 - \mu^2 + \mu |x|)^2
        \leq \mu^2 |x|^2 + \mu^2
        \, \overset{(\ast\ast)}{\leq} \, 2 \mu^2 |x|^2
    \end{displaymath}
    where, at~$(\ast)$, we use that~$|\EE_y'| = \mu (\mu - |x|) = \mu^2 - \mu |x|$, and, at~$(\ast\ast)$, we use that~$|x| \geq 1$ for nontrivial instances of \textsc{Set Cover}.
    Finally, we note that~$\OPT(\mathcal{I}') \leq C(x') \leq 2 \mu^2 |x|^2 = 2 \mu^2 \OPT(\mathcal{I})^2$.
\end{proof}

With the help of the previous two lemmas, we can give a concise proof of Lemma~\ref{sup:lem:final-new}.
As already mentioned, this lemma states that a low-cost solution~$x'$ translates to a low-cost solution~$x$.

\begin{lem}\label{sup:lem:final-new}
    Let~$x'$ be a solution to~$\mathcal{I}'$.
    Let~$x$ be the solution to~$\mathcal{I}$ obtained from~$x'$ by the backward mapping.
    If~$C(x') \leq \rho \OPT(\mathcal{I}')$ for some~$\rho \geq 1$, then~$|x| \leq \sqrt{2 \rho} \OPT(\mathcal{I})$.
\end{lem}
\begin{proof}
    \begin{displaymath}
        |x|
        \, \overset{\text{Map \ref{sup:cons:mdst2sc-new}}}{=}
        \min_{\ell=1}^\mu |V_{\geq 2}^{x'}(\VV_{\mathcal{S},\ell})| \,
        \leq \frac{1}{\mu} |V_{\geq 2}^{x'}(\VV_\mathcal{S})|
        \, \overset{\text{Lem \ref{sup:lem:deg-new}}}{\leq} \, \frac{1}{\mu} \sqrt{C(x')}
        \leq \frac{1}{\mu} \sqrt{\rho \OPT(\mathcal{I'})}
        \, \overset{\text{Lem \ref{sup:lem:opt-cmp-new}}}{\leq} \, \frac{1}{\mu} \sqrt{2 \rho \mu^2 \OPT(\mathcal{I})^2}
        = \sqrt{2 \rho} \OPT(\mathcal{I})
    \end{displaymath}
\end{proof}

Notably, the potential gap to the optimum tightens: if~$x'$ deviates from the optimum by a factor of at most~$\rho$, then~$x$ deviates from the optimum by a factor of at most~$\sqrt{2\rho}$.
\\
\paragraph*{\underline{Step 4.}}
With the help of Lemma~\ref{sup:lem:final-new}, we can finally prove Theorem~\ref{sup:thm:mdst-log2-new}.
We denote the number of nodes of a MDST instance by~$N$.
Assume for contradiction that MDST can be approximated within a factor of~$c \log^2 N$ for every~$c > 0$.
Let us consider some~$c > 0$ to be fixed.
Then, given an instance~$\mathcal{I}$ of \textsc{Set Cover}, we perform the following steps.

First, we produce an instance~$\mathcal{I}'$ of MDST using the polynomial-time forward mapping.
By assumption, we can compute a solution~$x'$ to~$\mathcal{I}'$ with~$C(x') \leq c \log^2(N)\OPT(\mathcal{I}')$ in polynomial time.
Then, we run the polynomial-time backward mapping with~$x'$ as input to get a solution~$x$ to~$\mathcal{I}$.
Note that executing the entire sequence of steps only takes polynomial time.

By setting~$\rho \coloneqq c \log^2 N$, Lemma~\ref{sup:lem:final-new} implies that~$|x| \leq \sqrt{2c \log^2 N} \OPT(\mathcal{I})$.
Hence,~$|x| \leq \sqrt{2c} \log(N) \OPT(\mathcal{I})$.
Next, recall that~$N = \mu^2 + \nu \mu + 2$.
Since~$\nu \geq 1$ and~$\mu \geq 1$ for nontrivial instances of \textsc{Set Cover}, it follows that~$N \leq (\nu + \mu)^3$.
Thus,~$|x| \leq \sqrt{2c} \log((\nu + \mu)^3) \OPT(\mathcal{I}) = 3\sqrt{2c} \log(\nu + \mu) \OPT(\mathcal{I})$.
By appropriately choosing~$c$, the term~$3\sqrt{2c}$ can be made arbitrarily close to zero, and hence \textsc{Set Cover} can be approximated within a factor of~$\tilde{c} \log(\nu + \mu)$ for any~$\tilde{c} > 0$, a contradiction unless $\mathrm{P} = \mathrm{NP}$.
This concludes the proof of Theorem~\ref{sup:thm:mdst-log2-new}.

\subsection{Reduction of $\GG$ to a Graph with Minimum Degree 2}
\label{sec:sup:red_graph_deg_2}

We now briefly discuss how nodes with only one adjacent node can be removed when their corresponding flow demands/injections are transferred to their neighbors (cf.~Ref.~\cite{silva_qubo_2023}). The single-edge incident can not be reconfigured; it has to be in every spanning tree $\TT$ of $\GG$. Hence, removing these nodes and edges does not provide any additional degree of freedom to the optimization problem. However, the resulting reduction becomes handy as it significantly reduces the number of variables $y_{e,n}$ and thus e.g., the number of quantum registers as well as the memory for a classical computer. 

Let $n \in \VV \setminus \{n_0 \}$ be a node with only one adjacent neighbor $m$. Then, the oriented edge $e=(m,n)$ must be in all spanning trees of $\GG$, as otherwise the node $n$ would not be connected. The flow on $e$ is thus fixed as $\mathfrak{f}_n$ for all configurations. Consequently, the operating cost for edge $e$ becomes an offset in the MDST cost function. Hence, we can remove node $n$ and set the flow demand of node $m$ as $\mathfrak{f}_n + \mathfrak{f}_m$ to obtain an equivalent optimization problem up to the offset. A similar reduction can be done if $n = n_0$. Then, the node $m$ becomes the new root with $\mathfrak{f}_m = \sum_{n^\prime \neq m} \mathfrak{f}_{n^\prime} = \mathfrak{f}_{n_0} - \mathfrak{f}_m$. Repeating these steps, we can reduce $\GG$ with flow demands/injections $\mathfrak{f}_n$ to a new graph with minimum degree 2 and updated injections. 

\subsection{MIP Formulation}
\label{sec:sup:mip_formulation}

Combining Eq.~\eqref{net_rec:eq:flows_in_tree} and the definition of the binary variables $y_{e,n}$ (cf. the main text or \ref{sec:partial_mixers}), the total cost function for MDST can be written as 
\begin{align}
\label{eq:total_cost_flows_tree}
\begin{split}
    C(\TT) = \sum_{e \in \EE} \sum_{n, m \in \VV\setminus \{n_0\}} \alpha_e y_{e,n} y_{e,m} \mathfrak{f}_n \mathfrak{f}_m.
\end{split}
\end{align}
Hence, it remains to define the constraints enforcing that the variables $y_{e,n}$ describe a spanning tree $\TT$. 

Three necessary conditions for any spanning tree $\TT$ with root $n_0$ are that
\begin{enumerate}
\itemsep-0.4em 
    \item the number of active edges is $\lvert \VV \rvert -1$,
    \item  no cycles are formed,
    \item  all nodes are connected to the root. 
\end{enumerate}
The combination of any two of those three conditions is also sufficient for a spanning tree, and thus, implies the third condition. 

We now want to formulate these constraints in terms of the binary variables $y_{e,n}$. However, we first need to make sure that the variables are locally consistent. That is, if a node $n$ is downward of an (oriented) edge $e=(u,m) \in \EE_\TT$ with $m \neq n$, there must exist another (oriented) edge $e^\prime = (m, v) \in \EE_\TT \setminus \{e\}$ such that $n$ is downward of $e^\prime$ as well, including the case that $v = n$. Local consistency can be enforced by the constraints 
\begin{equation}
\label{eq:net:rec:app:local_consistency}
    y_{e,n} (1-\lvert E_{n, e} \rvert) = (1-\lvert E_{n, e} \rvert) \sum_{m \in \VV \setminus \{n_0, n\}} \sum_{e^\prime \in \EE \setminus \{ e\}} y_{e, m} y_{e^\prime, n} \lvert E_{m,e} \rvert \lvert E_{m,e^\prime} \rvert,  \quad \forall n \in \VV \setminus \{n_0 \}, \, \forall e \in \EE. 
\end{equation}
The term $(1-\lvert E_{n, e} \rvert)$ evaluates to zero if $n$ is incident to $e$ to exclude this case. Then, the right-hand side enforces that for the edge $e^\prime$ downward of $e$ we have that $y_{e^\prime, n}=1$ if and only if $y_{e,n}=1$, that is, the left-hand side is $1$. However, local consistency alone does not prevent (oriented) cycles from being formed. In an oriented cycle, every node is downward of every edge such that local consistency is trivially fulfilled. 

With local consistency guaranteed, the first condition in the number of active edges can be readily formulated as a constraint 
\begin{equation}
    \label{eq:constraint_number_edges}
    \sum_e \sum_{n \in \VV \setminus \{n_0\}} \lvert E_{n,e} \rvert y_{e,n} = \lvert \VV \rvert - 1. 
\end{equation}
We note that this constraint is necessary for the fact that the number of active edges is $\lvert \VV \rvert - 1$, but not sufficient. If for an edge $e = \{n, m \} \in \EE$, we have that $y_{e,n}=1$ and $y_{e,m}=1$, the edge $e$ is counted twice in the sum. For example, such a configuration arises if a cycle is formed, since then any node in the cycle will be downward of any edge. Hence, constraint \eqref{eq:constraint_number_edges} becomes a sufficient condition for the number of edges to be equal to $\lvert \VV \rvert - 1$, if it is combined with a constraint enforcing that no cycle is formed or all nodes are connected to the root.

\begin{figure}[t]
  \centering
  \includegraphics[width=\columnwidth]{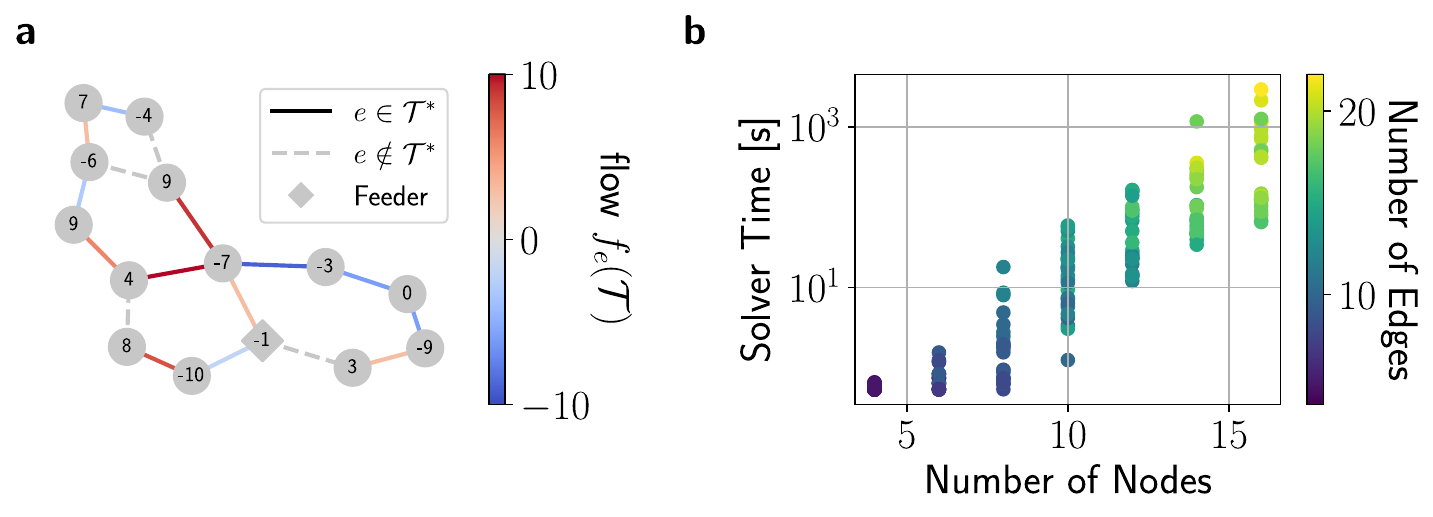}
  \caption{Obtaining minimum dissipation spanning trees by Gurobi for randomly generated Newman-Watts-Strogratz graphs \cite{newman1999renormalization} with $k=2$ and $p=0.2$ for different system sizes $\lvert \VV \rvert$. The resulting graphs $\GG$ contain a few cycles and thus mimic some aspects of distribution grids. The parameters $\alpha_e$ and $\mathfrak{f}_n$ are also chosen randomly. \textbf{a:} Minimum dissipation spanning tree $\TT^*$ for one Newman-Watts-Strogratz graph with $\lvert \VV \rvert = 14$. \textbf{b:} Exponential scaling of the solver time with the system size given by the number of nodes $\lvert \VV \rvert$. For each $\lvert \VV \rvert$, 25 random instances have been generated such that the number of edges $\lvert \EE \rvert$ differs between these instances by chance. The nodal in/outflows for all $n \in \VV \setminus \{n_0 \}$ are random integers between $-10$ and $10$, the in/out flow at the feeder is then set as the sum overall in/outflows to get a balanced grid. The dissipation constants are modeled as random integers between 1 and 5.  }
  \label{fig:MDST_solve_gurobi}
\end{figure}

We now turn to the formulation of the other two necessary conditions for a spanning tree $\TT$ with root $n_0$. Silva et al. \cite{silva_qubo_2023} proposed constraints that enforce that no cycles are formed. Their formulation requires the definition of additional binary variables. Additionally, some of the constraints are not linear in the binary variables and thus can not be written as QUBO penalties without overhead, cf., for example, equations (10a) and (10b) and the discussion in section 5.1 in \cite{silva_qubo_2023}. On the other hand, enforcing connectivity can be achieved by once again turning to the KCL. We introduce an additional ``dummy'' flow, such that $\iota_n = 1$ for all nodes $n \neq n_0$ and $\iota_{n_0} = - \lvert \VV \rvert +1$. Then the KCL can only be fulfilled for all nodes if all nodes are connected to the root. Defining the dummy flow on edge $e$ as in the main manuscript, we have the following KCL-based constraints for each node $n \in \VV \setminus \{n_0\}$
\begin{equation}
\label{eq:net_rec:app:KLC_constraint}
    1 = \iota_n = \sum_{e \in \EE} E_{n,e}(\TT) \iota_e(\TT) 
\end{equation}
with 
\begin{align*}
    E_{n,e}(\TT) &= E_{n,e} ( E_{n,e} y_{e,n} + \sum_{u \in \VV \setminus \{n_0, n\}} E_{u,e} y_{e,u}) \\ 
     \iota_e(\TT) &=  \sum_{m \in \VV \setminus \{n_0\}} \iota_m y_{e,m} = \sum_{m \in \VV \setminus \{n_0\}} y_{e,m} . 
\end{align*}

Finally, we note that the constraints~\eqref{eq:net:rec:app:local_consistency}~and~\eqref{eq:net_rec:app:KLC_constraint} are quadratic in the binary variables and thus can not be mapped to QUBO. 

We numerically verify that the constraints ~\eqref{eq:net:rec:app:local_consistency}~-~\eqref{eq:net_rec:app:KLC_constraint} enforce spanning trees using Gurobi. We solve MDST with these constraints for randomly generated topologies with different sizes and flow inputs/demands, see~Fig.~\ref{fig:MDST_solve_gurobi}. For all test cases, the optimal solutions are spanning trees. As expected, the time to find the optimal solution tends to increase with the size of the problem, that is, the number of nodes $\VV$ and/or the number of edges $\EE$. 

\newpage
\section{QAOA Simulation}
\label{sec:qaoa_simulation}

This section presents the methods and extended results for the numerical QAOA simulations. The goal is to evaluate and compare both spanning tree sampling methods, by using penalties and by restricting to the invariant feasible subspace, presented in the main paper. In \ref{sec:qaoa_simulation_methods}, we provide methodological insights into the numerical scheduled QAOA simulation and evaluation. In \ref{sec:qaoa_results} we show detailed results.  

\subsection{Methods}
\label{sec:qaoa_simulation_methods}

For the method based on penalty terms, we use a standard mixer and can thus simulate LR-QAOA. LR-QAOA provides good out-of-the-box performance for many optimization problems (cf. \ref{sec:qaoa}). For the a invariant feasible subspace method, which restricts the quantum evolution to the feasible subspace, the initial ground state is not known, and we use techniques from reverse annealing. We first present both scheduled QAOA variants and define the metrics to evaluate the optimization quality. Afterwards, we briefly discuss the experimental setup: The algorithmic implementation, the hyperparameter search, and the problem instance.  

\subsubsection*{Penalty Method: LR-QAOA}

\begin{figure}[b!]
     \centering
    \includegraphics[width=\columnwidth]{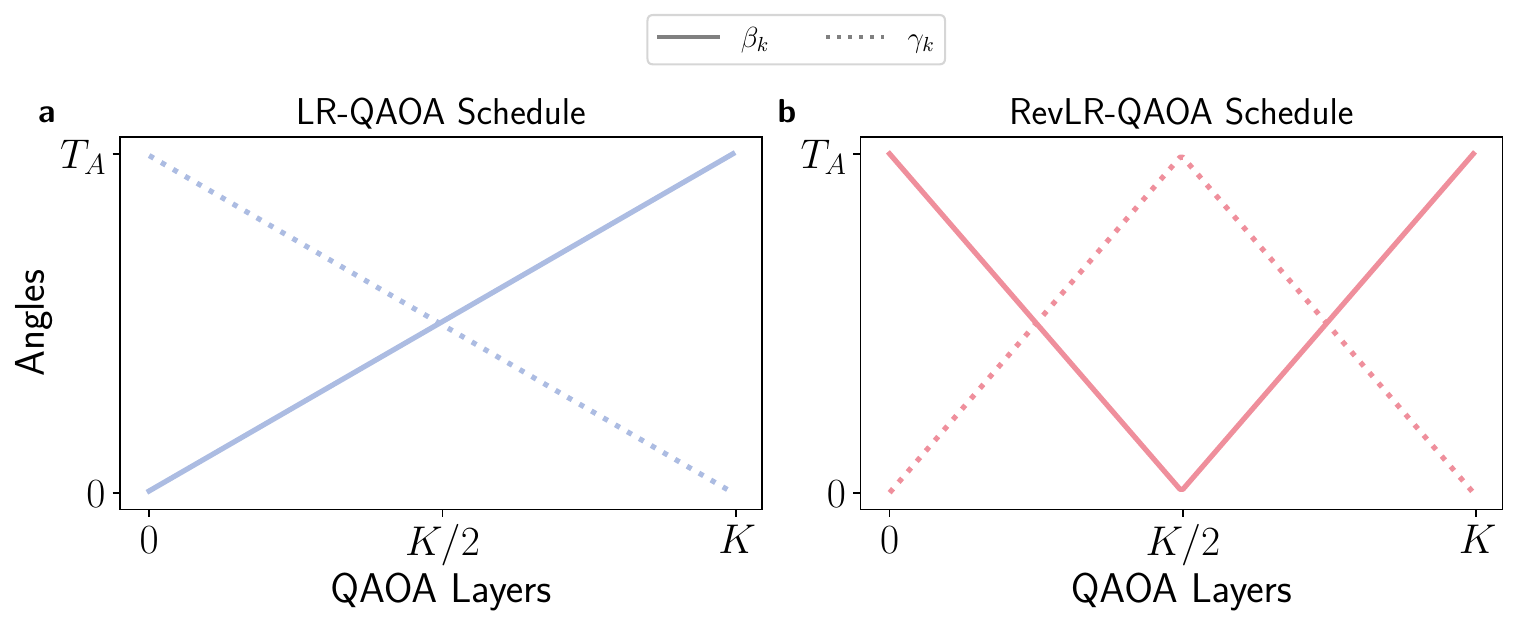}
    \caption{QAOA schedules ($\beta_k, \gamma_k$) for LR-QAOA (\textbf{a}) and RevLR-QAOA (\textbf{b}). }
    \label{sup:fig:qaoa_sched}
\end{figure}

To simulate solving MDST using LR-QAOA with the standard Mixer, we begin by transforming the original Mixed-Integer Program (MIP)~\eqref{eq:total_cost_flows_tree}-\eqref{eq:constraint_number_edges} into a Polynomial Unconstrained Binary Optimization (PUBO) problem. This is done by incorporating the constraints into the objective function as penalty terms, specifically by squaring the constraint violations. The resulting PUBO is then mapped to a fourth-order Ising hamiltonian $H_\mathrm{P}$, where binary variables are represented by spin variables according to $s_j = 1 - 2 y_j$ with $j=e \lvert \VV - 1 \rvert + (n-1)$. We then find the ground state $\ket{\psi_{E_0}}$ and corresponding ground state energy $E_0 = \braket{\psi_{E_0} |H_\mathrm{P}| \psi_{E_0}}$ of the Ising Hamiltonian using a classical solver to obtain a reference solution.

Finally, we simulate a QAOA protocol using a linear ramp annealing schedule, cf. Fig. \ref{sup:fig:qaoa_sched}\textbf{a} , defined by:
\begin{equation*}
    \beta_k = T_\mathrm{A} (1 - \frac{k}{K}), \quad \gamma_k = T_\mathrm{A} \frac{k}{K},
\end{equation*}
where $T_\mathrm{A}$ denotes the annealing time and $k=0,\ldots,K$ indexes the discrete time steps, that is, the QAOA layers. That is, we implement the sequence of unitaries as given in Eq.~\eqref{eq:schedule_qaoa} for this schedule, where we associate $A(t_k)= \beta_k$ and $B(t_k) = \gamma_k$. Since the constraints are incorporated in the cost Hamiltonian $H_P$, we use the standard Mixer Hamiltonian \eqref{eq:transverse_field_mixer} and initialize the circuit in its ground state, the uniform superposition of all possible spin configurations. For each layer $k$, we compute the fidelity of the intermediate (and final) state $\ket{\psi(k)}$ with the ground state, that is, the probability to find the ground state in a measurement,  
\begin{equation}
    F(k) = \lvert \braket{\psi(k) | \psi_{E_0} }\rvert^2
\end{equation}
and the approximation ratio 
\begin{equation}
    \Sigma(k) = \frac{\braket{\psi(k) |H_\mathrm{P}| \psi(k)}}{E_0}
\end{equation}
to quantify how well the intermediate (and final) solution compares to the optimal solution. 

\subsubsection*{Invariant Feasible Subspace Method: RevLR-QAOA}

In the case of advanced mixer unitaries, the ground state is generally not known analytically, which makes annealing-inspired LR-QAOA initialization, typically in the ground state of the mixer, challenging. To address this, we employ a reverse-annealing QAOA protocol, which allows us to begin the evolution from a known, easily-preparable ground state and gradually anneal towards the desired problem instance. Furthermore, since the constraints are incorporated in the Mixer, we only need to map the objective function~\eqref{eq:total_cost_flows_tree} to a cost Ising Hamiltonian $H_C$.

We begin by selecting an elementary problem instance $(\alpha_e, \mathfrak{f}_n)$ with the same underlying graph topology for which the ground state is known, or can be computed.
This state serves as the initial state $\ket{\psi(0)}$ of the RevLR-QAOA (reverse linear ramp) framework. Then, the annealing schedule proceeds in two phases (cf. Fig.~\ref{sup:fig:qaoa_sched}\textbf{b}):

\begin{enumerate}
    \item \textbf{Reverse Annealing Phase (first $K/2$ layers):}
    We interpolate linearly from the cost unitary $e^{-i \gamma_k H_C^{init.}}$ of the initial elementary problem instance toward the Mixer Unitary $U_M(\beta_k)$, gradually suppressing the cost term while increasing the strength of the mixer, that is, we set $\gamma_k = T_\mathrm{A} (1 - \frac{2k}{K}) $ and $\beta_k = T_\mathrm{A} \frac{2k}{K}$. At the end of the first half of the schedule, the system is governed purely by the mixer. Thus, if the annealing time $T_\mathrm{A}$ is long enough $\ket{\psi(K/2)}$ is a good approximation of the ground state of the Mixer, restricted to feasible space. 
    \item \textbf{Forward Annealing Phase (second $K/2$ layers): } We replace the cost Hamiltonian of the initial elementary instance $ H_C^{init.}$ with the target cost Hamiltonian $ H_C^{prob.}$, and reverse the annealing direction: the strength of the mixer is gradually decreased while the cost term is turned back on. This drives the system toward the ground state of the target problem Hamiltonian in the feasible space $\psi_{E_0}$. 
\end{enumerate}

Since the advanced mixer for MDST involves ancillary qubits for controlled operations, cf.~\ref{sec:partial_mixers}, we evaluate the fidelity and approximation ratio based on the reduced state obtained by tracing out the ancilla subsystem. That is, let $\rho(k) = \text{Tr}_{anc} (\ket{\psi(k})\bra{\psi(k)})$ be the density matrix of the reduced system, the fidelity is computed as
\begin{equation}
    F(k) = \braket{\psi_{E_0} | \rho(k) | \psi_{E_0} }
\end{equation}
and the approximation ratio as
\begin{equation}
    \Sigma(k) = \frac{\text{Tr}(\rho(k) H_C^{prob.})}{E_0}.
\end{equation}

\subsubsection*{Experimental Setup}

\paragraph*{Code implementation} In practice, we model MDST instances using Pyomo. For LR-QAOA, the full model (cost and constraints) is then converted into a PUBO using quboify \cite{quobify2025}, which provides automatic $\lambda_{pen}$-selection based on a naive upper bound for the cost function. For RevLR-QAOA, only the cost function is converted. Afterward, the cost functions are mapped to Ising Hamiltonians using Qiskit. 

The Mixers circuits are also implemented in Qiskit. The partial Mixer circuits \ref{sup:eq:partialmixer} are constructed by following the decomposition into smaller circuits presented in \ref{sec:partial_mixers}. 

Finally, both scheduled-QAOA variants are then simulated in Qiskit using the statevector method, that is, parameterized circuits for the Mixer and $H_C$ are applied consecutively according to the schedule. The state vector method has the advantage that the fidelity and the approximation ratio can be swiftly computed after each layer, thus providing effective logging during the ``anneal''. For larger instances, we suggest alternative Qiskit-AER simulation methods, such as ``matrix\_product\_state'' method \cite{vidal2003efficient}, with the downside that intermediate steps are not accessible. 

\paragraph*{Hyperparameter Tuning}  The performance of both QAOA variants is highly sensitive to the choice of hyperparameters $T_\mathrm{A}$ and $K$. To systematically explore their effect, 
we conduct a grid search over $K \in \{10, 50, 100, 200\}$, and 1000 values for $T_\mathrm{A}$ that are loguniformly separated in $[0.01, 1.5]$. This approach ensures that both small and large values of $T_\mathrm{A}$ are adequately represented, capturing regimes where the algorithm may behave qualitatively differently. In the future, also tuning the penalty coefficients manually can be considered for the penalty approach; however, this further increases the overhead. 

\paragraph*{Problem Instance} For the numerical simulation, we use the simplest non-trivial example, consisting of three nodes and one cycle, cf. Fig. 2 in the main manuscript. The corresponding Mixer $U_M$ is depicted in Fig.~\ref{sup:fig:mixers_simple_example}. The problem instance that we want to solve has $\alpha_0 = \alpha_1 = 1$ and $\alpha_2 = 10$. The flow demands are set as $\mathfrak{f}_0 = -3$, $\mathfrak{f}_1 = 1$ and $\mathfrak{f}_2 = 2$. The optimal solution is thus $y_{0, 1} = 1$, $y_{0,2}=1$, $y_{1,2}=1$ and all other $0$, which corresponds to the bit string $110100$, cf. Fig. 2 in the main manuscript. For RevLR-QAOA, we initialize the algorithm in the state $100001$, which is optimal for the instance $\alpha_i = 0$ and $\mathfrak{f}_0 = -2$, $\mathfrak{f}_1 = 1$ and $\mathfrak{f}_2 = 1$. 

\subsection{Results}
\label{sec:qaoa_results}

\begin{figure}[t]
     \centering
    \includegraphics[width=\columnwidth]{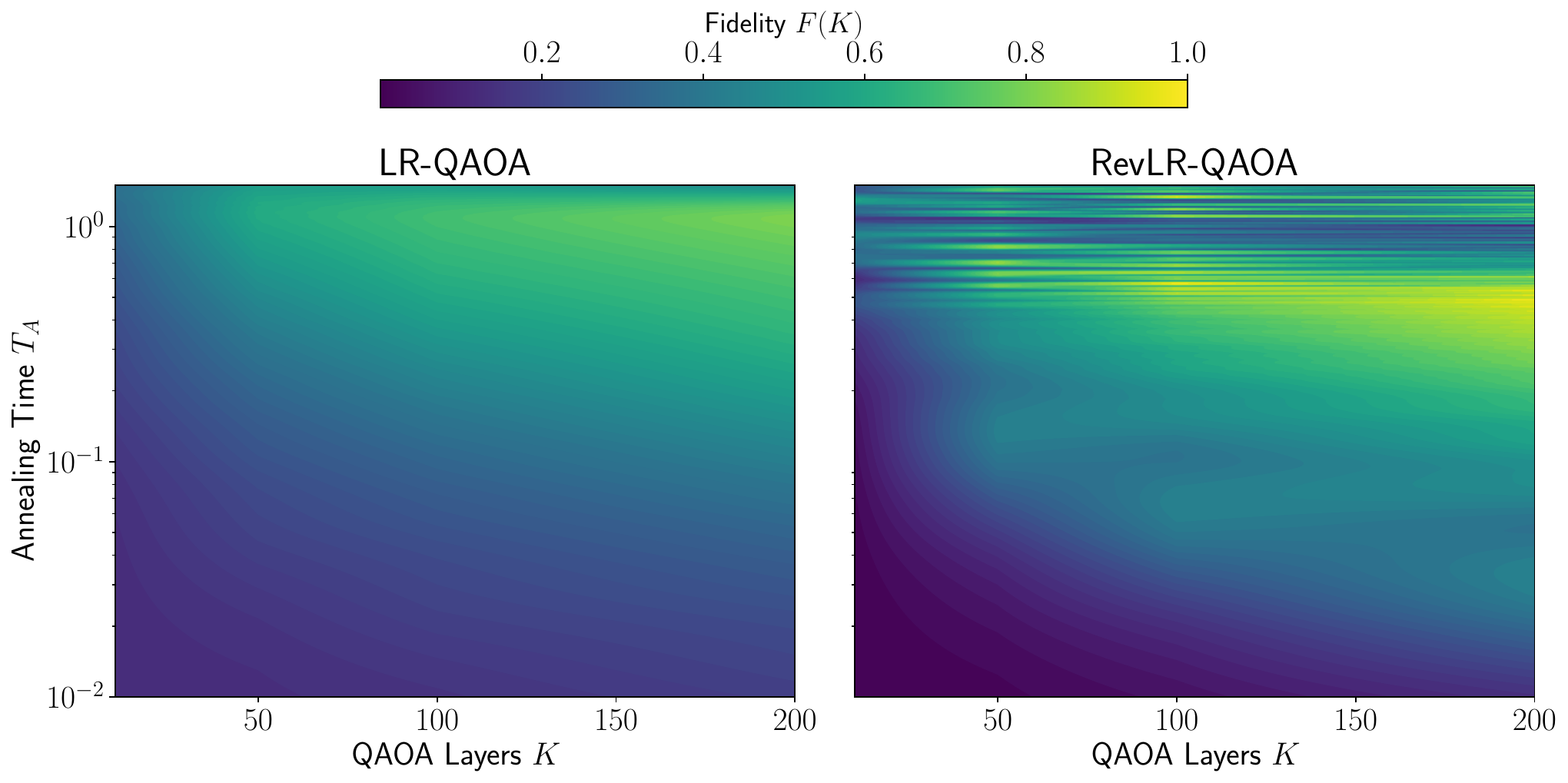}
    \caption{Comparison of fidelity $F(K)$ landscapes between LR-QAOA (left) and RevLR-QAOA (right). The landscape is based on a grid search of (hyper-)parameters $T_\mathrm{A}$ and $K$.}
    \label{sup:fig:fidelity_comp}
\end{figure}

\begin{figure}[t]
     \centering
    \includegraphics[width=\columnwidth]{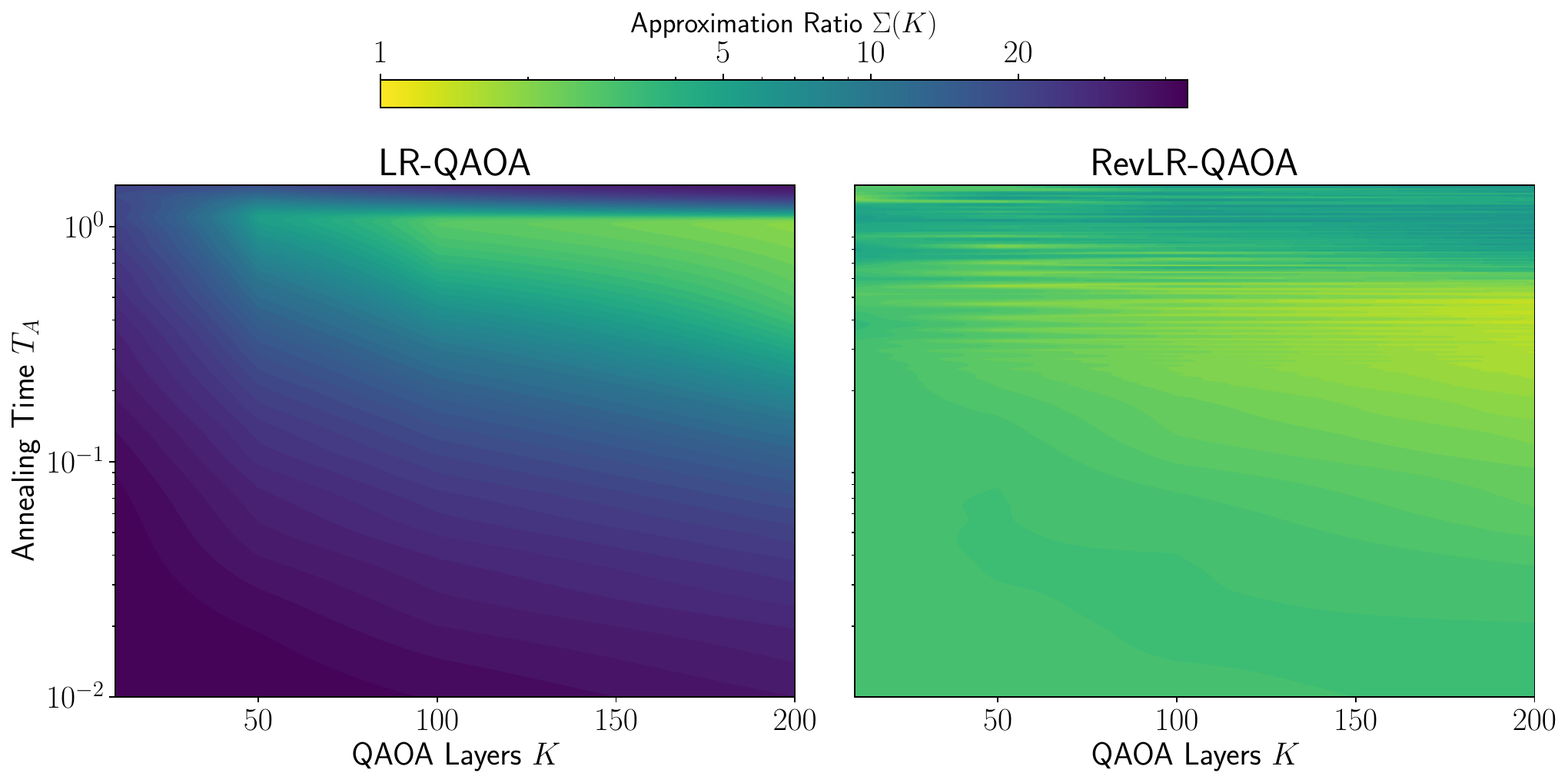}
    \caption{Comparison of approximation ratio $\Sigma(K)$ landscapes between LR-QAOA (left) and RevLR-QAOA (right). The landscape is based on a grid search of (hyper-)parameters $T_\mathrm{A}$ and $K$.}
    \label{sup:fig:approx_comp}
\end{figure}

For the small test grid, RevLR-QAOA with the feasiblity-preserving mixer consistently outperforms LR-QAOA with the standard mixer across a wide range of parameter configurations ($T_\mathrm{A}$, $K$), both in terms of fidelity $F(K)$ (cf.~Fig.~\ref{sup:fig:fidelity_comp}) and, in particular, approximation ratio $\Sigma(K)$ (cf.~Fig.~\ref{sup:fig:approx_comp}). This advantage is already visible for relatively small numbers of QAOA layers $K$.  Nevertheless, RevLR-QAOA exhibits a higher sensitivity to hyperparameters: small parameter changes can significantly reduce performance. By contrast, LR-QAOA performance improves systematically with increasing $T_\mathrm{A}$ and larger $K$. Importantly, performance does not increase monotonically with $T_\mathrm{A}$ in either approach. Beyond a certain $T_\mathrm{A}$ threshold, errors of order $\mathcal{O}(T_\mathrm{A} / K)$ accumulate, and the adiabatic evolution ceases to be well-approximated. At this point, RevLR-QAOA with the feasibility-preserving mixer effectively samples random feasible states, while LR-QAOA based on the penalty method samples random bit strings.  

The comparatively larger approximation errors observed for LR-QAOA with the standard mixer arise from infeasible outcomes, particularly at small $T_\mathrm{A}$ (or very large $T_\mathrm{A}$), where the state remains close to a uniform superposition dominated by high-cost, infeasible configurations. In contrast, the approximation error for RevLR-QAOA with the feasibility-preserving mixer is bounded above by the error of the most costly feasible state. In the penalty-based method, $\lambda_{\text{pen}}$ is typically chosen to create a spectral gap between feasible and infeasible states. As a consequence, the approximation errors of feasible states are significantly smaller than those of infeasible ones. 

simulations. For LR-QAOA, significant improvements in the approximation error occur mainly in the final stages of the evolution (last $\sim$25\%), when the cost function becomes dominant and provides stronger guidance to the mixing dynamics, cf.  Fig.~\ref{sup:fig:qoao_statitics}\textbf{a}. In contrast, for RevLR-QAOA, the main performance gain occurs immediately after the cost functions are swapped at $K/2$. Interestingly, we observe a pronounced initial drop in the approximation error for most values of $T_\mathrm{A}$; however, for some cases this is followed by stronger oscillations and even a rebound to higher approximation errors, cf. Fig.~\ref{sup:fig:qoao_statitics}\textbf{b}.

A statistical analysis reveals that these rebounds occur when the initial drop is strongest, which corresponds to larger values of $\beta_{K/2+1} \propto T_\mathrm{A}$, cf. Fig.~\ref{sup:fig:qoao_statitics}\textbf{d}. This is consistent with the observation that larger $T_\mathrm{A}$ generally leads to worse overall performance, as errors of order $\mathcal{O}(T_\mathrm{A}/K)$ accumulate. Since an initial drop is almost always observed, this suggests the potential of a modified RevLR-QAOA schedule with only a few layers after the cost function swap. However, caution is required, as this behavior may be an artifact of the small test system.

\begin{figure}[t]
     \centering
    \includegraphics[width=\columnwidth]{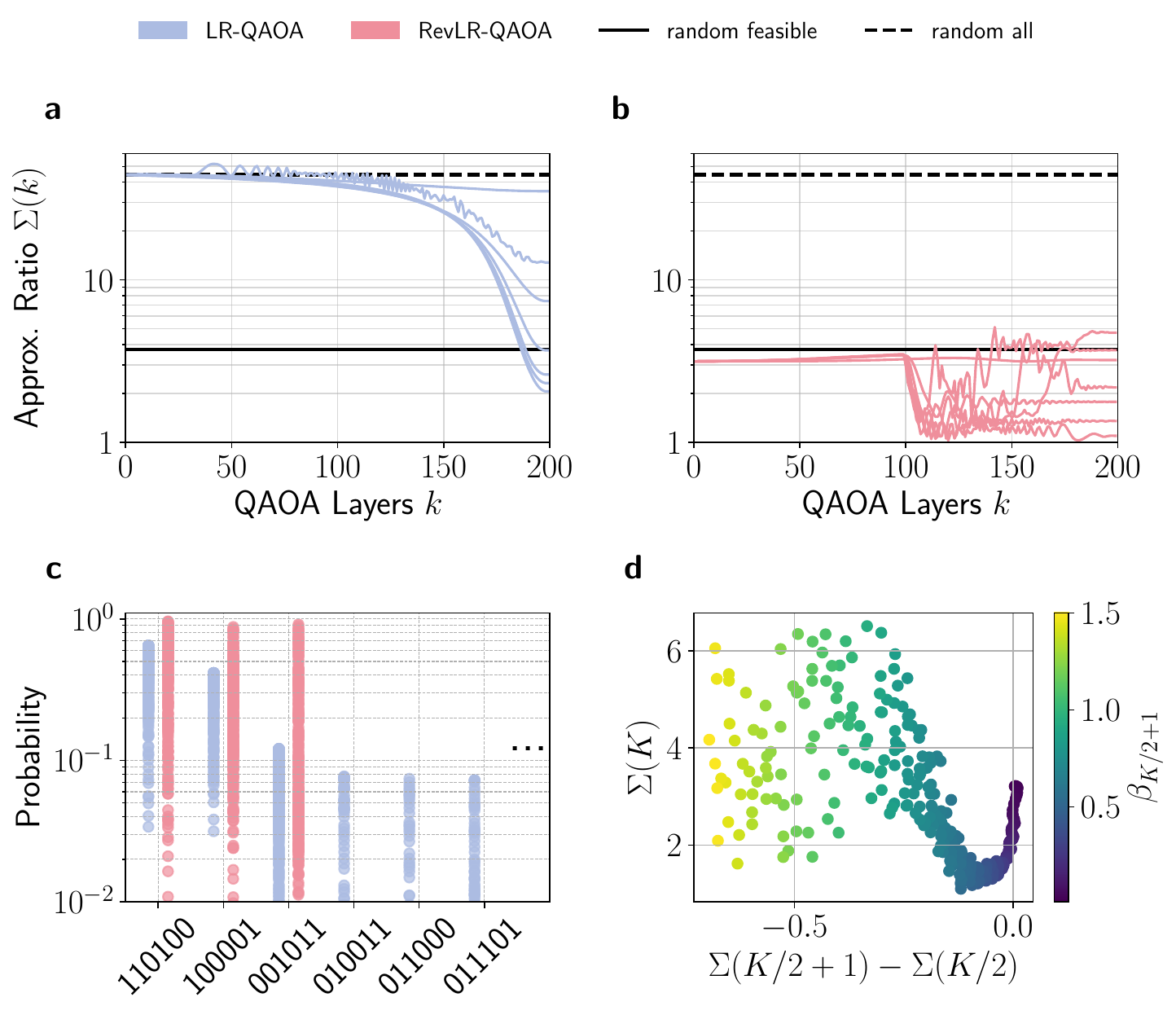}
    \caption{Statistical analysis of QAOA results for fixed $K=200$ across the interval $T_\mathrm{A} \in [0.01, 1.5 ]$.
    \textbf{a and b} Approximation error $\Sigma(k)$ over the layers $k$ for selected values of $T_\mathrm{A}$. For LR-QAOA, the improvement happens mainly towards the end. For RevLR-QAOA, the improvement occurs mainly around $ K/2 = 100$, when the cost function of the problem to be solved is introduced. \textbf{c} Compariosn of measurement statistics for all $T_\mathrm{A}$. \textbf{d} Dependence of the final approx ratio $\Sigma(K)$ in RevLR-QAOA on the initial improvement $\Sigma(K/2 + 1)-\Sigma(K/2)$ in the approx ratio after the problem cost function is introduced. Stronger initial improvements, due to stronger values of $\beta_{K/2 + 1}$, lead to worse performance. }
    \label{sup:fig:qoao_statitics}
\end{figure}

\clearpage
\section{Explicit Construction of the MDST+ Cost Function for Minimum Loss Network Reconfiguration}
\label{sec:mapping_distribution_grids}

In this section, we show how the Minimal Loss Network Reconfiguration problem can be solved on a reduced graph obtained by contracting nodes between switches, such that all remaining edges are switchable. The resulting cost function retains the structure of the standard MDST cost function but incorporates additional features; we refer to it as an MDST+ cost function.

Local distribution grids are normally operated radially, that is, each consumer is connected to the feeder (which is the connection to the transmission grid) by a unique path. However, for operating reasons, local distribution grids usually have a couple of switches that can be opened or closed to allow rerouting of the power flows in case of a fault or to reduce the losses and, thus, operating costs. In the latter case, one wants to find the configuration of closed/open switches to minimize the total losses. 

Let $\GG_\mathrm{grid}$ represent the topology of a distribution grid where nodes represent buses and thus connections to consumers or distributed energy resources (DER). The edges represent electrical cables (or transformers) or electrical switches. In other words, $\GG_\mathrm{grid}$ corresponds to a grid where all switches are closed. The flow injections are naturally the electrical current injections $I_n$. Again, if $I_n > 0$ bus $n$ demands current, whereas for $I_n<0$ current is injected at bus $n$. The flows on the lines are current on the lines $i_e$. The energy dissipation due to ohmic losses on each branch gives the cost of a given topology. For each line $e \in \GG_\mathrm{grid}$, the dissipation is given by 
\begin{equation*}
    l_{e} = R_{e} i_{e}^2,
\end{equation*}
where $R_{e}$ is the electric resistance of the line $e$. Without loss of generality, we can assume that (closed) switches have no resistance. Otherwise, we can replace the non-ideal switch by a lossless switch and a resistor in series.

Any valid configuration of switches corresponds to a spanning tree of $\GG_\mathrm{grid}$. The opposite is not true since only a few edges are switchable. Hence, the re-configuration of the distribution grid can not be tackled directly by spanning tree re-configuration. However, we can define a reduced graph $\GG_\mathrm{red}$ by contracting all nodes between switches to one ``super-node'' $v$, that is, every node $v \in \VV_\mathrm{red}$ corresponds to a sub-tree in $\GG_\mathrm{grid}$ that can not be reconfigured. By construction, all edges $s \in \GG_\mathrm{red}$ correspond to switches in the electrical grid \footnote{We note that the reduced graph $\GG_\mathrm{red}$ has multi-edges, that is, multiple edges with the same tail and head, if the underlying distribution grid has multiple switches between the same two super nodes, that is, between two un-reconfigurable sets of buses. Besides the need for additional bookkeeping, an edge is not uniquely defined by the incident nodes; all results obtained in this paper are still applicable. A simple example is the IEEE 123-node test feeder.}. Then, we have a one-to-one correspondence between valid switching configurations and spanning trees $\TT$ in $\GG_\mathrm{red}$, see Fig.~1 in the main manuscript for a schematic example. Based on this one-to-one correspondence, we can minimize the loss using local tree re-configurations on the reduced graph $\GG_\mathrm{red}$ while evaluating the losses for the currents $i_e$ on the full grid $\GG_\mathrm{grid}$. The cost evaluation can thus be decoupled into two steps: First, evaluate the switch currents $f_s (\TT)$ for a spanning tree $\TT$ in $\GG_\mathrm{red}$. Second, use the KCL (cf. Eq.~\eqref{net_rec:eq:KCL}) to calculate the currents $i_e$ in the electrical grid based on the switch currents $f_s (\TT)$. We explain both steps using the example from Fig.~1 in the main manuscript.  

To calculate the switch flows on $\GG_\mathrm{red}$, we define current demands/injections for the ``super-nodes'' $v \in \GG_\mathrm{grid}$ by summing over all current demands/injections for the contracted nodes, 
\begin{equation*}
    \mathfrak{f}_{v} = \sum_{n \in v} I_n. 
\end{equation*}
Then, in the spirit of Eq.~\eqref{net_rec:eq:flows_in_tree}, the flow on the switch $s$ for a given spanning tree $\TT$ in $\GG_\mathrm{red}$ is given by 
\begin{equation}
    \label{net_rec:eq:switch_currents}
    f_s (\TT) = \sum_{v \in \VV_{red}} y_{s,v} \mathfrak{f}_v,
\end{equation}
where the binary variables $y_{s,v}$ encode the tree $\TT$, cf. ~Eq.~\eqref{sup:eq:binary_variables}. By construction, if a switch $s \notin \TT$ we have that $f_s (\TT)=0$. Furthermore, the spanning tree $\TT$ together with the root ``super-node'' $v_0$, containing the feeder, induces a natural orientation for all closed switches: the head of the switch points downwards, away from the feeder bus. This orientation complies with the sign of the switch currents: 
If $f_s(\TT) > 0$, current flows downwards on the switch $s$ to meet the downward demand, and vice versa if $f_s(\TT) < 0$, current flows upwards since the injections downwards of the switch $s$ exceed the local downwards demands. 

The switch flows $f_s(\TT)$ and the current injections $I_n$ in $\GG_\mathrm{grid}$ uniquely define the currents $i_e$ on all other lines $e \in \EE_\mathrm{grid}$ by the KCL, which can be solved for each ``super-node'' independently. To use the KCL, we need to construct an edge-incidence matrix $\matr E(\TT)$ of $\GG_\mathrm{grid}$ that admits the orientation induced by $\TT$ on the switches $s$. Let $\matr E$ be the edge-incidence matrix for any orientation in $\EE_\mathrm{grid}$ and let $v(n) \in \GG_\mathrm{red}$ be the super-node containing the node $n \in \GG_\mathrm{grid}$, then for any switch $s = (n,m)$ the entries $\matr E(\TT)$ can be constructed using the binary variables $y_{s,v}$ as
\begin{equation}
    \label{net_rec:eq:orientation_induced_e_matrix}
    E(\TT)_{n,s} = E_{n,s} (E_{n,s} y_{s, v(n)} + \sum_{m \neq n} E_{m, s} y_{s, v(m)}).
\end{equation}
Note that for all lines $e$ in the super nodes, the sign of $i_e$ is irrelevant for the loss function and thus any orientation within the ``super-node'' can be arbitrary. 

Using the KCL, we can then solve for the line currents $i_e$ straightforwardly. Since the ``super-nodes'' correspond to trees, the resulting system of equations is always over-determined. We have $\lvert v \rvert$ equations for the nodal injections $I_n$ and one additional equation since the KCL must also be fulfilled for the whole ``super-node''. On the other hand, there are only $\lvert v \rvert - 1$ unknown branch currents $i_e$ in $v$. We now demonstrate solving the KCL for the branch currents for the ``super-node'' $v_2$ from the example grid, cf.~Fig.~1 in the main manuscript. Ignoring that for the depicted configuration $f_{s_4}(\TT) = 0$, the general KCL for $v_2$ reads 
\begin{align*}
    I_1 &= E(\TT)_{1, s_3} f_{s_3}(\TT) + E_{1,1} i_1 \\
    I_2 &= E(\TT)_{2, s_2} f_{s_2}(\TT) + E_{2,2} i_2 \\
    I_3 &= E_{3,1} i_1 + E_{3,2} i_2 + E_{3,3} i_3 \\
    I_4 &= E_{4,3} i_3 + E_{4,4} i_4 \\
    I_5 &= E(\TT)_{5, s_4} f_{s_4}(\TT) + E_{5,4} i_4 \\
    I_1 + ... + I_5 &= E(\TT)_{1, s_3} f_{s_3}(\TT) + E(\TT)_{2, s_2} f_{s_2}(\TT) \\ & \quad + E(\TT)_{5, s_4} f_{s_4}(\TT). 
\end{align*}
Hence, one solution for the branch currents is given by
\begin{align*}
    i_1 &= E_{1,1} \left(I_1 - E(\TT)_{1, s_3} f_{s_3}(\TT) \right) \\ 
    i_2 &= E_{2,2} \left(I_2 - E(\TT)_{2, s_2} f_{s_2}(\TT) \right) \\
    i_3 &= E_{4,3} \left(I_4 - E_{4,4} E_{5,4} \left(I_5 - E(\TT)_{1, s_3} f_{s_3}(\TT) \right) \right)\\
    i_4 &= E_{5,4} \left(I_5 - E(\TT)_{1, s_3} f_{s_3}(\TT) \right).
\end{align*}
Explicitly solving the KCL for all ``super-nodes'' and inserting Eq.~\eqref{net_rec:eq:switch_currents}~and~\eqref{net_rec:eq:orientation_induced_e_matrix} we get closed expressions for all branch currents $i_e$ that are quadratic in the binary variables $y_{s,v}$. 

We conclude that the MDST+ cost function 
\begin{equation*}
    C = \sum_e R_e i_e^2
\end{equation*}
can be expressed in closed form in terms of the binary variables $y_{e,n}$. This expression has to be constructed in a preprocessing step before the optimization is carried out by solving a linear system of equations. The resulting expression for the cost function is of 4th order in the binary variables $y_{e,n}$. We note that the cost function can be mapped to a 4th-order Ising Hamiltonian by the same steps as for the quadratic cost functions, cf. Eq.~\eqref{eq:Ising_hamiltonian}. However, simulating 4th order Hamiltonians requires more resources than the standard Ising Hamiltonian, see the discussion in \ref{sec:qaoa}.

\end{widetext}

\clearpage
\newpage
\bibliography{references}

\end{document}